\documentclass[10pt]{article} 
\usepackage{amsmath}
\usepackage{latexsym}
\usepackage{amssymb}

\usepackage{amsthm}

\font\tengoth=eufm10 at 10pt
\font\sevengoth=eufm7 at 6pt
\newfam\gothfam
\textfont\gothfam=\tengoth
\scriptfont\gothfam=\sevengoth

\newcommand{\mlabel}[1]{\marginpar{#1}\label{#1}}

\newcommand{\g}{{\mathfrak g}}

\newcommand{\fa}{{\mathfrak a}}
\newcommand{\fb}{{\mathfrak b}}

\newcommand{\fh}{{\mathfrak h}}

\newcommand{\fk}{{\mathfrak k}}

\newcommand{\fq}{{\mathfrak q}}
\newcommand{\fp}{{\mathfrak p}}

\newcommand{\fz}{{\mathfrak z}}

\renewcommand{\:}{\colon}
\newcommand{\1}{\mathbf{1}}
\newcommand{\0}{{\bf 0}}

\newcommand{\cA}{\mathcal{A}}

\newcommand{\cD}{\mathcal{D}}
\newcommand{\cE}{\mathcal{E}}

\newcommand{\cH}{\mathcal{H}}

\newcommand{\cL}{\mathcal{L}}

\newcommand{\cO}{\mathcal{O}}

\newcommand{\cS}{\mathcal{S}}
\newcommand{\cT}{\mathcal{T}}

\newcommand{\cW}{\mathcal{W}}

\newcommand\bx{{\bf{x}}}
\newcommand\by{{\bf{y}}}
\newcommand\bz{{\bf{z}}}

\newcommand\bq{{\bf{q}}}

\newcommand{\eset}{\emptyset}

\newcommand{\dd}{{\tt d}}

\newcommand{\subeq}{\subseteq}
\newcommand{\supeq}{\supseteq}

\newcommand{\into}{\hookrightarrow}
\newcommand{\eps}{\varepsilon}

\newcommand{\shalf}{{\textstyle{\frac{1}{2}}}}

\newcommand{\N}{{\mathbb N}}
\newcommand{\Z}{{\mathbb Z}}
\newcommand{\R}{{\mathbb R}}
\newcommand{\C}{{\mathbb C}}

\newcommand{\bE}{{\mathbb E}}
\newcommand{\bH}{{\mathbb H}}
\newcommand{\bS}{{\mathbb S}}

\renewcommand{\hat}{\widehat}

\renewcommand{\tilde}{\widetilde}

\newcommand{\GL}{\mathop{{\rm GL}}\nolimits}
\newcommand{\SL}{\mathop{{\rm SL}}\nolimits}

\newcommand{\SO}{\mathop{{\rm SO}}\nolimits}

\newcommand{\OO}{\mathop{\rm O{}}\nolimits}

\newcommand{\U}{\mathop{\rm U{}}\nolimits}

\newcommand{\Sym}{\mathop{{\rm Sym}}\nolimits}

\newcommand{\fsl} {\mathop{{\mathfrak{sl} }}\nolimits}

\newcommand{\so}  {\mathop{{\mathfrak{so} }}\nolimits}

\newcommand{\Exp}{\mathop{{\rm Exp}}\nolimits}
\newcommand{\Fix}{\mathop{{\rm Fix}}\nolimits}

\newcommand{\ad}{\mathop{{\rm ad}}\nolimits}
\newcommand{\Ad}{\mathop{{\rm Ad}}\nolimits}

\renewcommand{\Re}{\mathop{{\rm Re}}\nolimits}
\renewcommand{\Im}{\mathop{{\rm Im}}\nolimits}

\renewcommand{\>}{>\!\!>}

\newcommand{\Aut}{\mathop{{\rm Aut}}\nolimits}

\newcommand{\diag}{\mathop{{\rm diag}}\nolimits}
\newcommand{\End}{\mathop{{\rm End}}\nolimits}
\newcommand{\id}{\mathop{{\rm id}}\nolimits}

\newcommand{\supp}{\mathop{{\rm supp}}\nolimits}

\newcommand{\conv}{\mathop{{\rm conv}}\nolimits}

\newcommand{\Spann}{\mathop{{\rm span}}\nolimits}

\newcommand{\dS}{\mathop{{\rm dS}}\nolimits}

\newcommand{\Rarrow}{\Rightarrow}
\newcommand{\nin}{\noindent} 
\newcommand{\oline}{\overline}

\newcommand{\la}{\langle}
\newcommand{\ra}{\rangle}

\newcommand{\up}{\mathop{\uparrow}}

\newcommand{\res}{\vert}

\newcommand{\Spin}{{\rm Spin}}

\newcommand{\ssssarr}{\hbox to 15pt{\rightarrowfill}}
\newcommand{\sssarr}{\hbox to 20pt{\rightarrowfill}}
\newcommand{\ssarr}{\hbox to 30pt{\rightarrowfill}}
\newcommand{\sarr}{\hbox to 40pt{\rightarrowfill}}
\newcommand{\arr}{\hbox to 60pt{\rightarrowfill}}
\newcommand{\sssslarr}{\hbox to 15pt{\leftarrowfill}}
\newcommand{\ssslarr}{\hbox to 20pt{\leftarrowfill}}
\newcommand{\sslarr}{\hbox to 30pt{\leftarrowfill}}
\newcommand{\slarr}{\hbox to 40pt{\leftarrowfill}}
\newcommand{\larr}{\hbox to 60pt{\leftarrowfill}}

\newcommand{\Arr}{\hbox to 80pt{\rightarrowfill}}

\newcommand{\mapright}[1]{\smash{\mathop{\arr}\limits^{#1}}}

\def\theoremname{Theorem}
\def\propositionname{Proposition}
\def\corollaryname{Corollary}
\def\lemmaname{Lemma}
\def\remarkname{Remark}
\def\conjecturename{Conjecture} 

\def\definitionname{Definition}
\def\exercisename{Exercise}
\def\examplename{Example}
\def\examplesname{Examples}
\def\problemname{Problem}
\def\problemsname{Problems}

\def\bemerkname{Bemerkung}
\def\aufgname{Aufgabe}

\def\@thmcounter#1{\noexpand\arabic{#1}}
\def\@thmcountersep{}
\def\@begintheorem#1#2{\it \trivlist \item[\hskip 
\labelsep{\bf #1\ #2.\quad}]}
\def\@opargbegintheorem#1#2#3{\it \trivlist
      \item[\hskip \labelsep{\bf #1\ #2.\quad{\rm #3}}]}
\makeatother
\newtheorem{theor}{\theoremname}[section]
\newtheorem{propo}[theor]{\propositionname}
\newtheorem{coro}[theor]{\corollaryname}
\newtheorem{lemm}[theor]{\lemmaname}

\newenvironment{thm}{\begin{theor}\it}{\end{theor}}
\newenvironment{theorem}{\begin{theor}\it}{\end{theor}}

\newenvironment{prop}{\begin{propo}\it}{\end{propo}}

\newenvironment{cor}{\begin{coro}\it}{\end{coro}}
\newenvironment{corollary}{\begin{coro}\it}{\end{coro}}

\newenvironment{lem}{\begin{lemm}\it}{\end{lemm}}
\newenvironment{lemma}{\begin{lemm}\it}{\end{lemm}}

\newtheorem{rema}[theor]{\remarkname}

\newenvironment{remark}{\begin{rema}\rm}{\end{rema}}
\newenvironment{rem}{\begin{rema}\rm}{\end{rema}}

\newtheorem{stepnow}[theor]{}

\newtheorem{defin}[theor]{\definitionname} 

\newtheorem{exerc}{\exercisename}[section]

\newtheorem{exa}[theor]{\examplename}

\newenvironment{ex}{\begin{exa}\rm}{\end{exa}}

\newtheorem{exas}[theor]{\examplesname}

\newenvironment{exs}{\begin{exas}\rm}{\end{exas}}

\newtheorem{conj}[theor]{\conjecturename}

\newtheorem{pro}[theor]{\problemname}

\newtheorem{prs}[theor]{\problemsname}

\newenvironment{prf}{\begin{proof}}{\end{proof}}
 
\newcommand{\pmat}[1]{\begin{pmatrix} #1 \end{pmatrix}}

\usepackage{hyperref}
\numberwithin{equation}{section}

\newcommand{\cref}{\ref}
 \newcommand{\oQ}{\overline{Q}}

\usepackage{color}

\newcommand\be{{\bf{e}}}

\newcommand\bw{{\bf{w}}}

\newcommand\cut{{\rm{cut}}}

\renewcommand\up{{\uparrow}}

\newcommand{\sH}{{\sf H}}

\newcommand{\sV}{{\tt V}}
\newcommand{\sE}{{\tt E}}

\renewcommand{\phi}{\varphi}

\newcommand{\hgf}{{}_2F_1}
\newcommand{\hgfabc}{{}_2F_1(\alpha,\beta;\gamma;z)}

\newcommand{\drho}{{\textstyle\frac{d-1}{2}}}

\newcommand{\WdS}{W_{\dS^d}^+}

\renewcommand\mlabel{\label} 

\begin{document}

\title{Realization of unitary representations\\  of the  Lorentz group
  on de Sitter space} 
\author{Jan Frahm, Karl-Hermann Neeb, Gestur \'Olafsson}

\maketitle

\begin{center}
 {\bf Dedicated to the memory of Gerrit van Dijk} \\[5mm]
\end{center}

\begin{abstract}
{This paper builds on our previous work in which we showed
      that, for all connected semisimple linear Lie groups $G$
      acting on a non-compactly causal symmetric space $M = G/H$,
      every irreducible unitary representation of $G$ can be realized
      by boundary value maps of holomorphic extensions
      in distributional sections  of a vector bundle over $M$. In the present paper
      we discuss this procedure for the connected Lorentz group
      $G = \SO_{1,d}(\R)_e$ acting on de Sitter space $M = \dS^d$. 
      We show in particular that the previously constructed
      nets of real subspaces satisfy the locality condition.
      Following ideas of Bros and Moschella from the 1990's, we
      show that the matrix-valued spherical function that corresponds to
      our extension process extends analytically to a large domain
      $G_\C^{\rm cut}$ in the complexified group $G_\C = \SO_{1,d}(\C)$,
      which for $d = 1$ specializes to the complex cut plane
      $\C \setminus (-\infty, 0]$. A number of special situations
      is discussed specifically: (a) The case $d = 1$, which closely
      corresponds to standard subspaces in Hilbert spaces,
      (b) the case of scalar-valued functions, which for $d > 2$ is
      the case of spherical representations, for which we also describe
      the jump singularities of the holomorphic extensions on the
      cut in de Sitter space, (c) the case $d = 3$, where
      we obtain rather explicit formulas for the matrix-valued
      spherical functions.
      }
\end{abstract}

\tableofcontents 

\section{Introduction} 
\mlabel{sec:1}

{\bf Unitary representations.}
To explain the context of this work, we first recall some basic
structures related to (continuous) unitary representations
$U \: G \to \U(\cH)$ of a Lie group $G$ with Lie algebra~$\g$
on a complex Hilbert space $\cH$, mostly denoted~$(U,\cH)$. We
write
\[ U^v \: G \to \cH, \quad g \mapsto U(g) v \]
for its continuous orbit maps, 
$U^{v,w}(g) := \la v, U(g)w\ra$ for its matrix coefficients,
and $\cH^\infty \subeq \cH$ for the subspace of vectors with smooth orbit
maps (the smooth vectors). The latter subspace is dense
and carries a natural Fr\'echet topology 
(\cite{Go69, Ne10}). 
The space $\cH^{-\infty}$ of continuous antilinear functionals
$\cH^\infty \to \C$ ({\it distribution vectors}) 
contains in particular the functions 
$\la \cdot, v \ra$, $v \in \cH$.
Here we follow the convention common in physics
that $\la \cdot, \cdot \ra$ is linear in the second argument. 
We thus obtain complex linear embeddings 
$\cH^\infty \into \cH \into \cH^{-\infty}$. The group $G$ acts on all three spaces 
by representations denoted $U^\infty, U$ and $U^{-\infty}$, respectively.
All these representations can be integrated to the
convolution algebra $C^\infty_c(G) := C^\infty_c(G,\C)$ of
test functions, for instance
\[ U^{-\infty}(\phi) := \int_G \phi(g)U^{-\infty}(g)\, dg,\] 
where $dg$ stands for a left Haar measure on $G$. 
We write $\partial U(x)$ for the infinitesimal generator of the one-parameter
group $U(\exp t x) = e^{t \partial U(x)}$, $t \in \R$. 

{\bf Analytic continuation of orbit maps and matrix coefficients.}  
Let us assume, for simplicity, that $G$ is contained in its universal
complexification $G \into G_\C$. If an orbit map $U^v$ is analytic,
it extends to a holomorphic map $\hat U^v$
on a $G$-left invariant domain in $G_\C$, and likewise 
the matrix coefficients $U^{v,w}$ extend to domains in $G_\C$,
but these are typically not $G$-invariant.
The situation is most transparent for $G = \R$, corresponding to
unitary one-parameter groups $(U_t)_{t \in \R}$.
Then there exists a unique Borel spectral
measure $P$ on $\R$ with $U_t = \int_\R e^{itp}\, dP(p)$ and 
each $v \in \cH$ specifies by
$P^{v} := \la v, P(\cdot) v \ra$ a finite positive measure on $\R$, for which
$U^{v,v}(t) = \int_\R e^{itp} \, dP^v(p)$ is the Fourier transform
of $P^{v}$. In terms of $P^v$, the maximal domain
for analytic continuation of the orbit map $U^v$ is the strip 
\[ \cD_v := \Big\{ z = x + i y\in \C \: \int_\R e^{-2 y p}\, dP^v(p) < \infty
  \Big\},\]
and the maximal domain for the extension of the matrix coefficient $U^{v,v}$
is 
\[ \cD_{v,v} := 2 \cD_v
  = \Big\{ z = x + i y\in \C \: \int_\R e^{- yp}\, dP^v(p) < \infty \Big\}
  \quad \mbox{ with } \quad
  \hat U^{v,v}(z) = \int_\R e^{izp}\, dP^v(p).\]
The situations for which the upper and lower boundary values of $\hat U^{v,v}$
exist as tempered distributions
are of particular interest (see Section~\ref{sec:d=1} for a detailed
discussion).

{\bf Specializing to the Lorentz group.}
In this paper we shall study the corresponding problems
for the connected Lorentz group
$G = \SO_{1,d}(\R)_e$, which for $d = 1$, specializes to $G  = \SO_{1,1}(\R)_e \cong
\R$. Our starting point are the following results from \cite{FNO23}:
Consider the involution 
\[ \tau_h(x_0,\cdots, x_d) = (-x_0, -x_1,x_2, \ldots, x_d) \]
in $\SO_{1,d}(\R)$ and let
$(U,\cH)$ be an antiunitary representation of the extended group
$G_{\tau_h} := \SO_{1,d}(\R) \cong  G \rtimes \{\1, \tau_h\}$,
which is a finite sum of irreducible ones. Fix a finite-dimensional 
subspace $\cE \subeq \cH$ that is invariant under the maximal
compact subgroup $K := \SO_d(\R)$ and the conjugation $J := U(\tau_h)$. 
The  Kr\"otz--Stanton Extension Theorem (\cite[Thm.~3.1]{KSt04})
implies that the orbit map $U^v$ extends to a domain 
$\Xi_{G_\C} \subeq G_\C$ independent of $v$, and one of the main results
in \cite{FNO23} is the observation that, for the generator 
\[ h.(x_0,x_1,\ldots, x_d) = (x_1, x_0, 0,\ldots, 0) \]
 of the Lorentz boost in the $(x_0,x_1)$-plane, the limits
 \[  \beta^\pm(v)
   := \lim_{t \to \mp\pi/2} \hat U^v(\exp(ith)) 
   = \lim_{t \to \mp\pi/2} e^{it \partial U(h)}v \] 
exist in the space $\cH^{-\infty}$ of distribution vectors 
for the weak-$*$-topology (cf.\ Subsection~\ref{subsec:2.6}). The real subspace 
\begin{equation}
  \label{eq:EHpm}
  \sE_H := \beta^+(\cE^J) \subeq \cH^{-\infty}
\end{equation}
is finite-dimensional and invariant under the action of the subgroup
$H := G_{\be_1} \cong \SO_{1,d-1}(\R)_e$, so that it defines a real vector bundle
\[ 
\bE_H := G \times_H \sE_H \to M := G/H \cong \dS^d,
\]
where $\dS^d = \{ (x_0,\bx) \in \R^{1,d} \: x_0^2 - \bx^2 = - 1\}$ is
$d$-dimensional de Sitter space, realized as a Lorentzian
submanifold of the Minkowski space~$\R^{1,d}$. If $(U,\cH)$ is irreducible,
this leads to a natural realization of the representation
in the space $C^{-\infty}(G/H, \bE_H)$ of distributional
sections of $\bE_H$ (see Subsection~\ref{subsec:real} for details).
It follows in particular that {\bf every}
irreducible antiunitary representation of $G$ has a
natural realization in a vector bundle over de Sitter space,
which is a Lorentzian symmetric space, hence a natural model of a
curved space-time in Quantum Field Theory (\cite{Wa94}, \cite{AA20},
\cite{BJM23},\cite{Br90, BM96, Mo95}).

The motivation for this work was to connected the general results
from \cite{FNO23} on analytic continuation of orbit maps and their
boundary values with Algebraic Quantum Field Theory (AQFT)
on de Sitter space and in particular with the work
\cite{BM96} of J.~Bros and U.~Moschella
and the causal kernels studied by J.~Faraut, J.~Hilgert 
and the third author in \cite{FHO95}.

{\bf The case $d = 1$} corresponding to antiunitary representations of 
  $\SO_{1,1}(\R) \cong \R^\times \cong \R \rtimes \{\pm 1 \}$ is discussed in
  some detail in Section~\ref{sec:d=1}, where we specify,
  without restrictions on $U$,  a subspace
  $\cH^J_{\rm temp}\subeq \cH^J$ of those vectors on which $\beta^+$ induces a
  bijection onto the space $\cH^{-\infty}_{\rm KMS}$
  of those distribution vectors $\eta$ satisfying
  the {\it KMS (Kubo--Martin--Schwinger) condition} that the orbit map
$U^\eta(t) := U^{-\infty}(t)\eta$ has a continuous extension
$\hat U^\eta \colon \oline{\cS_\pi} \to \cH^{-\infty}$
to the closed strip $\oline{\cS_\pi}$,
which is holomorphic on the interior and satisfies 
$\hat U^\eta(\pi i) = J \eta$. 
  On the geometric side, the case
  $d = 1$ is particularly interesting because
  $\dS^1_\C \cong \C^\times$, where $\dS^1 \cong \R_+$,
  $\bH^1 \cong i \R_+$ and its crown domain is $\Xi_+ = \C_+$.
  Therefore the matrix coefficient $U^{v,v}$ of $v \in \cH^J_{\rm temp}$
 is a holomorphic function on the cut plane $\C \setminus [0,\infty)$
  with two sided tempered boundary values on $\R_+$, defining a
  ``jump singularity''. Although somewhat degenerate, this case displays
  already several key features of the general situations, where we also
  study boundary values on de Sitter space from two domains $\Xi_\pm$.
  For a detailed discussion of antiunitary representations
  of $\R \times \{\pm 1\}$ from the perspective of KMS conditions, we refer to
  \cite{NO19}.

  {\bf Standard subspaces.} It is a central fact 
  that the antiunitary representations of $\SO_{1,1}(\R) \cong \R \times
  \{ \pm 1\}$ are in  one-to-one correspondence
  with {\it standard subspaces} $\sV \subeq \cH$.
  These are closed real subspaces for which $\sV + i \sV$ is dense in $\cH$
  ($\sV$ is {\it cyclic}) 
  and $\sV \cap i \sV = \{0\}$ ($\sV$ is {\it separating}). Then
  the complex conjugation $T_\sV$ on $\sV + i \sV$ has a polar decomposition
$T_\sV = J_\sV \Delta_\sV^{1/2}$, where
$J_\sV$ is a conjugation (an antilinear involutive isometry)
and $\Delta_\sV$ is a positive selfadjoint operator
satisfying $J_\sV \Delta_\sV J_\sV = \Delta_\sV^{-1}$.
The unitary one-parameter group $U_t := \Delta_\sV^{it}$ 
(the {\it modular group of $\sV$}) preserves the subspace $\sV$
and commutes with $J_\sV$, which establishes the correspondence with
antiunitary representations of $\R\times \{\pm 1\}$. 
We refer to \cite{Lo08, NO17} for more on standard subspaces.

{\bf Nets of real subspaces.} The notion of a standard subspace
brings us to nets of real subspaces, a fundamental
structure in AQFT.
In particular, they play a {central} role in the recent study of
entropy and energy inequalities (see  \cite{MTW22,Lo20, CLRR22, CLR20} and references therein),
new constructions in AQFT (\cite{MN22,LL15,LMPR19,MMTS21}),
and in a very large family of models (see references in \cite{DM20}). 
A family $\sH(\cO)$ of closed real subspaces of $\cH$, indexed by open subsets
$\cO \subeq M = \dS^d$ is called a {\it net of real subspaces}.
For such a net we consider the following properties
(with respect to a unitary representation $(U,\cH)$ of $G$): 
\begin{itemize}
\item[(Iso)] {\bf Isotony:} $\cO_1 \subeq \cO_2$ 
implies $\sH(\cO_1) \subeq \sH(\cO_2)$ 
\item[(Cov)] {\bf Covariance:} $U(g)\sH(\cO) = \sH(g\cO)$ for $g \in G$. 
\item[(RS)] {\bf Reeh--Schlieder property:} 
$\sH(\cO)$ is cyclic  if $\cO \not=\eset$. 
\item[(BW)] {\bf Bisognano--Wichmann property:}
  For the wedge region $\WdS  = \{ x \in \dS^d \: x_1 > |x_0|\}$, 
  the space $\sV := \sH(W_{\dS^d})$ is standard with
  $\Delta_\sV = e^{2\pi i \partial U(h)}$ and $J_\sV = U(\tau_h)$.
\item[(Loc)] {\bf Locality:} If $\cO_1$ and $\cO_2$ are spacelike
  separated, then $\sH(\cO_1) \subeq \sH(\cO_2)'$,
  where, for a real subspace $\sH\subeq\cH$, we write
  $\sH' = \{ w \in \cH \:  \Im \la \sH, w \ra =\{0\}\}$
  for its  symplectic orthogonal space.
\end{itemize}

The main result of \cite{FNO23} is the following:
For an open subset $\cO \subeq \dS^d$,
let $\sH^{dS^d}_{\sE_H}(\cO) \subeq \cH$ be the closed
real subspace generated by 
  \[ U^{-\infty}(\xi)\sE_H, \quad \xi \in C^\infty_c(q^{-1}(\cO),\R), \quad
    \mbox{ where } \quad q \:  G \to G/H \cong \dS^d, \quad g \mapsto gH,\]
  then we obtain (for every irreducible antiunitary representation of
  $\SO_{1,d}(\R)$) a net satisfying (Iso), (Cov), (RS) and (BW).
 One of the main results of this paper is that these nets also satisfy
  the locality condition~(Loc) (Section~\ref{sec:6}).
  This result is based on an inspection of the support properties
  of the distributions on $G$ obtained by the boundary values 
  of the matrix coefficients $U^{v,v}$, $v \in \cH^J_{\rm temp}$, on
  the submanifold
  \[ G \cong \exp\Big(\frac{\pi i}{2}h\Big) G \exp\Big(\frac{\pi i}{2}h\Big)
    \subeq G_\C.\]

  {\bf The generalized spherical function.}
  If $P_\cE \: \cH \to \cE$ is the orthogonal projection, then
  \[ \phi_\cE \: G \to B(\cE),\quad \phi_\cE(g) := P_\cE U(g) P_\cE \]
  is called the corresponding {\it generalized spherical function}. It satisfies
  \[ \phi_\cE(k_1 g k_2) = U_\cE(k_1) \phi_\cE(g) U_\cE(k_2) \quad \mbox{
      for } \quad g \in G, k_1, k_2 \in K = \SO_d(\R),\]
  where $U_\cE$ denotes the representation of $K$ on $\cE$.
  Inspired by \cite{BM96}, we describe a rather large  domain
  $G_\C^\cut \subeq G_\C$ to which $\phi_\cE$ extends as a holomorphic
  $B(\cE)$-valued function (as $\cE$ is finite-dimensional,
  this is a matrix-valued function). 
  The terminology $G_\C^\cut$ for the {\it cut domain} comes from the
  case $d = 1$, where $G_\C = \C^\times$ and $G_\C^\cut \cong \C \setminus (-\infty,0]$
  is a cut-plane. The existence of this holomorphic extension
  is our second main result  (Sections~\ref{sec:3} and \ref{sec:4}),
  and boundary values of $\phi_\cE$ 
  are studied in Subsection~\ref{sec:5}.

  {\bf Spherical functions.}
For the case of spherical representations, i.e., where
  $\cE^J = \R v$ for a $K$-fixed vector $v$, our bundle-valued distributions
  simplify to distributions on $\dS^d$. In this case we compare
  the two boundary values on $\dS^d$ from two sides and relate the
  difference of the boundary values to the spherical kernels
  studied in \cite{FHO95}. Here
  the estimates for the spherical functions 
that imply $v \in \cH^J_{\rm temp}$ for a $K$-fixed vector
also follow from the estimates of B.~Kr\"otz and E.~Opdam (\cite[Thm.~7.2]{KO08}). 

For $d = 3$, the Lorentz group acting on $4$-dimensional
  Minkowski space, we obtain explicit formulas for the
  generalized spherical function $\phi_\cE$ in terms of
  hypergeometric functions

The content of the paper is as follows.
In Section~\ref{sec:2} we introduce our 
notation and describe the interface with \cite{FNO23} in some detail.
Sections~\ref{sec:3} provides some geometric background that is used in
Section~\ref{sec:4} to obtain the analytic continuation of $\phi_\cE$ to
$G_\C^{\rm cut}$, and in Subsection~\ref{sec:5} we investigate the
boundary values of the corresponding matrix coefficients.
Locality of the net $\sH^{\dS^d}_{\sE_H}$ is proved in Section~\ref{sec:6}.
Section~\ref{sec:d=1} discusses the case $d = 1$ in some detail.
The case one-dimensional $K$-types is discussed
briefly in Section~\ref{sec:8}, and the important case of
spherical functions is treated in detail in Section~\ref{sec:sphfunc}.
In Section~\ref{sec:msphfunc} the matrix-valued spherical
functions are calculated explicitly for $d = 3$.
We conclude this paper with a short discussion of perspectives
in Section~\ref{sec:pers}.

{\bf Relation to the work of van Dijk.} A theme that makes an appearance in many of van Dijk's articles is analysis on non-Riemannian symmetric spaces of rank one. Together with some of his younger colleagues he was able to make the Plancherel formula for essentially all such spaces explicit. The Plancherel decomposition typically consists of a direct integral of (degenerate) unitary principal series representations and, possibly, an infinite direct sum of discrete series representations. One of the key players in their work are explicit Poisson transforms, providing realizations of the relevant unitary representations in spaces of functions/distributions on the symmetric spaces.\\
In contrast to the work of van Dijk, we focus on one particular family of non-Riemannian symmetric spaces of rank one, the de Sitter spaces. On the other hand, our construction of realizations of unitary representations on spaces of distributions on de Sitter spaces is not restricted to the scalar case, but also extends to distributional sections of vector bundles, thus making it possible to explicitly realize all irreducible unitary representations of the Lorentz group in this way. One of the main motivations for this work was to make this construction as explicit as possible, much in the spirit of van Dijk's explicit Plancherel formulas.

\section{Preliminaries and notations}
\mlabel{sec:2}

In this section we introduce our 
notation and describe the interface with \cite{FNO23} in some detail.

\subsection{Notation} 
\mlabel{subsec:2.1} 
On $V := \R^{1,d}$, $d \geq 1$, we consider the Lorentzian form 
\[ \beta(x,y) := x_0 y_0 - \bx \by = x_0y_0 - x_1y_1 - \cdots -x_dy_d,\]
and call $(V,\beta)$ Minkowski space.
The {\it future light cone} is 
\begin{equation}
  \label{eq:v+}
  V_+ := \{ x = (x_0, \bx) \in V \: x_0 > 0, x_0^2 - \bx^2 > 0\},
\end{equation}
and the two associated complex tube domains  are 
\begin{equation}
  \label{eq:tpm}
  \cT_\pm  = V \pm  i V_+ \subeq V_\C \cong \C^{1+d}.
\end{equation}

We write $G := \SO_{1,d}(\R)^\up = \SO_{1,d}(\R)_e$ for the connected Lorentz group 
acting on $(V,\beta)$ by linear isometries preserving the cone $V_+$,
representing the {\it time-orientation}. Accordingly,
$G$ is called the {\it orthochronous proper Lorentz group}.
The stabilizer $K = G_{\be_0} \cong \SO_d(\R)$ is a maximal compact subgroup
and $H := G_{\be_1} \cong \SO_{1,d-1}(\R)$.
We consider the 
{\it de Sitter space} 
\[ M := \dS^d := \{ x \in V \: \beta(x,x)= -1\} = G.\be_1
  \cong G/H\]
and endow it with the causal structure
specified by the $G$-invariant field of
open cones $V_+(m) := T_m(M) \cap V_+$. Note that
$V_+(g.\be_1) = g.V_+(\be_1)$ for $g \in G$. 
We also note that the Minkowski metric~$\beta$ 
restricts to a Lorentzian metric on $\dS^d$.
We realize the {\it hyperbolic space $\bH^d$} (the Riemannian symmetric space of
$G$) as
\[ \bH^d := \{ x \in i V \: \beta(x,x)= -1\} = G.i\be_0
  \cong G/K.\]
Sometimes we shall also write $\bH^d_\pm := G.(\pm i \be_0)$ to distinguish
the two copies of $\bH^d$ in $iV$.

Elements of $\C^{1+d}$ are denoted $z = (z_0,\bz)$ with $z_0 \in \C$ and
$\bz \in \C^d$. The form $\beta$ extends in a complex bilinear way
to $V_\C = \C^{1+d}$ by
\[ \beta(z,w) = z_0 w_0 - \bz \bw = z_0w_0 - z_1w_1 - \cdots - z_dw_d.\]
The complex linear isometry group of $(V_\C,\beta)$ is
\[ G_\C = \SO_{1,d}(\R)_\C \cong \SO_{1+d}(\C).\]
If $\oline z = (\oline{z_0}, \ldots, \oline{z_d})$, then
$\beta(z,\oline w)$ 
defines a sesquilinear form on $V_\C \times V_\C$, which is invariant under the
action of $G_\C$ by $g.(z,w) := (gz, \oline{g}w)$:
\begin{equation}\label{eq:beta-c-inv}
  \beta(gz, \oline{\oline{g}w}) = \beta(z,\oline w). 
\end{equation}

De Sitter space and hyperbolic space have a common ``complexification'' 
\[ \dS^d_\C := \{ z \in \C^{1+d} \: z_0^2 - \bz^2 = - 1\}
  = G_\C.i\be_0 = G_\C.\be_1\]
with
\[ \dS^d_\C \cap \R^{1,d} = \dS^d \quad \mbox{ and } \quad 
 \dS^d_\C \cap i\R^{1,d} = \bH^d_+\, \dot\cup\, \bH^d_-.\] 
It contains the two {\it crown domains}
\begin{equation}
  \label{eq:xipm}
  \Xi_\pm := \dS^d_\C \cap (V \pm i V_+) = \dS^d_\C \cap \cT_\pm,
\end{equation}
which can be identified with the crown domains of the 
Riemannian symmetric spaces $\bH^d_\pm$
(\cite[Thm.~6.5, App.~D]{NO23} and
\eqref{eq:crown} below). 
In \cite{BM96}, the domains $\Xi_\pm$ are called {\it forward and backward tube}.
Both domains contain de Sitter space $\dS^d$ in their boundary. 

\subsection{The modular flow}

The standard boost vector field $X_h(v) = hv$ is defined by
$h \in \so_{1,d}(\R)$, given by 
\begin{equation}
  \label{eq:h}
hx = (x_1, x_0, 0,\ldots, 0).
\end{equation}
It generates the flow 
\[ \alpha^h_t(x) = e^{th} x 
= (\cosh t \cdot x_0 + \sinh t \cdot x_1, 
\cosh t \cdot x_1 + \sinh t \cdot x_0, x_2, \ldots, x_d)\]
that extends to a holomorphic $2\pi i$-periodic flow on $\dS^d_\C$.

Note that $\R h \cong \so_{1,1}(\R)$, acting in the first two coordinates, 
and that $\exp(\C h) \cong \SO_{1,1}(\C)$,
where 
\begin{equation}
  \label{eq:so1c}
\SO_{1,1}(\C) = \Big\{\pmat{ a & b \\ b & a} \: a^2 - b^2 = 1, a,b
  \in \C \Big\} \to (\C^\times,\cdot), \quad \pmat{ a & b \\ b & a} \to a + b
\end{equation}
is a biholomorphic homomorphism mapping $\exp(zh)$ to $e^z$.
As the orbit map $\SO_{1,1}(\C) \to \dS^1_\C, g \mapsto g.\be_1$
is biholomorphic, we obtain a biholomorphic
map
\begin{equation}
  \label{eq:c-coord}
  \dS^1_\C \to \C^\times, \quad (z_0,z_1) \mapsto z_0 + z_1
  \quad \mbox{ with } \quad  \R_\pm \to \dS^1_\pm
  \quad \mbox{ and } \quad \C_\pm \to \Xi_\pm.
\end{equation}
So we consider $\dS^1_\C$ as obtained from gluing
the two domains $\Xi_\pm$ along their boundaries.

The  involution 
\begin{equation}
  \label{eq:tauh}
  \tau_h(x) := e^{\pi i h} x= (-x_0, -x_1, x_2, \ldots, x_d),
\end{equation}
extends to an antilinear involution
$\oline\tau_h$ on $V_\C$.
For $d = 1$, the holomorphic modular flow yields for the base point
$\be_1$ the complex exponential map
\begin{equation}
  \label{eq:sinh-cosh}
 z \mapsto \alpha^h_z(\be_1)
  = \sinh(z) \be_0 + \cosh(z) \be_1, \quad \C \to \dS^1_\C \cong
  \C^\times
\end{equation}
(cf.\ Section~\ref{sec:d=1}). 

We also consider the ``Wick rotation'' 
\begin{equation}
  \label{eq:kappah}
 \kappa_h(x) = e^{-\frac{\pi i}{2} h} x
  = (-i x_1, -i x_0, x_2, \ldots, x_d) \quad \mbox{ satisfying } \quad
  \kappa_h^2 = \tau_h \quad \mbox{ and } \quad \kappa_h(i\be_0) = \be_1.
\end{equation}

As $\dS^d$ is $\alpha^h$-invariant, 
the vector field $X_h$ is tangential to $\dS^d$, 
and \cite[Prop.~D.9]{NO23} implies that 
its {\it positivity domain} is given by the {\it wedge region} 
\begin{equation}
  \label{eq:wdef}
  \WdS  
:= \{ x \in \dS^d \:  X_h(x) \in V_+(m) \} 
= \{ x \in \dS^d \: x_1 > |x_0|\}.
\end{equation}
The centralizer of $h$ in $G$ is the connected subgroup 
\begin{equation} \label{eq:gh-lorentz}
G^h = \exp(\R h) \SO_{d-1}(\R) \cong 
\SO_{1,1}(\R)^\up \times \SO_{d-1}(\R) 
\end{equation}
(cf.\ \cite[Lemma~4.12]{NO17}). This subgroup leaves $\WdS$ invariant,
whereas  $\tau_h(\WdS ) = -\WdS$  is the positivity domain
of the vector field $-X_h = X_{-h}$. 

\subsection{Lie group data} 

We now collect some of the typical Lie group notation
that is needed to specialize the general results
from \cite{FNO23} to the case of the Lorentz group. 

On the group $G = \SO_{1,d}(\R)^\up$ with Lie algebra $\g = \so_{1,d}(\R)$,
we consider the Cartan involution, given
by $\theta^G(g) = (g^\top)^{-1}$, which induces on the Lie algebra
$\g$ the involution $\theta(x) = - x^\top$.
The corresponding subgroup of $\theta$-fixed points is
$K = G^\theta \cong \SO_d(\R)$.

As $\tau_h$ from \eqref{eq:tauh} is contained in $\SO_{1,d}(\R)$, we have an
isomorphism of Lie groups 
\begin{equation}
  \label{eq:gtauh}
 \SO_{1,d}(\R) \cong \SO_{1,d}(\R)^\up\rtimes \{\1,\tau_h\}
 \cong G_{\tau_h} = G \rtimes \{ \id,\tau_h^G\}.
\end{equation}

The boost generator $h$ from \eqref{eq:h}
spans a maximal abelian subspace
\[ \fa := \R h \subeq \fp := \{ x \in \g \: \theta(x) = - x\} \]
and the corresponding subgroups are 
\[ A = \exp(\fa) \cong \SO_{1,1}(\R)^\up \cong (\R_+,\cdot) 
  \subeq A_\C = \exp(\fa_\C) \cong \SO_{1,1}(\C) \cong \C^\times.\]
The element $h$ is an {\it Euler element}, i.e.,
\[ \g =  \g_1(h) \oplus \g_0(h) \oplus \g_{-1}(h)
  \quad \mbox{ for } \quad
  \g_\lambda(h) = \ker(\ad h - \lambda \1).\]
Note that $\g_{\pm 1}(h)$ are abelian Lie subalgebras of $\g$.

\subsection{Crown domains and tubes}

Here we list some notation connecting the general Lie theoretic
approach to the concrete situation for de Sitter space.
In the following we write:
\begin{itemize}
\item $\Omega_\fa = (-\pi/2,\pi/2)h$ and $\Omega_\fp = \Ad(K)\Omega_\fa
  = \{ x \in \fp \: \rho(\ad x) < \pi/2\}$ 
\item $\cS_{\pm \beta}= \{z \in \C \: |\Im z| < \beta\}$ and 
  $\cS_{\beta}= \{z \in \C \:  0 < \Im z < \beta\}$.   
\item $\Xi_{A_\C} = A \exp(i\Omega_\fa) = \exp(\fa + i \Omega_\fa)
  = \exp(\cS_{\pm \pi/2} h)$ 
\item  $\fa_\C^\cut := \fa + 2 i \Omega_\fa = \cS_{\pm \pi} h$,
  and   $A_\C^\cut := \exp(\fa_\C^\cut)$.
\end{itemize}

Identifying $A_\C \cong \SO_{1,1}(\C)$ with $\C^\times$
as in \eqref{eq:so1c}, 
the domain $A_\C^\cut$ corresponds to the cut plane
$\C \setminus (-\infty, 0] \subeq \C^\times$; hence the name.
We have
\begin{equation} \label{eq:crown}
  \Xi_\pm = G \exp(i\Omega_\fp).(\pm i\be_0 ) = \pm G\exp (i\Omega_\fp).i\be_0
\end{equation} 
(\cite[Thm.~6.5, App.~D]{NO23}), so that $\Xi_\pm$ can be identified
with the {\it crown domain} of the Riemannian
symmetric spaces $\bH^d_\pm \cong G/K$.

For the  orbit map $q_{M_\C} \: G_\C \to M_\C, g \mapsto g.i\be_0$, the domain
\begin{equation}
  \label{eq:xigc}
  \Xi_{G_\C} = G \exp(i\Omega_\fp)K_\C \subeq G_\C
\end{equation}
is the inverse image of $\Xi_+$ under~$q_{M_\C}$.
From $\Xi_+ = \dS^d_\C \cap \cT_{V_+}$ (cf.\ \eqref{eq:xipm}), we obtain 
\[ \Xi_{G_\C} = \{ g \in G_\C \: \Re(g.\be_0) \in V_+ \}.\]

\subsection{SL$_2$-notation}

In this subsection we collect some formulas concerning the 
$3$-dimensional Lie algebra $\fsl_2(\R)$ that we shall use below.  
For $\g  = \fsl_2(\R)$ ($\cong \so_{1,2}(\R)$), 
we fix the Cartan involution $\theta(x) = - x^\top$, so that 
\[ \fk = \so_2(\R) \quad \mbox{ and }\quad 
  \fp = \{ x \in \fsl_2(\R) \: x^\top = x \}.\]

The basis elements 
\begin{equation}
  \label{eq:sl2ra}
h^0 := \frac{1}{2} \pmat{1 & 0 \\ 0 & -1},\quad 
e^0 =  \pmat{0 & 1 \\ 0 & 0}, \quad
f^0 =  \pmat{0 & 0 \\ 1 & 0}
\end{equation}
and 
\begin{equation}
  \label{eq:sl2rb}
h^1 = \frac{1}{2} \pmat{0 & 1 \\ 1 & 0} = \frac{1}{2}(e^0 + f^0),\quad 
e^1 =  \frac{1}{2}\pmat{-1 & 1 \\ -1 & 1}, \quad 
f^1= \frac{1}{2} \pmat{-1 & -1 \\ 1 & 1}
\end{equation}
satisfy
\[  [h^j, e^j] = e^j, \quad [h^j, f^j] = -f^j, \quad 
[e^j, f^j] = 2 h^j \quad \mbox{ and }\quad \theta(e^j) = -f^j
\quad \mbox{ for } \quad j=1,2.\]
For
\[ k^0 := \frac{1}{2} \pmat{ 0 & -1 \\ 1 & 0} = \frac{1}{2}(f^1 - e^1),\]
we have $[h^1, k^0] = \frac{1}{2}(-f^1 - e^1) = h^0$, so that
\[ [k^0, h^0] = h^1, \quad [k^0, h^1] = -h^0 , \quad \mbox{ and } \quad
[h^1, h^0] = k^0.\]

We identify $\ad(\fsl_2(\R))$ with $\so_{1,2}(\R)$ via the adjoint
action on $\fsl_2(\R) \cong \R^{1,2}$, where the Cartan--Killing form 
of $\fsl_2(\R)$ corresponds to the Lorentzian form $\beta$ on $\R^{1,2}$ via 
\[  \be_0 = - k^0, \quad \be_1 = - h^1, \quad \be_2 = - h^0.\] 
 Then $h^0, h^1, k^0$ correspond to the following matrices in $\so_{1,2}(\R)$:
\begin{equation}\label{eq:handk2}
  h_0 := h = \pmat{ 0 & 1 & 0 \\
    1 & 0 & 0 \\
    0 & 0 & 0}, \quad 
  h_1 := \pmat{ 0 & 0 & 1 \\
    0 & 0  & 0 \\
    1 & 0 & 0}, \quad \mbox{ and }\quad
  k_0:= \pmat{ 0 & 0 & 0 \\
    0 & 0 & -1  \\    0 &  1 & 0}. 
 \end{equation} 

 \subsection{Distribution vectors by analytic continuation} 
 \mlabel{subsec:2.6}

 Let $(U,\cH)$ be a unitary representation of $G = \SO_{1,d}(\R)^\up$
that extends to an antiunitary representation of 
$G_{\tau_h} \cong \SO_{1,d}(\R)$
(cf.\ \eqref{eq:gtauh}). 
This means that $J = U(\tau_h)$
is a conjugation (an antiunitary
involution) with $J U(g) J =  U(\tau^G_h(g))$ for $g \in G$.

We write $\sV \subeq \cH$ for the standard subspace
$\sV = \Fix(J_\sV \Delta_\sV^{1/2})$, specified by 
\begin{equation} \label{eq:mod-obj} 
  \Delta_\sV = e^{2\pi i \cdot \partial U(h)} \quad \mbox{ and } \quad J_\sV = J,
\end{equation}
where $\partial U(h)$ is the skew-adjoint infinitesimal generator
of the one-parameter group $U(\exp th)$. 
Then $\xi\in \cH$ is contained in $\sV$ if and only if it satisfies
the {\it KMS like condition} that the orbit map
$U_h^\xi(t) := U(\exp th)\xi$ has a continuous extension
$U_h^\xi \colon \oline{\cS_\pi} \to \cH$
to the closed strip $\oline{\cS_\pi}$,
which is holomorphic on the interior and satisfies 
\begin{equation}\label{eq:UhJ}
 U_h^\xi(\pi i) = J \xi\quad \mbox{ resp.} \quad
  U_h^\xi(t + \pi i) = J U_h^\xi(t)\quad \mbox{ for }\quad
t \in \R\end{equation}
(cf.\ \cite[Prop.~2.1]{NOO21}). 

We write 
\begin{equation}
  \label{eq:kms-spaces-intro}
 \cH^{-\infty}_{\rm KMS} \subeq \cH^{-\infty} 
\end{equation}
for the space of distribution vectors for which 
the orbit map extends continuously, with respect to the
weak-$*$-topology, to the the closed strip $\overline{\cS}_\pi$, weak-$*$
holomorphic in the interior, such that
\begin{equation}\label{eq:UhJ2}
 U_h^\xi(\pi i) = J \xi\quad \mbox{ resp.} \quad
  U_h^\xi(t + \pi i) = J U_h^\xi(t)\quad \mbox{ for }\quad
  t \in \R \end{equation}
(see \cite{BN23} and \cite[\S 2.3]{FNO23} for more on this space).

We consider a finite-dimensional $K$-invariant subspace 
$\cE \subeq \cH$ which is also $J$-invariant and set
\[ \cE^J = \{ v \in \cE \colon Jv = v\}.\]
Note that $\tau_h(K) = K$ implies that $J$ leaves the subspace
$\cH^{[K]}$ of $K$-finite vectors invariant. Therefore $J$-invariant
finite-dimensional $K$-invariant subspaces exist in abundance.
For the case where $U\res_G$ is a {\bf finite sum of irreducible representations},
it follows from \cite[\S 3]{FNO23} that
\[  \cE
  \subeq 
  \bigcap_{|t| < \pi/2} \cD(e^{i t \partial U(h)}).\]
and that the limits
\[ \beta^+(v) := \lim_{t \to -\pi/2} e^{it \partial U(h)} v
  \quad \mbox{ for } \quad v \in \cH^{[K]} \cap \cH^J \]
exist in the space $\cH^{-\infty}_{\rm KMS}$ of distribution vectors
satisfying the KMS condition \eqref{eq:UhJ2} 
(\cite[Prop.~9]{FNO23}). Then 
\begin{equation}
  \label{eq:eh}
  \sE_H := \beta^+(\cE^J) \subeq \cH^{-\infty}_{\rm KMS}
\end{equation}
is a finite-dimensional $H$-invariant subspace
(\cite[Prop.~7]{FNO23}).  

\subsection{Realization in vector bundles}
\mlabel{subsec:real}

For a  finite-dimensional $H$-invariant real subspace
$\sE_H \subeq \cH^{-\infty}$,
let 
\begin{equation}
  \label{eq:EM}
  \bE_H := G \times_H \sE_H \to G/H = \dS^d
\end{equation}
be the corresponding finite-dimensional real vector bundle.
Then we obtain with
\cite{FNO23} for every irreducible 
antiunitary representation $(U,\cH)$ of $G_{\tau_h}$ a
$G$-equivariant embedding 
\[ \cH \into \cH^{-\infty} \to C^{-\infty}(G/H, \bE_H) :=
  \cL(C^\infty_c(G/H,\bE_H),\C).\]
Here $\cL(C^\infty_c(G/H,\bE_H),\C)$ denotes the space of
continuous real-linear maps from the real locally convex LF space 
$C^\infty_c(G/H,\bE_H)$ of compactly supported smooth sections of the real
vector bundle $\bE_H$ to $\C$. This realization map is the adjoint of the map
\[ C^\infty_c(G/H,\bE_H) \to \cH^\infty \]
that is obtained by factorization from the map
\[ C^\infty_c(G,\sE_H)  \cong  C^\infty_c(G,\R) \otimes_\R \sE_H
  \to \cH^\infty, \quad
  \xi \otimes \eta \mapsto U^{-\infty}(\xi)\eta\]
through the isomorphism
\[ C^\infty_c(G,\R) \otimes_H \sE_H \cong  C^\infty_c(G/H,\bE_H)\] 
(see \cite{NO24} for details). 
Accordingly, we also have a map
\begin{equation}
  \label{eq:Phi}
 \Phi \: \cH^{-\infty} \to C^{-\infty}(G,\sE_H)
  := \cL(C^\infty_c(G,\sE_H),\C), \quad
  \Phi(\alpha)(\xi \otimes \eta) = \alpha(U^{-\infty}(\xi)\eta)
\end{equation}
to $\sE_H$-valued distributions on $G$. Actually its values
are contained in $C^{-\infty}(G/H, \bE_H)$, identified canonically
with a subspace of $C^{-\infty}(G, \sE_H)$. 
So we obtain a realization of the unitary representation 
$(U,\cH)$ of $G$ on a Hilbert subspace 
of $C^{-\infty}(G/H,\bE_H)\subeq C^{-\infty}(G,\sE_H)$,
a space of vector-valued distributions on $G$, on which $G$ acts unitarily 
by left translations. For many purposes, it
is more convenient to work with distributions on $G$,
instead of bundle-valued distributions on $G/H$.
By construction, the distributions that we consider
are boundary values of holomorphic functions. 

For $\alpha,\gamma \in \sE_H$, the distribution
\[ D_{\alpha,\gamma}^+(\xi) := \gamma(U^{-\infty}(\xi)\alpha) \]
defines an element of $C^{-\infty}(G/H,\bE_H)$.
For
\[ \sE_H^+ := \sE_H \quad \mbox{ and } \quad
  \sE_H^- := J \sE_H = \beta^-(\cE^J), \]
the representations
\[ \rho^\pm\: H \to \GL(\sE_H^\pm), \quad
  \rho^\pm(g)\alpha := U^{-\infty}(g) \alpha = \alpha \circ U^\infty(g^{-1}),
  \quad g \in H, \]
satisfy
\[ \rho^-(g) \circ J = J \circ \rho^+(\tau_h(g)) \quad \mbox{ for } \quad
  g \in H.\]
Here we use that $J \beta^\pm = \beta^\mp$ on $\cE^J$ and that
\begin{equation}
  \label{eq:JU-inter}
  J U(g) = U(\tau_h(g)) J \quad \mbox{ for } \quad g \in G.
\end{equation}
As $H_K = H \cap K = H^\theta = H^{\tau_h}$ is fixed pointwise by $\tau_h$,
the polar decomposition $H = H_K \exp (\fh_\fp)$ implies that
$\tau_h$ acts on $H$ like a Cartan involution.

  \begin{rem}
We define real scalar products on $\sE_H^\pm$ in such a way that
the maps $\beta^\pm \: \cE^J \to \sE_H^\pm$ are isometric.
Then $H_K$ acts isometrically on $\sE_H^\pm$, and
the representation $\dd U \: \g_\C \to \End(\cH^\infty)$ and the
automorphism   $\zeta := e^{-\frac{\pi i}{2} \ad h} \in \Aut(\g_\C)$ satisfy the relations  
\[ \beta^+ \circ \dd U(x) =  \dd U^{-\infty}(\zeta(x)) \circ \beta^+
  \quad \mbox{ and }\quad
  \zeta(i\fq \cap \fk) = \fh \cap \fp \]
(cf.\ \cite[Prop.~7]{FNO23}). This implies that the scalar products 
on $\sE_H^\pm$ are compatible with the Cartan involution~$\tau_h$ on $H$.
\end{rem}

\begin{rem} \mlabel{rem:selfdual}
    (a) The $G$-equivariant vectors bundles
    $\bE_H^\pm = G \times_H \sE_H^\pm$ are equivalent if and only
    if the $H$-representations $\rho^\pm$ on $\sE_H$ are equivalent.

    As $J \: \sE_H^+ \to \sE_H^-$ satisfies
    $J \circ \rho^+ \circ J = \rho^- \circ \tau_h$ by \eqref{eq:JU-inter}, and
    $\tau_h$ is a Cartan involution on $H$, we have
    \[ \rho^- \cong (\rho^+)^* \]
    as real representations of $H$.
   Hence $\rho^- \cong \rho^+$ if and only if
    $\rho^+$ is selfdual.

    \nin (b) Since $\SO_{1,d-1}(\R) \cong \SO_{1,d-1}(\R)_e \rtimes \{\1,\tau_h\}$,
    the selfdual representations of $H = \SO_{1,d-1}(\R)_e$ are precisely
    those extending to $\SO_{1,d-1}(\R)$.

    \nin (c) Particular examples are the trivial representation,
    the identical representation and all $\SO_{1,d-1}(\R)$-invariant
    subspaces of tensor powers of $\R^{1,d-1}$.     
 
 \nin
 (d) All the irreducible finite-dimensional
 representations of $\SO_{1,2}(\R )$ are uniquely determined 
by their dimension and hence self-dual.
\end{rem}

\section{The cut domain in de Sitter space} 
\mlabel{sec:3}

As in Subsection~\ref{subsec:2.6},
we assume that $(U,\cH)$ is an antiunitary irreducible representation
of $\SO_{1,d}(\R)$ and
fix a finite-dimensional subspace $\cE \subeq \cH$ invariant under $K$ and $J$.
Then analytic continuation defines on the domain
$\Xi_{G_\C} = G \exp(i\Omega_\fp)K_\C \subeq G_\C$ a holomorphic map 
  \[ L \: \Xi_{G_\C} \to B(\cE,\cH), \quad L(g \exp(ix) k) v
    := U(g) e^{i\partial U(x)} U_\cE(k)v,\]
  where $U_\cE \: K_\C \to B(\cE)$ is the holomorphic extension
  of the $K$-representation on $\cE$ to $K_\C$. 
Then   
\begin{equation}
  \label{eq:Q}
 Q(g_1, g_2) :=  L(g_2)^* L(g_1) \in B(\cE)\quad \mbox{ for }
 \quad g_1, g_2 \in \Xi_{G_\C}
\end{equation}
defines a sesquiholomorphic kernel on $\Xi_{G_\C} \times \Xi_{G_\C}$
(holomorphic in the first and antiholomorphic in the second argument).
For $g \in G$, we have
\begin{equation}
  \label{eq:Qcovar}
 Q(gg_1, gg_2)   = L(g_2)^* U(g)^{-1} U(g) L(g_1) 
  = L(g_2)^* L(g_1)  = Q(g_1, g_2).
\end{equation}
For the action of $G_\C$ on $G_\C^2$ by
\begin{equation}
  \label{eq:GonGC}
  g.(g_1, g_2) := (gg_1, \oline g g_2),
\end{equation}
it is therefore natural to try to extend $Q$ to a
$G_\C$-invariant kernel on
the domain
\[ G_\C^{2,\cut} := G_\C.(\Xi_{G_\C} \times \Xi_{G_\C})\subeq G_\C^2 \]
by the relation $Q(gg_1, \oline g g_2)  = Q(g_1, g_2)$,
but it is not obvious that this is well-defined.

We therefore   consider the open subset
\begin{equation}
  \label{eq:gccut}
 G_\C^\cut :=
  \{ \oline{g_2}^{-1} g_1 \: g_1, g_2 \in \Xi_{G_\C} \}
  =   K_\C \Xi_{A_\C} G \Xi_{A_\C} K_\C
  \subeq G_\C \quad \mbox{ for } \quad \Xi_{A_\C} = A \exp(i\Omega_\fa),
\end{equation}
for which
\[ G_\C^{2,\cut}
  = G_\C.(G_\C^\cut \times \{ e\}).\]
Accordingly, the domains 
\[ M_\C^\cut   = G_\C^\cut.i\be_0 \subeq M_\C 
  \quad \mbox{ and } \quad M_\C^{2,\cut}
  := G_\C.(\Xi_+ \times \Xi_+)  \subeq M_\C^2\]
satisfy
\[ M_\C^\cut = \{ m \in M_\C \:  (m,i\be_0) \in M_\C^{2,\cut} \}
  \quad \mbox{  and } \quad
  M_\C^{2,\cut} = G_\C.(M_\C^\cut \times \{ i \be_0\}).\]

The Extension Theorem (Theorem~\ref{thm:ext})
provides a holomorphic function
\[ \phi_\cE \: G_\C^\cut \to B(\cE), \]
from which we obtain a $G_\C$-invariant
kernel $\hat Q$ on $G_\C^{2,\cut}$ by
$\hat Q(g_1, g_2) := \phi_\cE(\oline{g_2}^{-1} g_1).$ 
The function $\phi_\cE$ is uniquely determined by 
\[ \phi_\cE(g) = Q(g, e) \quad \mbox{ for } \quad
  g \in \Xi_{G_\C},\]
so that $\hat Q$ extends the kernel $Q$ in a $G_\C$-invariant way to the
$G_\C$-invariant domain $G_\C^{2,\cut}$.

Analytic continuation implies that 
\begin{equation}
  \label{eq:sphefunc}
 \phi_\cE(k_1 g k_2) = U_\cE(k_1) \phi_\cE(g) U_\cE(k_2)
  \quad \mbox{ for } \quad
  k_1, k_2 \in K_\C, g \in G_\C^{\rm cut}.
\end{equation}
In particular, $\phi_\cE\res_G$ is a
{\it generalized spherical function} with respect to the
representation of $K \times K$ on $B(\cE)$ by
$(k_1, k_2).A = U_\cE(k_1) A U_\cE(k_2)^{-1}$
(cf.~\cite[Ch.~6]{Wa72}). 
As  the domain $G_\C^{\rm cut}$ is right $K_\C$-invariant,
its geometry is well-represented by its image
in the quotient space $M_\C \cong G_\C/K_\C$:
\begin{equation}
\label{eq:mccut}
\dS_\C^{d,\cut} =  M_\C^\cut := K_\C \Xi_{A_\C} G \Xi_{A_\C}.i\be_0 = K_\C \Xi_{A_\C}.\Xi_+.
\end{equation}

For $d = 1$ we have $K = \{e\}$, $G = A$, and
the orbit map 
\begin{equation}
  \label{eq:ds1c}
 \sigma \: \SO_{1,1}(\C)  \to \dS^1_\C, \quad
 \pmat{ a & b \\ b & a} \mapsto g.\be_1 = \pmat{b \\ a}
\end{equation}
is bijective (\cite[Ex.~3.4]{NO20}).  Moreover,
\begin{equation}
  \label{eq:expzetah}
 \exp(\zeta h) = \pmat{ \cosh\zeta & \sinh\zeta \\
   \sinh\zeta & \cosh\zeta}, \quad
 \exp(\zeta h).i \be_0 = \pmat{ i \cosh\zeta \\ i \sinh\zeta}
\ \  \mbox{ and  } \ \ 
 \exp(\zeta h).\be_1 = \pmat{ \sinh\zeta \\ \cosh\zeta}.   
\end{equation}
The discussion in Remark~\ref{rem:cosh} shows that
\[ \cosh(\cS_{\pm \pi}) = \cosh(\R) \cup \cosh(\cS_\pi)
  = \C \setminus (-\infty,-1], \quad
  \cosh(\pm i\pi + \R) = (-\infty,-1].\]
Since $\sigma$ is bijective, it follows that:

\begin{lem}
  \mlabel{lem:accut} 
$A_\C^\cut.i\be_0 = \exp(\cS_{\pm \pi}h).i\be_0 
=     \{ z \in A_\C.i\be_0 \: -i z_0 \not\in (-\infty, -1]\}.$ 
\end{lem}

This lemma implies in particular that 
\begin{equation}
  \label{eq:ds1cut}
 \dS^{1,\cut}_\C = A_\C^\cut.i\be_0
 = \{ (z_0, z_1) \in \dS^1_\C \: -i z_0 \not\in (-\infty, -1]\}
\end{equation}
is biholomorphic to a ``cut plane''; hence the name.

We now collect some geometric facts on complexified
de Sitter space $\dS^d_\C$ that allow us to identify the cut
domain~$\dS_\C^{d,\cut}$ explicitly (Proposition~\ref{prop:cplx-desit-cut}). 

\begin{lem} \mlabel{lem:witt} {\rm($K_\C$-orbits in $\dS^d_\C$)}
  In $\dS^d_\C$ there are two types of $K_\C$-orbits:
  \begin{itemize}
  \item[\rm(a)] Orbits through points $z = (z_0, \bz)$ with
    $\beta(\bz,\bz) \not=0$. These $K_\C$-orbits are represented
    by points $(z_0, z_1,0,\ldots, 0)$ with $z_1 \not=0$. 
  \item[\rm(b)] Orbits through points $z = (z_0, \bz)$ with
    $\beta(\bz,\bz) =0$. These
    $K_\C$-orbits are represented by the fixed points $\pm i \be_0$,
    and, for $d \geq 3$, also by the two vectors
    $\pm i \be_0 +  i \be_1 + \be_2$,
   and for $d \geq 2$,  by the four vectors 
    $\pm i \be_0 \pm  i \be_1 + \be_2$. 
  \end{itemize}
\end{lem}

\begin{prf} By Witt's Theorem, the group
  $\OO_d(\C)$ acts transitively on the level sets of the function
  $\bz \mapsto \beta(\bz,\bz)$. For non-zero values the stabilizer of $\bz$
  contains a reflection, so that also the subgroup $K = \SO_d(\C)$
  acts transitively. This implies (a). For (b) we observe that
  $0 = \beta(\bz,\bz) = 1 + z_0^2$ implies $z_0 = \pm i$.
  For $\bz = 0$ we thus obtain the $K_\C$-fixed points $\pm i \be_0$,
  and for $d \geq 2$, we also have elements of the form
  $(\pm i, \bz)$ with $\beta(\bz,\bz) = 0$.

  It remains to show that, each isotropic vector $\bz \in \C^d$
  lies in the $\SO_d(\C)$-orbit of $\bw := i \be_1 + \be_2 \in \C^d$
  for $d \geq 3$, and possibly in the orbit of $\oline \bw$ for $d = 2$.
  This holds by Witt's Theorem for the group $\OO_d(\C)$.
  If $d \geq 3$, there exists an element $\OO_d(\C)$ of determinant
  $-1$ fixing $\bw$, so that our claim follows. We may therefore
  assume that $d = 2$. Then
  \[ \SO_2(\C) = \Big\{\pmat{ a & b \\ -b & a} \: a^2 + b^2 = 1,
    a,b\in \C\Big\}.\]
  Now
  \[ \pmat{ a & b \\ -b & a} \bw
    = \pmat{ a & b \\ -b & a} \pmat{ i \\ 1} 
    = \pmat{ ai + b \\ a - i b}
    = \pmat{ i(a - i b) \\ a - i b}
    = (a-ib)(i\be_0 + \be_1)\] 
  shows that the $\SO_2(\C)$-orbit of $\bw$ consists of all
  non-zero   vectors of the form $(iz,z)$.
  If $(z_1, z_2)\not=0$ is isotropic, then $z_1^2 = -z_2^2 \not=0$
  implies that $z_1 = \pm i z_2$, so that $z_1 = - i z_2$ is also possible,
  but these vectors are not obtained from $\bw$ with $\SO_2(\C)$.
\end{prf}

The following proposition provides an explicit description
of the cut domain $\dS_\C^{d,\cut}$ in complex de Sitter space. 
  
\begin{prop} \mlabel{prop:cplx-desit-cut}
  In complexified de Sitter space $\dS^d_\C$, the cut domain is given
  by
  \[ \dS_\C^{d, \cut}
    = \{ z \in \dS^d_\C \: -i z_0 \not\in (-\infty, -1]  \}.\]
\end{prop}

\begin{prf} For $d = 1$, the assertion is \eqref{eq:ds1cut}.
  We may therefore assume that $d \geq 2$. 
  Let $\Omega$ denote the domain on the right hand side.
  As $K_\C = G_{\C,i\be_0}$ leaves the domains on both sides invariant,
  they are unions of $K_\C$-orbits, so that we may use the 
  description of these orbits in Lemma~\ref{lem:witt}.

With \eqref{eq:ds1c} we obtain a biholomorphic map
$A_\C \to \dS^1_\C = \dS^d_\C \cap (\C \be_0 + \C \be_1), a \mapsto a.i\be_0.$ 
  Lemma~\ref{lem:witt} thus implies that
  \begin{equation}
    \label{eq:e2a}
    \dS^d_\C = K_\C A_\C.i\be_0 \dot\cup K_\C.\{\pm i \be_0 \pm i\be_1 + \be_2\}.
  \end{equation}
In view of \eqref{eq:mccut}, we have 
  \begin{equation}
    \label{eq:dscut}
\dS_\C^{d,\rm cut} = K_\C \Xi_{A_\C}.\Xi_+ \quad \mbox{ with }\quad
\Xi_+ = \dS^d_\C \cap (\R^{1,d} + i V_+).
  \end{equation}

  \nin``$\subeq$'':   In view of \eqref{eq:dscut} and $K_\C$-invariance,
  we have to show that, for $w \in \Xi_+$
  and $|t| < \pi/2$, the element $z := \exp(ith).w$
  satisfies $-i z_0 \not\in (-\infty,-1]$.
  We have
  \[ -i z_0 = \beta(z, -i\be_0)
    = \beta(\exp(ith).w, -i \be_0)
    = \beta(w, -i \exp(-ith).\be_0),\]
  and since $-i\exp(-ith).\be_0 \in \Xi_-$
  (the complex conjugate of   $\Xi_+$),
  \[ \{ \beta(z, \oline w) \: z,w \in \Xi_+ \}
    = \beta(\Xi_+ \times  \Xi_-)= \C \setminus  (-\infty, -1] \]
(\cite[Lemma~3.5]{NO20}) 
  implies that $-iw_0 \not\in (-\infty, -1]$.

  \nin``$\supeq$'': Now suppose that
  $z = (z_0, \bz)$ is such that $-i z_0 \not\in (-\infty, -1]$.
  If $\beta(\bz,\bz)\not=0$, Lemma~\ref{lem:witt} implies that
  $z \in K_\C A_\C.i\be_0$ and by Lemma~\ref{lem:accut} it further follows that
  $z \in K_\C A_\C^\cut.i\be_0$.
  This shows that
  \[ z \in K_\C A_\C^\cut.i\be_0 =   K_\C \Xi_{A_\C} \Xi_{A_\C}.i\be_0
    \subeq K_\C \Xi_{A_\C}.\Xi_+ = \dS_\C^{d,\cut}.\]

  If $\beta(\bz,\bz) = 0$ and $\bz\not=0$, then the $K_\C$-orbit of
  $z$ contains   either $i\be_0 \pm  i \be_1 + \be_2$ or
  $-i\be_0 \pm i \be_1 + \be_2$ (Lemma~\ref{lem:witt}),
  but the second case is excluded by
  the condition on $z_0$. Therefore
  \[ z \in K_\C.(i\be_0 \pm i \be_1 + \be_2)
    = K_\C.\Big(i\be_0 \pm \frac{i}{2}\be_1
    + \frac{1}{2}\be_2\Big)     \in K_\C.\Xi_+ \subeq \dS_\C^{d, \cut} \]
  follows from $\be_0 \pm \frac{\be_1}{2} \in V_+$.
  Here we use that the subgroup of $K_\C$ acting
  on $\Spann\{\be_1, \be_2\}$ is isomorphic to
  $\SO_2(\C) \cong \C^\times$, and acts by 
  \begin{equation}
    \label{eq:so2c}
 \lambda.(i\be_0 \pm i\be_1 + \be_2) = i\be_0 + \lambda^{\pm 1}(\pm i\be_1 + \be_2),
 \quad \mbox{ for }\quad \lambda \in \C^\times.\qedhere
 \end{equation}
\end{prf}

The following observation also appears in \cite[Prop.~2.2]{BM96},
  where it is actually used to define the cut domain as the regularity
  domain for the scalar-valued kernels that we shall discuss
  in Section~\ref{sec:sphfunc} below.

\begin{cor} \label{eq:dsccut} We have 
\[  (\dS_\C^d \times \dS_\C^d)^\cut
  = \{ (z,w) \in \dS^d_\C \times \dS^d_\C  \:   \beta(z, \oline w)
  \not\in (-\infty, -1]\}. \] 
\end{cor}

\begin{prf} This follows from Proposition~\ref{prop:cplx-desit-cut},
the $G_\C$-invariance of the kernel $\beta(z, \oline w)$ 
on $\dS^d_\C \times \dS^d_\C$ and
\[  (\dS_\C^d \times \dS_\C^d)^\cut
  = G_\C.(\dS_\C^{d,\cut} \times \{i\be_0\}).\qedhere\]   
\end{prf}

\begin{rem} \mlabel{rem:5.6}
For the subset $K_\C A_\C^\cut. i \be_0 \subeq \dS_\C^{d,\rm cut},$ 
we obtain from Lemma~\ref{lem:accut} that 
  \begin{equation}
    \label{eq:dense1}
    K_\C A_\C^\cut. i \be_0
    = \{i \be_0\} \cup \{ (z_0, \bz) \in \dS_\C^{d, \rm cut} \:
    \beta(\bz,\bz) \not=0\}.
  \end{equation}
  For $d \geq 2$, the point $i \be_0$ is {\bf not} inner in this subset
  because $i\be_0 + \frac{1}{n}(i\be_1 + \be_2) \to i\be_0$ 
  (cf.~\eqref{eq:so2c}).
  Further, the description in Proposition~\ref{prop:cplx-desit-cut}
  shows that 
  \[  \dS_\C^{d,\rm cut} \subeq \oline{K_\C A_\C^\cut. i \be_0},
    \quad \mbox{   hence also that } \quad
    G_\C^{d,\rm cut} \subeq \oline{K_\C A_\C^\cut K_\C}.\]
\end{rem}

The holomorphically extended kernel $\hat Q$ will be defined
on the $G_\C$-invariant domain $(\dS_\C^d \times \dS_\C^d)^\cut$ which contains
$\Xi_+ \times \Xi_+$. It also intersects
$\dS^d \times \dS^d$ in an open $G$-invariant subset.
By $G$-invariance, this subset is $G.\dS^{d,\cut}$, where 
\[ \dS^{d,\cut} := \{ x \in \dS^d \: (x,\be_1) \in
  (\dS_\C^d \times \dS_\C^d)^\cut\}.\]
To determine this domain explicitly, we use
Corollary~\ref{eq:dsccut}.
For $w = \be_1$, we thus obtain that
\[ (z,\be_1) \in  (\dS_\C^d \times \dS_\C^d)^\cut
  \quad \mbox{ if and only if } \quad
z_1 =  -\beta(z,\be_1) \in \C \setminus [1,\infty).\]

This discussion leads to the following description of the
cut domain in $\dS^d \times \dS^d$:

\begin{lem} \mlabel{lem:dsd-spacelike} We have the equalities 
  \begin{align} \label{dsd1}
 (\dS^d \times \dS^d)^\cut
&  = \{ (x,y) \in \dS^d \times \dS^d  \:   \beta(x,y) > -1\} \\
    \label{dsd2}
    &  = \{ (x,y) \in \dS^d \times \dS^d  \:   \beta(x-y,x-y) < 0\},
  \end{align}
which is the set of spacelike pairs in $\dS^d$,
and, in particular, 
\[ \dS^{d,\rm cut}
  = \{ x \in \dS^d \: \beta(x,\be_1) > -1\}
  = \{ x \in \dS^d \: x_1 < 1\}.\]
\end{lem}

\begin{prf} From \eqref{eq:dsccut} we get \eqref{dsd1}. 
So the lemma follows immediately from
\[ \beta(x-y,x-y) = \beta(x,x) + \beta(y,y) - 2 \beta(x,y)
  = - 2 (1 + \beta(x,y))\]
for $x, y \in \dS^d$.
\end{prf}

Equation~\eqref{dsd2} implies in particular that 
\begin{equation}
  \label{eq:desit-cut}
  \dS^{d,\cut} = \dS^d \setminus \oline{I(\be_1)}
  = \{ x \in \dS^d \: x_1 < 1\},
\end{equation}
where
\begin{equation}
  \label{eq:I(x)}
  I(x) = \{x+ y\in \dS^d \: \beta (y,y) >0\},\quad\text{and}\quad I^\pm (x) =\{x+y\in I(x) \: \pm y_0> 0\}.
\end{equation}
So $\oline{I^\pm(x)}$ are the {\it causal future/past} of $x$.

\section{Analytic continuation of the generalized spherical
  function}
\mlabel{sec:4}

In this section we use the holomorphic kernel
$Q \: \Xi_{G_\C} \times \Xi_{G_\C} \to B(\cE)$
from \eqref{eq:Q}  
to obtain a holomorphic function on $G_\C^\cut = K_\C \Xi_{A_\C} \Xi_{G_\C}$
extending $Q(\cdot,e)$. The main point is
to argue that, for $g_1, g_2 \in \Xi_{G_\C}$, the value
of $Q(g_1, g_2)$ only depends on $\oline{g_2}^{-1} g_1$.
We than also inspect the boundary values of this holomorphic function
on the submanifold
\[ \exp\Big(-\frac{\pi i}{2}h\Big)
  G \exp\Big(-\frac{\pi i}{2}h\Big) =  \exp\Big(\frac{\pi i}{2}h\Big)
  G \exp\Big(\frac{\pi i}{2}h\Big)\]
from two ``sides'', along the curves
$\exp(ti h) g \exp(tih)$ for
$t \to \pm \pi/2$.

\subsection{Existence of the analytic continuation}

To extend the holomorphic function
\[ \phi_\cE \: \Xi_{G_\C} \to B(\cE), \quad
  \phi_\cE(z) = Q(z,e) \]
to the domain $G_{\C}^\cut$, we first construct a function
$F \: K_\C (A_\C^\cut\setminus \{e\}) K_\C \to B(\cE)$ satisfying
\[ F(k_1 g k_2) = U_\cE(k_1) F(g) U_\cE(k_2) \quad \mbox{ for } \quad
  k_1, k_2 \in K_\C.\]
Then we show that $F$ and $\phi_\cE$ coincide on the
connected intersection of their mutual domains, and then
we use Hartogs' Theorem asserting that
local boundedness implies extendability to obtain an extension
of $\phi_\cE$ to $G_\C^\cut$.

\begin{lem} \mlabel{lem:ext3}
  For $a_1, a_2 \in \Xi_{A_\C} = \exp(\cS_{\pm \pi/2}h)$, the value
  of $Q(a_1, a_2)$ only depends on $\oline{a_2}^{-1}a_1$.
\end{lem}

\begin{prf} The exponential map $\exp \:  \fa_\C = \C h \to A_\C \cong \C^\times$
  is equivalent to the exponential map of the complex plane,
  hence maps $\fa + 2 i \Omega_\fa = \cS_{\pm \pi} h$ biholomorphically
  to $A_\C^\cut$. Hence 
 the relation $\exp(-\oline{z_2})\exp(z_1) = \exp z$
 for $z_1, z_2 \in \fa + i \Omega_\fa$ and $z \in \fa + 2i\Omega_\fa$
 is equivalent to $z_1 - \oline{z_2} = z$. The
 set 
 \[ F_z := \{ (z_1,z_2) \in \fa_\C^2 \: z_1 - \oline{z_2} = z \} \]
 is a real affine subspace with translation space $\fa_\C$,
 hence in particular convex. Therefore it intersects the convex product set
 $(\fa + i \Omega_\fa)^2$ in a convex set, which corresponds to an open
 subset of $\fa_\C$ that is invariant under $\fa$-translations.
 As the function
 \[ (\fa + i \Omega_\fa)^2 \to B(\cE), \quad
   (z_1, z_2) \mapsto Q(\exp z_1, \exp z_2) \]
 is sesquiholomorphic and defines an $A$-invariant kernel, 
 its restriction to the fibers $F_z$ is holomorphic in the $z_1$-coordinate
 and  constant along the $A$-orbits, hence constant by connectedness
 of the fibers.
  \end{prf}

With Lemma~\ref{lem:ext3}, we first
find a holomorphic function
\[ F_A \: A_\C^\cut = \Xi_{A_\C} \Xi_{A_\C} \to B(\cE),
  \quad \mbox{ satisfying } \quad
  F_A(\oline b^{-1} a) = Q(a,b) \quad \mbox{ for } \quad
  a,b \in \Xi_{A_\C}.\]
Then we consider on
the product space $K_\C \times A_\C^\cut \times K_\C$
the holomorphic function
\[ \hat F(k_1, a, k_2) := U_\cE(k_1) F_A(a) U_\cE(k_2).\]

\begin{lem} \mlabel{lem:4.26} There exists a function
$F \: K_\C A_\C^\cut K_\C \to B(\cE)$   satisfying
  \[   F(k_1 a k_2) := U_\cE(k_1) F_A(a) U_\cE(k_2) 
    \quad \mbox{ for }\quad k_1, k_2 \in K_\C, a \in A_\C^\cut.\]
\end{lem}

\begin{prf} We have to show that,
  $k_1 a k_2 = k_1' a' k_2'$ in $G_\C$ implies that
  $\hat F(k_1, a, k_2) = \hat F(k_1', a', k_2')$, so that we may define
  $F$ accordingly on $k_1 a k_2$.

Replacing
$k_1$ by $(k_1')^{-1}k_1$ and 
$k_2$ by $k_2(k_2')^{-1}$, we may w.l.o.g.\ assume that
$k_1' = k_2' = e$. So let $k_1, k_2 \in K_\C = \SO_d(\C)$ and
$a = a_z = \exp(zh)$, $a' = a_{z'} = \exp(z'h)$ with $|\Im z| < \pi$
and $|\Im z'| < \pi$ such that
\[ a_{z'} = k_1 a_z k_2.\]
Applying both sides to $\be_0$, we obtain
\[ \cosh(z') \be_0 + \sinh(z') \be_1
  = \cosh(z) \be_0 + \sinh(z) k_1.\be_1.\]

\nin {\bf Case 1:} $z = 0$. Then $a_z = e$ and
$\be_0 = \cosh(z') \be_0 + \sinh(z') \be_1$ implies
$e^{z'} = 1$, so that $z' = 0$ follows from $|\Im z'| < \pi$.
Thus $a' = e$ and $k_2 = k_1^{-1}$. We therefore
have
\[ \hat F(k_1, a, k_2)
  = U_\cE(k_1) F_A(e) U_\cE(k_2)
= U_\cE(k_1)  U_\cE(k_1)^{-1} = \1 = \hat F(e,a',e).\] 

\nin {\bf Case 2:} $z \not= 0$.
Then $\sinh(z) \not=0$ implies $k_1.\be_1 = \pm\be_1$
and $\sinh(z') = \pm \sinh(z)$. We thus obtain $z' = \pm z$,
i.e., $a' = a^{\pm 1}$. In the first case $k_1 \in Z_K(A)_\C = \SO_{d-1}(\C)$,
and in the second case $k_1.\be_1 = - \be_1$ implies that
$k_1 \in N_{K_\C}(\fa_\C)$ and $\Ad(k_1)h = - h$.
Hence
\[ a^{\pm 1} = a' = k_1 a k_2 = a^{\pm 1} k_1 k_2 \]
implies $k_2 = k_1^{-1}$. 
Using the polar decomposition
$N_{K_\C}(\fa_\C) = N_K(\fa) Z_K(A)_\C,$ 
we write $k_1 = k_0 m$ with $m \in Z_K(A)_\C$ and $k_0  \in N_K(\fa)$.
By analytic continuation, 
\[ F_A(\cA_\C^\cut) \subeq  U_\cE(Z_K(A)_\C)',\]
so that we obtain
\begin{align*}
 \hat F(k_1, a, k_2)
&  = U_\cE(k_1) F_A(a) U_\cE(k_1)^{-1} 
  = U_\cE(k_0) F_A(a) U_\cE(k_0)^{-1}\\
  &  = \phi_\cE(k_0 a k_0^{-1}) = \phi_\cE(a^{\pm 1}) =  \hat F(e,a',e).
\end{align*}
This completes the proof of the lemma.   
\end{prf}

Next we observe that the domain
\begin{equation}
  \label{eq:cd1}
  \cD_1 := K_\C (A_\C^\cut\setminus \{e\}) K_\C \subeq G_\C
\end{equation}
is open. In fact,
\begin{equation}
  \label{eq:cd2}
 \cD_1.i\be_0
  = \{ (z_0, \bz) \in \dS_\C^d \: \beta(\bz,\bz) \not=0,
  -i z_0 \not\in (-\infty, -1] \}
\end{equation}
(cf.~Remark~\ref{rem:5.6}) is an open subset of $\dS_\C^d$, so that
its inverse image $\cD_1 \subeq G_\C$ is open as well.

The set $\fa_\C^\cut \setminus \{0\}$
is open and consists of elements of the form 
$z h$ with $\Fix(e^{\ad (zh)}) = \fz_\g(\fa_\C)$.
The map 
\[ \Psi \: (K_\C/Z_{K_\C}(A)) \times A_\C \to G_\C^{-\theta}, \quad 
  (k,a) \mapsto k a k^{-1} \] 
is well-defined, holomorphic, and its differential
is regular in all points $(k,\exp x)$, for which the function 
$\exp \: \fp_\C \to G_\C^{-\theta}$ is regular in $x$ and
$Z_{K_\C}(\exp x) = Z_{K_\C}(A)$.
By \cite[Lemma~C.3]{NO23}, the fist condition
is satisfied for $x \in \fa_\C^\cut$, and if, in addition,
$x \not=0$, then also the second condition is satisfied.
Hence $\Psi$ is locally biholomorphic around $(e,a)$ for
$a \in A_\C^\cut\setminus \{e\}$. This easily implies that
$F\res_{\cD_1}$ is holomorphic. 

\begin{lem}
  \mlabel{lem:gamma} For the holomorphic function 
\[ \gamma \: \dS^d_\C \to \C,\quad
  z = (z_0, \bz) \mapsto \beta(\bz,\bz) = - \bz^2, \]
the value $0$ is regular. Likewise, $0$ is a regular value of
the function
\[ \gamma_G \:  G_\C \to \C, \quad
  \gamma_G(g) := \gamma(g.i\be_0),  \]
so that $\gamma_G^{-1}(0)$ is a holomorphic hypersurface of $G_\C$.
\end{lem}

\begin{prf}
If $\dd\gamma(z)$ vanishes on $T_z(M_\C)$, then
\[ \{0\} = \beta(\bz, T_z(M_\C)) = \beta(\bz, z^{\bot_\beta}) \quad
\mbox{ for } \quad z^{\bot_\beta} = \{ w \in \C^{1,d} \: \beta(z,w) = 0\}, \]
implies that
$\bz \in \C z$, which only happens for $z_0 = 0$.
Then $-1 = \beta(z,z) = - \beta(\bz,\bz)$,
which never happens in the $0$-fiber of~$\gamma$.
Therefore $\gamma^{-1}(0)$ is a complex hypersurface in $\dS^d_\C$.
Accordingly, $\gamma_G^{-1}(0)$ is a holomorphic hypersurface of $G_\C$.
\end{prf}

\begin{lem} \mlabel{lem:connected}
  The open subset $\Xi_{G_\C} \cap \cD_1 \subeq G_\C$ is connected. 
\end{lem}

\begin{prf}
  From $\Im z_0 > 0$ for $z \in \Xi_+$ and  \eqref{eq:cd2} we derive that
\begin{equation}
  \label{eq:d1a}
 \cD_1.i\be_0 \cap \Xi_+ 
 = \{ (z_0, \bz) \in \Xi_+  \: \beta(\bz,\bz) \not=0 \}, 
\end{equation}
so that $\cD_1 \cap \Xi_{G_\C} = \gamma_G^{-1}(\C^\times) \cap \Xi_{G_\C}$ 
is the complement of a complex hypersurface in the connected
  open subset $\Xi_{G_\C}$, hence connected by
\cite[Cor~I.C.4]{GR65}. 
\end{prf}

\begin{thm} {\rm(The Extension Theorem)} \mlabel{thm:ext}
  The holomorphic function
  $\phi_\cE := Q(\cdot,e) \: \Xi_{G_\C} \to B(\cE)$ extends to a holomorphic
  function on $G_\C^\cut$.   
\end{thm}

\begin{prf}
The definition of $F$ in Lemma~\ref{lem:4.26} implies that
$F$ coincides with $\phi_\cE$ on the real group 
$G = K AK$, hence in particular on
$K (A\setminus \{e\}) K = G \setminus K$,
which is a totally real submanifold of 
$\Xi_{G_\C} \cap \cD_1$. We conclude that
$F = \phi_\cE$ holds on the connected complex manifold
$\Xi_{G_\C} \cap \cD_1$ (Lemma~\ref{lem:connected}).
As a consequence, we can use $F$ to
extend the holomorphic function $\phi_\cE$ from $\Xi_{G_\C}$ to a
holomorphic function
\begin{equation}
  \label{eq:f1}
  F_1 \:  \Xi_{G_\C} \cup \cD_1 \to B(\cE).
\end{equation}

We recall from \eqref{eq:cd2} that $\cD_1$ maps to
\[  \cD_1.i\be_0
  = \{ (z_0, \bz) \in \dS_\C^d \: \beta(\bz,\bz) \not=0,
  -i z_0 \not\in (-\infty, -1] \}
\subeq \dS^{d,\cut}_\C = \{ z \: -i z_0 \not\in (-\infty, -1]\}.
\]

We now argue that $F_1$ is locally
bounded on neighborhoods of any point $g \in \gamma_G^{-1}(0)$. 
For $g \in K_\C$,  this follows from $g \in \Xi_{G_\C}$
and the holomorphy of $F_1$ on $\Xi_{G_\C}$.
Any other point in $\gamma^{-1}(0) \cap \dS_\C^{d,\cut}$
lies in the $K_\C$-orbit of
$i \be_0 \pm   \frac{i}{2} \be_1 + \frac{1}{2}\be_2
\in \Xi_+$ (Lemma~\ref{lem:witt}),
so that, for  any $g \in \gamma_G^{-1}(0) \cap G_\C^\cut$,
the $K_\C$-double coset intersects $\Xi_{G_\C}$.
This proves local boundedness around~$g$.
Now Riemann's Removable Singularity Theorem in several variables
(\cite[Thm.~I.C.3]{GR65}) 
implies that $F_1$ extends to a holomorphic function
on $G_\C^\cut$.
\end{prf}

\subsection{Boundary behavior of  the analytic continuation}
\mlabel{sec:5}

In this section we take a closer look at the boundary behavior
of the holomorphic function
  \[ \phi_\cE \: G_\C^\cut \to B(\cE) \quad \mbox{ satisfying } \quad
    Q(g_1, g_2) = \phi_\cE(\oline{g_2}^{-1} g_1) \quad\mbox{ for } \quad 
    g_1, g_2 \in \Xi_{G_\C}\]
  (cf.\ \eqref{eq:Q} and Theorem~\ref{thm:ext}). 
  We now take a closer look at the function $\phi_\cE$.
  
\begin{lem}    \mlabel{lem:phi-props}
   The function $\phi_\cE$ has the following properties:
\begin{itemize}
\item[\rm(a)] $\phi_\cE(g)^* = \phi_\cE(\oline g^{-1})$ for $g\in G_\C^\cut$.
\item[\rm(b)] For $\phi^{v,w}_\cE(g) := \la v, \phi_\cE(g) w \ra$
  for $v,w \in \cE^J$, we have 
  $ \oline{\phi^{v,w}_\cE(g)}  = \phi^{w,v}_\cE(\oline g^{-1}).$
\item[\rm(c)]  $J \phi_\cE(g) J = \phi_\cE(\tau_h(\oline g))
  = \phi_\cE(\exp(\pi i h) \oline g \exp(-\pi i h))$ for $g\in G_\C^\cut$.
\item[\rm(d)] $\oline{\phi^{v,w}_\cE(g)}
  = \phi^{v,w}_\cE(\tau_h(\oline g))$ for $g \in G_\C^{\rm cut}$ and
  $v,w \in \cE^J$. 
\end{itemize}
\end{lem}

\begin{prf} (a) follows by analytic continuation from
$\phi_\cE(g^{-1}) = \phi_\cE(g)^*$ for $g \in G$. 

\nin (b) follows from (a) by 
\[  \oline{\phi^{v,w}_\cE(g)}
  = \oline{\la v, \phi_\cE(g) w \ra}
  = \la \phi_\cE(g) w,v \ra
  = \la w, \phi_\cE(g)^* v \ra
  = \la w, \phi_\cE(\oline g^{-1}) v \ra 
  = \phi^{w,v}_\cE(\oline g^{-1}).\]
    
\nin (c) follows from the relation
$J \phi_\cE(g) J = \phi_\cE(\tau_h(g))$ for $g \in G$ 
by analytic continuation
  because both equations hold for $g \in G$, and all terms are
  antiholomorphic in~$g$.

\nin (d) With (c) we obtain 
  \begin{align*}
 \phi^{v,w}_\cE(g)
&    =     \la v, \phi_\cE(g) w \ra 
    =     \la Jv, \phi_\cE(g) J w \ra 
 \ {\buildrel {\rm(c)} \over =}\
  \la Jv, J \phi_\cE(\tau_h(\oline g)) w \ra \\
&    =     \la \phi_\cE(\tau_h(\oline g)) w, v \ra 
    \ {\buildrel {\rm(a)} \over  =}\
    \la  w, \phi_\cE(\tau_h(g)^{-1}) v \ra 
    = \phi^{w,v}_\cE(\tau_h(g)^{-1}) 
 \ {\buildrel {\rm(b)}\over =}\
  \oline{\phi^{v,w}_\cE(\tau_h(\oline g))}.  
\qedhere  \end{align*}
\end{prf}

\begin{rem} {\rm(The connection with \cite{BM96})}
  In \cite{BM96} Bros and Moschella consider kernels corresponding
  to spherical functions. This corresponds to a trivial
  one-dimensional $K$-representation~$\cE$, so that our generalized
  spherical function $\phi_\cE$ on $G_\C^\cut$
  actually becomes a spherical function
  on which is $K_\C$-biinvariant, hence factors through a
  function on the corresponding domain $\dS_\C^{d,\cut}$.

  Bros and Moschella consider   so-called {\it perikernels}
  $W$ on open domains in $\dS^d_\C \times \dS^d_\C$ which are
  holomorphic in both arguments and invariant under the diagonal
  action of $G_\C$. Concretely,
      \[ W(g_1.i\be_0, g_2.i\be_0)
      =  \phi_\cE(g_2^{-1} g_1) =  Q(g_1, \oline{g_2}) \]
    connects their approach with our notation. Recall that
    \[ \exp\Big(-\frac{\pi i}{2}h\Big).i\be_0 = e^{-\frac{\pi i}{2}h}.i\be_0 = \be_1 \]
    holds for the base points in $\bH^d_+$ and $\dS^d$. 

    The corresponding ``function'' $w = W(\cdot, \be_1)$ 
    of one argument on $\dS^d_\C$ which corresponds to the $G$-invariant
    kernel $W$ is 
    \begin{align*}
      w(g.\be_1) &
:= W(g.\be_1, \be_1)
= W\Big(g\exp\big(-\frac{\pi i}{2}h\big).i\be_0,
        \exp\big(-\frac{\pi i}{2}h\big).i\be_0\Big)\\
 &      = \phi_\cE\Big(\exp\Big(\frac{\pi i}{2}h\Big)
       g\exp\Big(-\frac{\pi i}{2}h\Big)\Big),
    \end{align*}
    hence defined on $\exp(-\frac{\pi i}{2}h) G_\C^\cut\exp(\frac{\pi i}{2}h)$,
    which contains in particular
$A_\C^\cut      = \exp(\cS_{\pm \pi} h)$. 
With the description of $\dS^{d,\cut}_\C$ from
Proposition~\ref{prop:cplx-desit-cut}  we see that,
    as a holomorphic function on $\dS^d_\C$, $w$ is defined on 
    \begin{align*}
 \exp\Big(-\frac{\pi i}{2}h\Big)\dS_\C^{d, \cut}
&     = \{ z \in \dS_\C^d \:
      \beta\Big(\exp\Big(\frac{\pi i}{2}h\Big)z,i\be_0\Big) \not\in (-\infty,-1]\}\\
&   = \{ z \in \dS_\C^d \:
                                                                                       \beta(z,\be_1)\not\in (-\infty,-1]\}
                                                                                       = \{ z \in \dS_\C^d \: z_1 \not\in [1,\infty)\}.
    \end{align*}
  \end{rem}

  For $v,w \in \cE^J$ and $g_1, g_2 \in G$, we have
  \[ \la v, \phi_{\cE}(g_2^{-1} g_1)w \ra
    = \la v, Q(g_1, g_2) w \ra = \la v, U(g_2^{-1} g_1) w \ra.\]
  For $t, s \in \R$ we thus obtain
  \begin{align*}
 \la e^{-t \partial U(h)}v, U(g_2^{-1} g_1) e^{s \partial U(h)} w \ra
&    = \la v, U\big(\exp(th) g_2^{-1} g_1 \exp(sh)\big) w \ra \\
    &    =  \la v, \phi_{\cE}\big(\exp(th) g_2^{-1} g_1 \exp(sh)\big)w \ra,
  \end{align*}
     so that analytic continuation leads for $|t|,  |s| < \pi/2$ to   
\begin{align}\label{eq:phihrel}
 \la e^{it \partial U(h)}v, U(g_2^{-1} g_1) e^{is \partial U(h)}w \ra
  &=  \la v, \phi_{\cE}(\exp(it h) g_2^{-1} g_1 \exp(ish))w \ra.
\end{align}

For $v,w \in \cE^J$, we consider the distribution vectors
\[ \alpha := \beta^+(v) \quad \mbox{ and } \quad
  \gamma := \beta^+(w) \quad \mbox{ in }  \quad \sE_H^+.\]
Then the distribution
\[ D^+_{\alpha,\gamma}(\xi) := \gamma(U^{-\infty}(\xi)\alpha)\]
on $G$ (cf.\ \cite{NO24}) satisfies
\begin{align*}
 \la U^{-\infty}(\xi)\alpha, U^{-\infty}(\psi)\gamma \ra 
  & = (U^{-\infty}(\psi)\gamma)(U^{-\infty}(\xi)\alpha) 
  = \gamma(U^{\infty}(\psi^*) U^{-\infty}(\xi)\alpha)\\
  & = \gamma(U^{-\infty}(\psi^* * \xi)\alpha) = D^+_{\alpha,\gamma}(\psi^* * \xi).
\end{align*}
To relate this distribution to $\phi_\cE$, we write it as
\begin{align*}
  D^+_{\alpha,\gamma}(\xi)
  &= \gamma(U^{-\infty}(\xi)\alpha)
=  \la U^{-\infty}(\xi)\alpha, \gamma \ra 
= \lim_{t \to \pi/2-} \la U(\xi) e^{-it \partial U(h)} v, e^{-it \partial U(h)} w \ra\\
&    = \lim_{t \to \pi/2-} \int_G \oline{\xi(g)}
    \la U(g)e^{-it \partial U(h)} v, e^{-it \partial U(h)} w \ra\, dg.
\end{align*}
In this sense, we have in the space $C^{-\infty}(G,\C)$ of distributions
on $G$ the relation
\begin{equation}
  \label{eq:dalphabeta}
 D^+_{\alpha,\gamma}(g) 
  =\lim_{t \to \pi/2-}
  \la U(g)e^{-it \partial U(h)} v, e^{-it \partial U(h)} w \ra 
\ {\buildrel \eqref{eq:phihrel} \over  =}\ \lim_{t \to \pi/2-}
\la v, \phi_\cE(\exp(-it h) g^{-1} \exp(-it h)) w \ra. 
\end{equation}

\begin{lem} \mlabel{lem:dab}
  The distributions $D^+_{\alpha,\gamma}$ are boundary values of
  holomorphic functions on the domain 
\[ G_\C^{\rm hol} :=
  \exp\Big(\frac{\pi i}{2}h\Big) G_\C^{\rm cut} 
  \exp\Big(\frac{\pi i}{2}h\Big) \subeq G_\C.\]
This domain is invariant under
complex conjugation with respect to $G$, 
inversion, and conjugation with $\exp(\pi i h)$.
We also have 
\[ G_\C^{\rm hol} = H_\C \exp(\cS_\pi h) G \exp(\cS_\pi h) H_\C 
= H_\C \exp(-\cS_\pi h) G \exp(-\cS_\pi h) H_\C.\]
\end{lem}

\begin{prf} First we note that 
\[ G_\C^{\rm cut} = \oline{\Xi_{G_\C}} \Xi_{G_\C}
  = K_\C \exp(i\Omega_\fa) G \exp(i\Omega_\fa) K_\C  \]
implies that $G_\C^{\rm cut}$ is invariant under inversion and complex conjugation.
In view of  
\[ \exp(it h) g^{-1} \exp(-it h)
  \in \Xi_{A_\C} G \Xi_{A_\C} \subeq G_\C^\cut, \]
the first assertion thus follows from \eqref{eq:dalphabeta}.

With 
\[   \exp\Big(\frac{\pi i}{2}h\Big) K_\C   \exp\Big(-\frac{\pi i}{2}h\Big)
  = H_\C \]
(see \cite[Prop.~7]{FNO23}) we get 
\begin{align*}
G_\C^{\rm hol} &=   \exp\Big(\frac{\pi i}{2}h\Big) G_\C^{\rm cut}
                 \exp\Big(\frac{\pi i}{2}h\Big)
                 = H_\C \exp\Big(\frac{\pi i}{2}h+ \Omega_\fa\Big) G
  \exp\Big(\frac{\pi i}{2}h + \Omega_\fa\Big) H_\C\\
  &  = H_\C \exp(\cS_\pi h) G \exp(\cS_\pi h) H_\C.  
\end{align*}
Then $\exp(-\pi ih) = \exp(\pi i h)$ implies that complex conjugation
with respect to $G$ yields the domain 
\begin{align*}
  \oline{G_\C^{\rm hol}}
  &=   \exp\Big(-\frac{\pi i}{2}h\Big) G_\C^{\rm cut}
                         \exp\Big(-\frac{\pi i}{2}h\Big)
  = \exp(-\pi i h) G_\C^{\rm hol} \exp(-\pi i h) 
  = \tau_h(G_\C^{\rm hol}) = G_\C^{\rm hol}.
\end{align*}
From this we further derive that
$G_\C^{\rm hol} = H_\C \exp(-\cS_\pi h) G \exp(-\cS_\pi h) H_\C.$
\end{prf}

The subset $G_\C^{\rm hol}\subeq G_\C$ is open and contains $G$ in its boundary,
but if $z,w \in \cS_\pi$ both approach
$\pi i$, then 
$\exp(zh) \exp(wh) \to \exp(2\pi i h) = e$ (in $G_\C = \SO_{1,d}(\C)$).
This is how we can also approach $\dS^d$ by passing through the
domain $\Xi_-$, which contains $\exp(zh).\be_1$ for
$ \pi < \Im z < 2\pi$.

Note that
\[ \exp(\pi i h) g \exp(\pi i h) = \tau_h(g), \]
so that ``coming from the other side'' involves a $\tau_h$-conjugation
with respect to $G$.

\begin{rem} In the sense of distribution kernels, we may write
  the ill-defined expression $\gamma(U^{-\infty}(g)\alpha)$
  as $D^+_{\alpha, \gamma}(g)$, we get from \eqref{eq:dalphabeta} 
\begin{align} \label{eq:dab1}
 D^+_{\alpha,\gamma}(g)
&  = \phi^{v,w}_\cE\Big(\exp\Big(-\frac{\pi i}{2}h\Big)
                       g^{-1} \exp\Big(-\frac{\pi i}{2}h\Big)\Big).
\end{align}
With Lemma~\ref{lem:phi-props}(b), this leads to
\begin{equation}
  \label{eq:dab2}
 \oline{D^+_{\alpha,\gamma}(g)} 
  = \phi^{w,v}_\cE\Big(\exp\Big(-\frac{\pi i}{2}h\Big)
  \oline g\exp\Big(-\frac{\pi i}{2}h\Big)\Big),
\end{equation}
                     and with Lemma~\ref{lem:phi-props}(d) to
 \begin{equation} \label{eq:dab3}
 \oline{D^+_{\alpha,\gamma}(g)} 
  = \phi^{v,w}_\cE\Big(\exp\Big(\frac{\pi i}{2}h\Big)
  \tau_h(\oline g)^{-1}\exp\Big(\frac{\pi i}{2}h\Big)\Big).
  \end{equation}

Further, 
\[ D^+_{\alpha,\gamma}(h_2^{-1} gh_1) = D^+_{h_1.\alpha,h_2.\gamma}(g). \]
Hence, for $\gamma = \alpha$, we obtain
\[ D^+_\alpha := D^+_{\alpha,\alpha} = \Phi(\alpha)
  \in C^{-\infty}(G/H,\bE_H), \qquad
  \bE_H = G \times_H \sE_H\]
(cf.\ Subsection~\ref{subsec:real}). 
\end{rem}

By \eqref{eq:dab1}, 
the analytic continuation of $D_{\alpha,\gamma}^+$
to elements of the form $\exp(-is h)g\exp(-ish)$,
for $0 < s < \pi$, is given by
\begin{align}
  \label{eq:dalphabetas}
&  \phi^{v,w}_\cE\Big(\exp\Big(\big(s- \frac{\pi}{2}\big) ih\Big) g^{-1}
  \exp\Big(\big(s- \frac{\pi}{2}\big) ih\Big)\Big)  
\end{align}
and for $s \to 0+$, we recover $D_{\alpha,\gamma}$.
For $s \to \pi-$, we get with \eqref{eq:dab3} the limit distribution 
\[  D_{\alpha,\gamma}^-(g)
  = \phi_\cE^{v,w}\Big(\exp\big(\frac{\pi i}{2}h\big) g^{-1}
  \exp\big(\frac{\pi i}{2}h\big)\Big)
  =   \oline{D_{\alpha,\gamma}^+(\tau_h(\oline g))}.\]
We thus obtain:

\begin{prop} The distributions $D_{\alpha,\gamma}^\pm$ on $G$ satisfy 
$D_{\alpha,\gamma}^- =   \oline{D_{\alpha,\gamma}^+ \circ \tau_h}.$ 

\end{prop}

If the $H$-representation $(\rho^+, \sE_H^+)$ is selfdual
(cf.~Remark~\ref{rem:selfdual}), then the bundles
$\bE_H^\pm$ are equivalent. Therefore the subspaces
$\sE_H^\pm \subeq \cH^{-\infty}$ lead to two different embeddings
\[ \Phi^\pm \: \cH^{-\infty} \to
  C^{-\infty}(G/H,\bE_H^+) \cong  C^{-\infty}(G/H,\bE_H^-).\]

\section{Locality of the net}
\mlabel{sec:6}

In this section we show that the nets
$\sH^{\dS^d}_{\sH_E}$ on de Sitter space obtained from the construction in
\cite{FNO23} also satisfy the locality condition.
Our argument is based on a convenient description of the set
of spacelike pairs in $\dS^d \times \dS^d$ as
$G.(\WdS\times (-\WdS))$ where $W_R= \{ x\in\dS^d\: x_1>|x_0|\}$ is
the Rindler wedge in Minkowski space $\R^{1,d}$ and $\WdS $ is
the canonical wedge
\begin{equation}\label{eq:CanWedge}
\WdS  = W_R\cap \dS^d=\{x\in \dS^d\: x_1> |x_0|\}.
\end{equation}

\subsection{Spacelike pairs in de Sitter space} 
We call $W \subeq \dS^d$ a {\it wedge region}
if there exists a $g \in \SO_{1,d}(\R)$ such that
$W = g.\WdS$ and write $\cW$ for the set of wedge regions in
$\dS^d$. We now describe the set of spacelike separated pairs
in $\dS^d$ in terms of wedge regions. 

\begin{prop} \mlabel{prop:b.7}
  The set $(\dS^d \times \dS^d)^{{\rm spacelike}}$ of spacelike
  pairs in $\dS^d$ coincides with
  \[ \SO_{1,d}(\R).(\WdS \times (-\WdS))
    = \bigcup_{W \in \cW} W \times (-W).\] 
  For $d > 1$, it even coincides with 
$\SO_{1,d}(\R)^\up.(\WdS \times (-\WdS)).$ 
\end{prop}

\begin{prf} For each wedge region
  $W = g.\WdS = \dS^d \cap g.W_R  \subeq \dS^d$,
  the  cone $W - W \subeq g.(W_R - W_R) = g.W_R$ consists
  of spacelike pairs. Therefore it remains to  show that,
  if two points $p,q \in \dS^d$ are spacelike,
  then there exists a wedge region $W \subeq M$ with
  $p \in W$ and $q \in - W$. 
As $G$ acts transitively on $\dS^d$, we
  may w.l.o.g.\ assume that $p = \be_1$. 
That $q = (q_0, \bq) \in \dS^d$ is spacelike with respect to
$\be_1$ is equivalent to $q_1 = - \beta(q,\be_1) < 1$
(Lemma~\ref{lem:dsd-spacelike}).

\nin {\bf Case 1:} $q_1 \in (-1,1)$.
Then
\[  q_0^2 - q_2^2 -\cdots - q_d^2 = -1 + q_1^2 < 0,\]
and acting with the stabilizer of $\be_1$, it follows that
we may also assume that $q_0 = 0$. Hence it suffices to
consider points $q \in \bS^{d-1} = \dS^d \cap \be_0^\bot$.

For two points $p,q \in \bS^{d-1}$, being spacelike is equivalent to
$p \not=q$.  On the other hand, the Rindler wedge 
$W_R$ intersects $\bS^{d-1}$ in the open half sphere containing
$\be_1$ and $-W_R = W_R'$ intersects $\bS^{d-1}$ in the opposite
open half sphere. Acting with $\SO_d(\R) \subeq \SO_{1,d}(\R)$, we
therefore obtain all open half spheres 
as intersections $W \cap \bS^{d-1}$ for a wedge region $W = g \WdS
\subeq \dS^d$.
Hence the assertion follows from the fact that any pair of different points
$p,q \in \bS^{d-1}$ can be separated by two complementary open
half spheres.

\nin {\bf Case 2:} $q_1 = -1$. 
Acting with the stabilizer of $\be_1$,
this case can be reduced to $q = q_0 \be_0 - \be_1 + q_2 \be_2$
with $q_0^2 = q_2^2$ and $|q_0| < 1 = -q_1$. Then
$q \in - W_R$ and $p \in W_R$.

\nin {\bf Case 3:} $q_1 < -1$. 
Acting with the stabilizer of $\be_1$,
this case can be reduced to $q = \sqrt{q_1^2 -1} \be_0 + q_1 \be_1$.
Then $p \in W_R$ and $q \in - W_R$.
\end{prf}

\begin{ex} For $d = 1$, the two connected components
  $\dS^1_\pm = \pm W_{\dS^1}$ are mutually spacelike 
  and, conversely, if $x$ and $y$ are spacelike,
  then they are contained in different connected components:
  \[ (\dS^1 \times \dS^1)^{\rm spacelike} =
    (\dS^1_+ \times \dS^1_-) \cup    (\dS^1_- \times \dS^1_+).\] 
\end{ex}

\subsection{Spacelike complements in de Sitter space}

We know from Lemma~\ref{lem:dsd-spacelike}
that $x, y \in \dS^d$ are spacelike if and only if 
$\beta(x,y) > -1.$ 
We conclude that, for an open subset $\cO \subeq \dS^d$,
the {\it (open) spacelike complement} is given by
\[ \cO' = \{ y \in \dS^d \:  (\forall x \in \cO)\ \beta(x,y) > -1\}^\circ,\]
and two open subsets $\cO_1$ and $\cO_2$ are spacelike separated if and only if
\begin{equation}
  \label{eq:betao1o2}
  \beta(\cO_1 \times \cO_2) \subeq (-1,\infty).
\end{equation}

The following lemma shows that spacelike complements are always
obtained as intersections of open convex subsets of $\R^{1,d}$
with $\dS^d$. 

\begin{lem} \mlabel{lem:5.4}
  Let $\cO \subeq \dS^d$ be an open subset. Then
  the convex subset 
\[ C := \{ y \in \R^{1,d} \:  (\forall x \in \cO)\ \beta(x,y) > -1\} \]
satisfies 
\begin{equation}
  \label{eq:co'}
  \cO' = C^\circ \cap \dS^d, 
\end{equation}
where $C^\circ$ is the interior of $C$ in $\R^{1,d}$. 
\end{lem}

\begin{prf} We have $\cO' = (C \cap \dS^d)^\circ$,
  so that $C^\circ \cap \dS^d \subeq \cO'$ holds trivially.
  For the converse inclusion, 
let $x \in \cO'$. Then there exists an open neighborhood $U$ of $x$ in
$\dS^d$ with $x \in U \subeq C$. Then $x$ is contained in the interior
of the convex hull $\conv(U)$  in Minkowski space, hence also in $C^\circ$, and
this proves the other inclusion.
\end{prf}

\begin{lem} For wedge domains $W \subeq \dS^d$, we have
  $W' = -  W$.
\end{lem}

\begin{prf} As $(g.W)' = g.W'$ for $g \in G$, it
  suffices to prove the assertion for the canonical
  wedge $\WdS$ (see \eqref{eq:CanWedge}). 
From $W_R' = - W_R$ in Minkowski space, we derive 
$- \WdS \subeq (\WdS)'.$ 
  For the converse, we assume that $y \in (\WdS)'$.
  Then
  \[ -1 < \beta(y, \cosh(t) \be_1 + \sinh(t) \be_0)
    = \sinh(t) y_0 - \cosh(t) y_1 \quad \mbox{ for }\quad  t \in \R.\]
  Dividing  by $\cosh(t)$, we obtain for $t \to \pm \infty$
  the inequalities
  $\pm y_0 \geq y_1$, i.e., $y_1 \leq - |y_0|$, which means that
  $y \in - \oline{W_R}$. 
  For the set $C$ in Lemma~\ref{lem:5.4}, this means that
  $C \subeq - \oline{W_R}$, so that $C^\circ \subeq - W_R$,
  and therefore $(\WdS)' \subeq \dS^d \cap - W_R = - \WdS$.
\end{prf}

Clearly, any pair of open subsets $\cO_1, \cO_2\subeq \dS^d$
for which there exists an affine wedge region
$x + g.W_R$ in Minkowski space with
\[ \cO_1 \subeq x + g.W_R \quad \mbox{ and }\quad 
\cO_2 \subeq (x + g.W_R)' = x - g.W_R \] 
are spacelike separated. In Minkowski space the converse is true:
If $\cO_1, \cO_2\subeq \R^{1,d}$ are convex and spacelike separated,
then there exists an affine wedge $W = x + g.W_R$, $g \in G$,
such that $\cO_1 \subeq W$ and $\cO_2 \subeq W' = x - g.W_R$
(\cite[Prop.~3.1]{TW97}). Here convexity is crucial,
as the pairs $(\cO,\cO')$, where $\cO$ is a double cone show. 

Note that not all intersections
$(x + g.W_R) \cap \dS^d$, $x \in \dS^d$, are wedge regions in $\dS^d$. 
The following discussion shows that, in de Sitter space,
there also exist open subsets $\cO_1$ and $\cO_2$ with
$\cO_1 \subeq \cO_2'$, for which there exists no  wedge 
  region $W \subeq \dS^d$ with $\cO_1 \subeq W$ and $\cO_2 \subeq W'$.

\begin{ex} In $\dS^2$, we consider the double cone
$\cO$ generated by
\[ p_0 := \be_1 \quad \mbox{ and } \quad p_1 := (\sinh t, \cosh t, 0),
  \quad \mbox{ for some } \quad t > 0,\] 
i.e.,
\[ \cO = \dS^2 \cap (p_0  + V_+) \cap (p_1 - V_+).\]
We claim that there exists no wedge region $W \subeq \dS^2$ containing
\[ \cO'
  = \{ x \in \dS^2 \: x_1 < 1, \cosh(t) x_1 - \sinh(t) x_0 < 1\}.\]
For $\eps > 0$, we consider the curves
\[ \beta^\pm(s) := (s,1-\eps, \pm \sqrt{1 - (1-\eps)^2 + s^2}) \in \dS^2.\] 
As
\[ \cosh(t) (1-\eps) - \sinh(t) s < 1 \]
for $s$ sufficiently large, we have $\beta^\pm(s) \in \cO'$ for $s \> 0$.
The asymptotic directions of $\beta^\pm(s)$ for $s \to \infty$ are
\[ \lim_{s \to \infty} \frac{\beta(s)}{s} = \be_0 \pm \be_2.\]
If $\cO'$ is contained in some wedge region, then
$W = W_M \cap \dS^2$ for a linear wedge region $W_M$ in Minkowski space,
i.e., $W_M$ is an open convex spacelike cone. We then obtain
$\be_0 \pm \be_2 \in \oline{W_M}$, and hence
$\be_0 \in W_M$, contradicting that $W_M$ is spacelike.
\end{ex}

In the preceding example, the problem is that $\cO'$ is not contained
in any wedge domain, i.e., the convex cone $\conv(\R_+ \cO')$ is
not spacelike. This is clearly necessary for a wedge $W$ to exist
such that $\cO_1\subeq W$ and $\cO_2 \subeq W'$. 

\begin{lem} {\rm(Separation by wedge regions)}
  \mlabel{lem:sepwedges}
  For two open subsets $\cO_1, \cO_2 \subeq \dS^d$,
  the following are equivalent:
  \begin{itemize}
  \item[\rm(a)] The existence of a wedge region $W\subeq \dS^d$ with
  $\cO_1 \subeq W$ and $\cO_2 \subeq W'$. 
\item[\rm(b)] The open convex cones $C(\cO_j):= \conv(\R_+ \cO_j)$,
    $j = 1,2$, are spacelike separated in Minkowski space~$\R^{1,d}$. 
  \end{itemize}
\end{lem}

If two open convex cones $C(\cO_j) \subeq \R^{1,d}$ are spacelike separated,
then the cone $C(\cO_1) - C(\cO_2)$ is spacelike, and this implies that
$C(\cO_j)$ are both spacelike. 

Note that the condition that $\conv(\R \cO)$ is spacelike in $\R^{1,d}$
is a non-trivial requirement that restricts the ``size'' of the open subset
$\cO \subeq \dS^d$. It is equivalent to $\cO$ being contained in a
wedge region, because the wedge regions in $\R^{1,d}$ that are convex cones
are precisely the maximal spacelike cones. 

\begin{prf} (a) $\Rarrow$ (b):
   The existence of a wedge region $W$ with
   $\cO_1 \subeq W$ and $\cO_2 \subeq W'$
   implies $C(\cO_1) \subeq W$ and $C(\cO_2) \subeq -W = W'$,
   so that $C(\cO_j)$ are spacelike separated in $\R^{1,d}$. 

\nin (b) $\Rarrow$ (a):    
  If, conversely, $C(\cO_j)$ are spacelike separated in $\R^{1,d}$,
  then there exists a wedge region $W$ in $\R^{1,d}$ with
$C(\cO_1) \subeq W$ and $C(\cO_2)\subeq W'$ 
  (\cite[Prop.~3.1]{TW97}). Here $W$ is a translate $x + W_0$
  of a convex cone $W_0$,
  so that the fact that $C(\cO_j)$ are convex cones implies
  \[ \cO_1 \subeq C(\cO_1)\subeq W_0 \quad \mbox{ and } \quad \cO_2 \subeq
    C(\cO_2)\subeq  - W_0,\]
  so that $\cO_1$ and $\cO_2$ are separated by the wedge regions
  $\pm W_0 \cap \dS^d$ in de Sitter space. 
\end{prf}

\subsection{The Locality Theorem}

We now prove the locality property for the net
$\sH^{\dS^d}_{\sE_H}(\cO)$, $\cO \subeq \dS^d$ open, associated to the 
$H$-invariant  real subspace $\sE_H^+ \subeq \cH^{-\infty}$ as 
defined in  {\eqref{eq:EHpm}.  

\begin{thm} {\rm(Locality Theorem)} The net $\sH^{\dS^d}_{\sE_H}$ is local
  in the sense that, for open subsets $\cO_1, \cO_2 \subeq \dS^d$ with
  $\cO_1 \subeq \cO_2'$, we have
  \begin{equation}
    \label{eq:loc-net}
    \sH^{\dS^d}_{\sE_H}(\cO_1) \subeq \sH^{\dS^d}_{\sE_H}(\cO_2)'.
  \end{equation}
\end{thm}

\begin{prf}
We consider the element
\[ g_0 := \diag(1,-1,-1,1,\ldots, 1) \in \SO_d(\R). \] 
It commutes with $\tau_h$ and satisfies $\Ad(g_0) h = - h$. Accordingly,
\[ g_0.W_{\dS^d}^+(h)
  = W_{\dS^d}^+(\Ad(g_0)h)
  = W_{\dS^d}^+(-h)
= - W_{\dS^d}^+(h).\] 
So covariance of the net $\sH^{\dS^d}_{\sE_H}$ implies with $W := W_{\dS^d}^+(h)
= W_R \cap \dS^d$ that 
\[ U(g_0) \sV = U(g_0) \sH^{\dS^d}_{\sE_H}(W) 
=  \sH^{\dS^d}_{\sE_H}(g_0W) =  \sH^{\dS^d}_{\sE_H}(W').\] 
On the other hand, $U(g_0)$ commutes with $J_\sV$ and 
$U(g_0) \Delta_\sV U(g_0)^{-1} = \Delta_{\sV}^{-1}$, so that 
$U(g_0)\sV = \sV'.$ This implies that 
\[ \sH^{\dS^d}_{\sE_H}(W') = \sH^{\dS^d}_{\sE_H}(W)'\]  
for $W = W_{\dS^d}^+(h)$. 

Let $W_G := q_{\dS^d}^{-1}(W)$ and $W'_G := q_{\dS^d}^{-1}(W') = g_0 W_G$. 
For $\alpha, \gamma \in \sE_H$, we consider the distribution on $G$, 
defined by 
\[ D_{\alpha,\gamma}(\phi) := \gamma(U^{-\infty}(\phi)\alpha).\] 
The corresponding left invariant distribution kernel 
on $G \times G$ is given  by 
\[ D_{\alpha,\gamma}^{(2)}(\phi \otimes \psi) 
  := D_{\alpha,\gamma}(\psi^* * \phi) \quad \mbox{ for } \quad
  \phi, \psi \in C^\infty_c(G,\R).\] 
Note that $\psi^* = \psi^\vee$ for real-valued functions.

For $\xi \in C^\infty_c(G,\C)$, we put 
\[ \supp_{\dS^d}(\xi) := q_{\dS^d}(\supp(\xi)) \quad \mbox{ for } \quad 
q_{\dS^d} \:G \to G/H, \quad g \mapsto gH.\] 
For $\supp_{\dS^d}(\phi) \subeq W$ and $\supp_{\dS^d}(\psi) \subeq W'$, we then have
\begin{align*}
 D_{\alpha,\gamma}^{(2)}(\phi \otimes \psi)
&= D_{\alpha,\gamma}(\psi^* * \phi)
  = \gamma(U^{-\infty}(\psi^* * \phi)\alpha)
  = \gamma(U^{\infty}(\psi^*) U^{-\infty}(\phi)\alpha)\\
&  =(U^{-\infty}(\psi)\gamma)(U^{-\infty}(\phi)\alpha)
  = \la U^{-\infty}(\phi)\alpha, U^{-\infty}(\psi)\gamma \ra 
  \in \la \sV, \sV' \ra \subeq \R.
\end{align*}
This means that the left invariant
distribution kernel $D_{\alpha,\gamma}^{(2)}$ on $G \times G$
is real-valued on $W_G \times W_G'$.

Proposition~\ref{prop:b.7} implies that the set of spacelike pairs
in $\dS^d \times \dS^d$ coincides with $G.(W \times W')$.
Hence its inverse image in $G \times G$ coincides with the set
$G.(W_G \times W_G')$, where $G$ acts on $G \times G$ by
$g.(g_1, g_2) = (gg_1, gg_2)$. The left invariance of
$D_{\alpha,\gamma}$ thus implies that this distribution kernel is real
on $G.(W_G \times W_G')$.

Now let $\cO_1, \cO_2 \subeq \dS^d$ be spacelike open subsets,
i.e., $\cO_1 \subeq \cO_2'$, which by Proposition~\ref{prop:b.7}
is equivalent to $\cO_1 \times \cO_2 \subeq G.(W \times W').$ 
For $\supp_{\dS^d}(\phi) \subeq \cO_1$ and 
$\supp_{\dS^d}(\psi) \subeq \cO_2$, we then obtain
\[ \la U^{-\infty}(\phi)\alpha, U^{-\infty}(\psi)\gamma \ra
  =  D_{\alpha,\gamma}(\psi^* * \phi)
  =  D_{\alpha,\gamma}^{(2)}(\phi \otimes \psi) \in \R,\]
and hence that $\sH^{\dS^d}_{\sE_H}(\cO_1) \subeq \sH^{\dS^d}_{\sE_H}(\cO_2)'.$
\end{prf}

\section{The case $d = 1$}
\mlabel{sec:d=1}

In this section we explain how the general theory developed
for de Sitter space in the preceding sections specifies for $d = 1$
to a ``one-parameter'' picture that links for a unitary
one-parameter group $(U_t)_{t \in \R}$ and a conjugation $J$, commuting
with $U$, elements in the space
$\cH^{-\infty}_{\rm KMS}$ of distribution vectors, satisfying the
KMS condition \eqref{eq:kms-spaces-intro} 
with elements of a subspace $\cH^J_{\rm temp}$, specified in
terms of the spectral measure $P$ of $U$.

Let $P$ be the uniquely determined spectral measure
on $\R$ for which
\[ U_t = \int_\R e^{itx}\, dP(x), \quad \mbox{ resp.} \quad
  U_t = e^{itA}, \ t \in \R,  \quad \mbox{ with } \quad
  A = \int_\R p\, dP(p).\] 
For $v \in \cH$, we thus obtain finite measures $P^v := \la v, P(\cdot) v\ra$,
and we define
\begin{equation}
  \label{eq:hjtemp}
 \cH^J_{\rm temp} := \{ v \in \cH^J \: e^{ \pi p}\, dP^v(p)\ 
  \mbox{ tempered}\}
  = \{ v \in \cH^J \: e^{-\pi p}\, dP^v(p)\
  \mbox{ tempered}\}.
\end{equation}
The equality of both spaces on the right follows from the symmetry
of the measures $P^v$, which is a consequence of $Jv = v$.
For the positive selfadjoint operator $\Delta := e^{-2\pi A}$, we have
$J\Delta J = \Delta^{-1}$, so that
$J\cD(\Delta^{1/4}) = \cD(\Delta^{-1/4})$ implies that
\[ \cD(\Delta^{1/4}) \cap \cH^J = \cD(\Delta^{-1/4}) \cap \cH^J 
= \Big\{ v \in \cH^J \: \int_\R e^{\pm \pi p}\, dP^v(p) < \infty\Big\}
  \subeq \cH^J_{\rm temp}.\]

\begin{thm} For $v \in \cH^J \cap \bigcap_{|t| < \pi/2} \cD(e^{tA})$, 
the following are equivalent:
\begin{itemize}
\item[\rm(a)] $v \in \cH^J_{\rm temp}$. 
\item[\rm(b)] The limits $\beta^{\pm}(v) := \lim_{t \to \pm\pi/2} e^{-tA} v$
  exist in $\cH^{-\infty}(U)$.
\item[\rm(c)] There exist $C,N > 0$ such that 
  $\|e^{\pm tA}v\|^2 \leq C \big(\frac{\pi}{2}-|t|\big)^{-N}$
  for $|t| < \pi/2$.   
\end{itemize}
\end{thm}

\begin{prf} (a) $\Leftrightarrow$ (b): 
From \cite[Prop.~4]{FNO23}, we recall that the temperedness of the
measure $\nu_v$, given by  $d\nu_v(p) := e^{\pi p}\, dP^v(p)$
is equivalent to the existence of $C, N > 0$ with 
\[ \int_\R e^{(\pi - t)p}\, dP^v(p) \leq C t^{-N}
  \quad \mbox{ for }  \quad 0 \leq t < \pi.\] 
Further, \cite[Lemma~10.7]{NO15} shows that this condition is equivalent
to the function $e^{\pi p/2}$ to define a distribution vector
for the canonical multiplication representation on $L^2(\R, P^v)$.
This representation is equivalent to the subrepresentation of $(U,\cH)$, 
generated by~$v$, where the constant function $1$ corresponds to~$v$.

\nin (b) $\Rarrow$ (c): If $\lim_{t \to \pi/2} e^{tA} v$ exist in $\cH^{-\infty}(U)$, 
then   \cite[Lemma~10.7]{NO15}, applied to the cyclic subrepresentation 
generated by $v$, implies that the measure $\nu_v$ 
is tempered. Then the argument from above implies the existence
of $C, N > 0$ with 
\begin{equation}
  \label{eq:exp-esti}
 \|e^{tA}v\|^2 = \int_\R e^{2tx}\, dP^v(x)  
  \leq C \Big(\frac{\pi}{2}-t\Big)^{-N}
  \quad \mbox{ for }  \quad
  |t| < \pi/2.
\end{equation}
If $\lim_{t \to -\pi/2} e^{tA} v$ also exists in $\cH^{-\infty}(U)$,
then the same argument  applies again and we obtain (c). 

\nin (c) $\Rarrow$ (a): With the leftmost equality in
\eqref{eq:exp-esti}, we see that (c) implies that
the measures $d\nu_v(x) := e^{\pm \pi x}\, dP^v(x)$ are tempered
(\cite[Prop.~4]{FNO23}). 
  Here we use that the measure $P^v$ is symmetric because $Jv = v$.
\end{prf}

\begin{prop} The map $\beta^+$ defines a bijection 
$\beta^+ \: \cH^J_{\rm temp} \to  \cH^{-\infty}_{\rm KMS}.$ 
\end{prop}

\begin{prf} (a) First we show that  $\beta^+(v) \in \cH^{-\infty}_{\rm KMS}$.
To this end, note that, for a real-valued test function
  $\phi \in C^\infty_c(\R,\R)$, we have $J U(\phi) = U(\phi)J$.
  For $v \in \cH^J_{\rm temp}$ we therefore have
  $w := U(\phi)v \in \cH^J$. Moreover,
  \[ dP^w(x) = |\hat\phi(x)|^2 dP^v(x) \quad \mbox{ with } \quad
    \hat\phi(x) = \int_\R e^{itx}\phi(t)\, dt,\]
  where $\hat\phi$ is a Schwartz function,
  which even implies that the measure
  \[ e^{\pi x}\, dP^w(x) = e^{\pi x} |\hat \phi(x)|^2\, dP^v(x) \]
  is finite, and thus $w \in \cD(\Delta^{1/4}) \cap \cH^J$. 
This implies that 
$ U^{-\infty}(\phi) \beta^+(v) = \beta^+(U(\phi)v)    \in \sV.$ 

 From \cite[Prop.~9]{FNO23}, we derive for $G = \R$ that
 \[ \cH^{-\infty}_{\rm KMS} = \{ \alpha \in \cH^{-\infty} \:
   (\forall \phi \in C^\infty_c(\R,\R))\,   U^{-\infty}(\phi) \alpha \in \sV \}.\]
 Hence the above argument implies that
 $\beta^+(v) \in \cH^{-\infty}_{\rm KMS}$.
 
 \nin (b) To see that $\beta^+$ is injective, we assume that $\beta^+(v) = 0$.
 Then the above argument implies that 
 $U(\phi)v \in \cH^J \cap \cD(\Delta^{1/4})$ vanishes for every
 $\phi \in C^\infty_c(\R,\R)$. Using an approximate identity in this space,
 $v = 0$ follows. 
 
 \nin (c) To see that $\beta^+$ is surjective, let
 $\gamma \in \cH^{-\infty}_{\rm KMS}$. Replacing $\cH$ by the
 cyclic subrepresentation generated by $\gamma$, resp., the subspace
 $U^{-\infty}(C^\infty_c(\R,\C))\gamma \subeq \cH$, we may w.l.o.g.\
 assume that $\cH = L^2(\R,\nu)$ for a positive Borel measure,
 where the constant function $1$ corresponds to $\gamma$.
 Hence the measure $\nu$ on $\R$ is tempered (\cite[Lemma~10.7]{NO15}).
 Then, for $z = x + iy \in \cS_\pi$, the analytic continuation of the orbit
 map of $\gamma = 1$ takes the form 
 \[ U^\gamma \: \oline{\cS_\pi} \to L^2(\R,\nu)^{-\infty},
   \quad U^\gamma(z)(p) = e^{izp} = e^{ixp} e^{-yp}.\]
 Therefore all measures $e^{-yp}\, d\nu(p), 0 \leq y \leq \pi$, are tempered.
 It follows in particular  that they are actually finite for $0 < y < p$. 
 Hence $v(p) := e^{- \pi p/2}$ is an $L^2$-function, and
 $v = U^\gamma(\pi i/2)$ implies that $Jv = v$. As a consequence, the measure
$dP^v(p) = e^{- \pi p}\, d\nu(p)$ 
 is finite and $e^{\pi p}\, dP^v(p) =\, d\nu(p)$ is tempered, so that
 $v \in \cH^J_{\rm temp}$. Therefore $\beta^+(v) = 1$ shows that
 $\beta^+$ is surjective.
\end{prf}

For $v, w\in \cH^J_{\rm temp}$, we consider the complex-valued measure
\[ P^{v,w}(E) := \la v, P(E) w \ra, \quad E \subeq \R.\]
Then
\begin{equation}
  \label{eq:pwv}
 \oline{P^{v,w}(E)} = \oline{\la v, P(E) w \ra}
 = \la w, P(E) v \ra = P^{w,v}(E)
\end{equation}
and  the relation 
  $J dP(E) J = dP(-E)$ implies that
  \begin{equation}
    \label{eq:pvw-symm}
 P^{v,w}(E) = \la Jv, P(E) Jw \ra 
    = \la Jv, J P(-E) w \ra 
    = \la P(-E) w, v \ra = P^{w,v}(-E) = \oline{P^{v,w}(-E)}.
  \end{equation}
In particular, the measures $P^{v,v}$ are symmetric and positive.

We obtain on the strip $\cS_{\pm \pi}$ the
holomorphic function
\[ \phi^{v,w}(z) :=\hat{P^{v,w}}(z) = \int_\R e^{izp}\, dP^{v,w}(p),\]
and the temperedness of the measures $e^{\pm \pi p}\, dP^{v,w}(p)$
implies that this function has boundary values that are
tempered distributions on $\pm\pi i + \R$. For $t \in \R$, we have
$\phi^{v,w}(t) =\la v, U_t w \ra.$ 
Hence
\[ \phi^{w,v}(-t)  = \oline{\phi^{v,w}(t)}
  = \la U_t w, v \ra
  =  \la U_t J w, J v \ra 
  =  \la J U_t w, J v \ra 
  =  \la v, U_t w \ra = \phi^{v,w}(t),\]
and therefore
\begin{equation}
  \label{eq:phivwoline}
 \oline{\phi^{v,w}(z)} = \phi^{w,v}(-\oline z) = \phi^{v,w}(\oline z)
 \quad \mbox{ for }\quad   z \in \cS_{\pm \pi}.
\end{equation}

For $\alpha := \beta^+(v)$ and $\gamma:= \beta^+(w)$
the distribution
\[ D_{\alpha,\gamma}(\xi) := \gamma(U^{-\infty}(\xi)\alpha) \]
can be represented by the boundary values of a holomorphic function
\begin{align*}
 D_{\alpha,\gamma}(x)
&  = \lim_{t \to \pi/2}  \la U_x e^{tA}v,  e^{tA}w \ra
                        = \lim_{t \to \pi/2}  \int_\R e^{(2t-ix)p}\, dP^{v,w}(p)
  = \phi^{v,w}(-\pi i -x)  = \phi^{w,v}(\pi i +x).                                   
\end{align*}

\begin{rem}
For $G = \SO_{1,1}(\R)_e \cong \R_+$ we have 
$\dS^1_+ = G.\be_1 \subeq \R^{1,1}$.
Identifying $\dS^1_+$ with $\R \cong \R_+ \subeq \C^\times$,
the other connected component $\dS^1_-$ of $\dS^1$ corresponds to
$\R + \pi i \subeq \C$, resp., to $\R_-\subeq \C^\times$, 
and $\dS^1_\C \cong \C^\times$ with $\Xi_\pm = \C_\pm$.

As a distribution on $\R$, the distribution 
\[ D_{v,w}(x)
  := \int_\R e^{ip(x - \pi i)}\, dP^{v,w}(p)
  = \int_\R e^{ipx} e^{\pi p}\, dP^{v,w}(p)  = \hat{\nu_{v,w}}(x) \]
is a boundary value of a holomorphic function on
$\C \setminus [0,\infty) \cong \cS_{2\pi}$
and the Fourier transform of the tempered measure 
\[ d\nu_{v,w}(p) = e^{\pi p}\, dP^{v,w}(p), \]
which extends by 
\[ D_{v,w}(x + iy)  = \int_\R e^{i(x+iy)p}\, d\nu_{v,w}(p)
  \quad \mbox{ for } \quad 0 < y < 2\pi\]
to the strip $\cS_{2\pi}$. 
\end{rem}

To identify the jump singularity of the corresponding
holomorphic function on the slit plane, we
consider the upper boundary values of the function $\hat{\nu_{v,w}}$ on
$\cS_{2\pi}$:
\[ D_{v,w}^-(x) = \lim_{y \to 0+} D_{v,w}(x + (2\pi -y)i)
  = \lim_{y \to 0+} \int_\R e^{ixp} e^{-(2\pi - y)p}\, d\nu_{v,w}(p)
  =  \int_\R e^{ixp} e^{-2\pi p}\, d\nu_{v,w}(p).\]  
From $d\nu_{v,w}(p) = e^{\pi p} \, dP^{v,w}(p)$ we derive 
\[ d\nu_{w,v}(-p) = e^{-\pi p} \, dP^{w,v}(-p) = e^{-\pi p} \, dP^{v,w}(p)
  = e^{-2\pi p}\, d\nu_{v,w}(p).\] 
Hence
\[ D_{v,w}^-(x)
  =  \int_\R e^{ixp} \, d\nu_{w,v}(-p) 
  =  \int_\R e^{-ixp} \, d\nu_{w,v}(p) 
  =  \int_\R e^{-ixp} \, d\oline{\nu_{v,w}(p)} 
  = \oline{\hat{\nu_{v,w}}(x)}. \]
Therefore the jump on $\R_+ = e^{\R} = e^{\R + 2\pi i} \cong \dS^1_+$,
is given by
\begin{equation}
  \label{eq:jump-d-1}
 D_{v,w}^+ - D_{v,w}^-
 = 2 i \Im \hat{\nu_{v,w}},
\end{equation}
and this is a purely imaginary distribution on $\R$.

To make this formula more concrete, we write
\begin{align*}
2i  \Im \hat{\nu_{v,w}}(x) 
  &\ {\buildrel \eqref{eq:pwv} \over = }\
    \int_\R e^{\pi p} e^{ipx} \, dP^{v,w}(p)
    - \int_\R e^{\pi p} e^{-ipx} \, dP^{v,w}(-p)
= \int_\R e^{\pi p} e^{ipx} - e^{-\pi p} e^{ipx} \, dP^{v,w}(p)\\
&= \int_\R (e^{\pi p} - e^{-\pi p}) e^{ipx} \, dP^{v,w}(p)
=  2\int_\R  \sinh(\pi p) e^{ipx} \, dP^{v,w}(p).    
\end{align*}
For $v = w$, we get in particular 
\begin{align} \label{eq:jump2} 
 \Im \hat{\nu_{v,v}}(x) 
&=  -i \int_\R  \sinh(\pi p) e^{ipx} \, dP^{v,v}(p)     
=   \int_\R  \sinh(\pi p) \sin(px) \, dP^{v,v}(p),    
\end{align}
where we have used that this  expression is real to eliminate
$\cos(ps)$. Hence $\Im \hat{\nu_{v,v}}$ is odd and a superposition of
functions $\sin(px)$, $p \geq 0$. These are solutions of the initial
value problem 
\begin{equation}
  \label{eq:ivp}
  f'' + p^2 f = 0, \quad f(0) = 0, \quad f'(0) = p.
\end{equation}

\begin{rem} \mlabel{rem:6.5}
  Let $(U,\cH)$ be an antiunitary representation
of $G_{\tau_h} = \SO_{1,1}(\R)$, where $\tau_h = - \1$. 
We consider the standard subspace $\sV$, specified by
$J := U(-\1)$ and $\Delta := e^{2\pi i \partial U(h)}$.

The irreducible non-trivial antiunitary representations are realized on 
$\cH = \C^2$ for $m > 0$ by 
\[ U(\exp th) = e^{t \partial U(h)}
  =\pmat{
    \cos(m t) & -\sin(m t) \\ 
    \sin(m t) & \cos(m t)} \quad \mbox{ and } \quad
  J z = \oline z, z \in \C^2.\] 
Clearly, this function extends to all of $\C$, hence in particular
to the strip $\cS_{\pm \pi/2}$ with
\[e^{\mp\frac{\pi i}{2}\partial U(h)} 
  =\pmat{
    \cosh\big(m \frac{\pi}{2}\big) & i \sinh\big(m \frac{\pi}{2}\big) \\ 
-i \sinh\big(m \frac{\pi}{2}\big)& \cosh\big(m \frac{\pi}{2}\big)}.
\]
For a  unit vector  $v\in \cH^J = \R^2$ we obtain the matrix coefficient
\[ U^{v,v}(\exp z h) = \cos(mz)
  \quad \mbox{ for } \quad z \in \C.\] 
We thus obtain as boundary values of the analytic continuation to the
strip $\cS_{\pm \pi}$: 
\begin{equation}
  \label{eq:e.1b}
 U^{v,v}(\exp (x \pm \pi i)h)
  =  \cos(m(x \pm \pi i)) 
  =  \big(\cos(mx) \cosh(m\pi) \mp i \sin(mx) \sinh(m\pi)\big),
\end{equation}
which leads to the jump 
\[ U^{v,v}(\exp (x - \pi i)h) - U^{v,v}(\exp (x + \pi i)h) 
  = - 2 i  \sin(mx) \sinh(m\pi).\] 
\end{rem}

\section{One-dimensional $K$-types}
\mlabel{sec:8}

Our situation simplifies if there exists a
$K$-eigenvector $v$ in $\cH$. Then either $v$ is $K$-fixed
(a spherical vector),
or $K\cong \SO_d(\R)$ has a non-trivial character.
The latter can only happen if $K\not=\{e\}$ is 
abelian, i.e., if $d =2$.

Let $v_\mu \in \cH$ be a normalized $J$-fixed $K$-eigenvector with
$U(k)v_\mu = \chi_\mu(k) v_\mu$, where 
\[\chi_\mu(\exp t k_0)= e^{i \mu t},\quad \mu \in \R, t \in \R \]
with $k_0$ as in \eqref{eq:handk2}.
We consider the scalar-valued kernel
\[ Q(g_1, g_2) := \la U(g_2) v_\mu, U(g_1) v_\mu \ra
  = \phi_\mu(g_2^{-1} g_1) \quad \mbox{ for } \quad
  \phi_\mu(g) = \la v_\mu, U(g) v_\mu\ra.\]
As $\phi_\mu$ satisfies
\[ \phi_\mu(k_1 g k_2) = \chi_\mu(k_1) \phi_\mu(g) \chi_\mu(k_2) \quad \mbox{ for } \quad
  g \in G, k_1, k_2 \in K,\]
it is a {\it $\chi_\mu$-spherical function}.
For $\cE^J = \R v_\mu$, it corresponds to the function
$\phi_\cE$ introduced above in the general context. 
As the basis $k_0, h,h_1 \in \so_{1,2}(\R)$ satisfies 
$[k_0,h] = h_1$ and $[k_0,h_1] = -h$
(cf.\ \eqref{eq:handk2}), we have
$h_1 = e^{\frac{\pi}{2}\ad k_0} h.$ Therefore
\[ \phi_\mu(\exp t h_1) = \phi_\mu(\exp t h) \quad \mbox{ for } \quad t \in \R.\]
According to \cite[\S 5.1.3]{FNO23} and \cite[p.~384]{Sh94}, there exists for
$\mu \not=0$ a $\lambda \in \C$ for which
$\phi_\mu$ takes the form
\[\varphi_{\mu,\lambda} (\exp th)
  = \Big(\frac{1 + \cosh(t)}{2}\Big)^{-\mu/2} \cdot
  {}_2F_1 \Big( \frac{1-\mu}{2} +\lambda, \frac{1-\mu}{2} -\lambda;
  1; \frac{1- \cosh(t)}{2}\Big).\]
Note that $\phi_{\mu,\lambda} = \phi_{\mu,-\lambda},$ 
so that, by \cite[\S 5.1.2]{FNO23}, the formula also holds for $\mu = 0$.
This function of $t$ extends analytically to the strip
$\cS_{\pm \pi}$ which is mapped by
$\cosh$ to $\C \setminus (-\infty, -1]$.
As
\[ \beta(e^{zh}.i\be_0, \oline{i\be_0})  = \cosh(z), \]
and the kernel $Q(z,\oline w)$ on $\dS^2_\C$ is $G$-invariant,
we obtain a corresponding invariant kernel on $\Xi_+$ by 
\begin{equation}
  \label{eq:Q1}
 Q'(z,w) := \Big(\frac{1 + \beta(z, \oline w)}{2}\Big)^{-\mu/2} \cdot
  {}_2F_1 \Big( \frac{1-\mu}{2} +\lambda, \frac{1-\mu}{2}  -\lambda;
  1; \frac{1- \beta(z,\oline w)}{2}\Big).
\end{equation}
From its explicit form, it follows immediately, that this kernel  extends
to a sesquiholomorphic kernel on $(\dS^2_\C \times \dS^2_\C)^{\rm cut}$,
where $\beta(z,\oline w)$ avoids the set $(-\infty, -1]$
(cf.\ Remark~\ref{rem:cosh}).

We now turn to the special case where
$\mu = 0$, i.e., $v_\mu = v_0$ is fixed by $K$.
These situations also occur for general $d \geq 2$
and have been studied in detail in \cite{BM96} and \cite{NO20}.
Then the orbit map $U^{v_0}$ extends to a holomorphic map 
$\Xi_+ \to \cH$. This leads on $\Xi_+ \times \Xi_+$
to the sesquiholomorphic kernel
\[ Q(z, w) := \la U^{v_0}(w), U^{v_0}(z)  \ra\]
that is uniquely determined by the corresponding spherical
function $\phi_\lambda := \phi_{0,\lambda}$ on $G$.
Recall from \cite[Lemma~3.5]{NO20} that 
$\beta(\Xi_+ \times \Xi_-) = \C \setminus   (-\infty, -1].$ 
Accordingly, we have on $\Xi_+ \times \Xi_+$ for $m \geq 0$
and $d \geq 2$ the sesquiholomorphic $G$-invariant kernels 
\[ Q_\lambda(z,w) =
  \Psi_m(z,w)
  = \hgf \Big(\frac{d-1}{2}+\lambda,\frac{d-1}{2}-\lambda;\frac{d}{2};
  \frac{1-\beta(z,\oline w)}{2}\Big),\]
where 
 \[\lambda = \lambda_m := 
\begin{cases}
\sqrt{\left(\frac{d-1}{2}\right)^2-m^2} & \text{ for}\  0\leq m\le (d-1)/2  \\
i \sqrt{m^2- \left(\frac{d-1}{2}\right)^2} 
& \text{ for }  m\ge (d-1)/2
\end{cases}\]
and the kernel is normalized by
$\Psi_m(i\be_0, i\be_0) = 1$.
For $d = 2$, we get in particular the special case
of \eqref{eq:Q1} for $\mu = 0$. 

\begin{rem} (a) The $K$-biinvariant function
$\phi_{\lambda}(g) = Q_\lambda(g.i\be_0, i \be_0)$ 
  on $G = \SO_{1,d}(\R)_e$ is positive definite if and only if
  $\lambda \in i \R$ (spherical principal series) or
  $\lambda \in [-(d-1)/2,(d-1)/2]$ (spherical complementary series)
  (see Kostant's Theorem in \cite[Thm.~11]{FNO23}).

\nin (b) On $a_t = \exp(th)$ it is given by 
  \begin{equation}
    \label{eq:philambda-def}
 \phi_\lambda(a_t)
    =Q_\lambda(a_t.i\be_0, i \be_0)
  = \hgf \Big(\frac{d-1}{2}+\lambda,\frac{d-1}{2}-\lambda;\frac{d}{2};
  \frac{1-\cosh(t)}{2}\Big).
  \end{equation}
\end{rem}

\begin{exs} (a) For $d = 1$ and $\lambda = im\in i \R$ we have
\[  \hgf(\lambda, - \lambda; \shalf; \shalf(1 - \cos t)) = \cosh(i\lambda t)
  = \cos(\lambda t)\]
(cf.\ \cite[Rem.~4.14]{NO20}),
so that 
\[ \phi_\lambda(\exp th)
  =  \hgf\big(im, - im; \shalf; \shalf(1 - \cosh(t))\big) = \cos(m t).\]
Analytic continuation thus yields 
\[ \phi_\lambda(\exp (t \pm \pi i)th)
  = \cos(m(t \pm \pi i))
  = \cos(mt) \cosh(\pi m) \mp \sin(mt) i \sinh(m\pi),\]
so that the jump distribution is the odd function
\[ \phi_\lambda(\exp (t -\pi i)th) - \phi_\lambda(\exp (t +\pi i)th)
  = 2i \sin(mt) \sinh(m\pi),\]
which has no singularity (cf.~Remark~\ref{rem:6.5}). 

\nin (b) For $d = 2$ the function 
\[  \phi_\lambda(\exp th)
  = \hgf\Big(\frac{1}{2} + \lambda,\frac{1}{2} -\lambda; 1;
  \frac{1-\cosh t}{2}\Big) \]
describes the spherical function on $\bH^2$, 
and the corresponding kernel on $\Xi_+$ is 
\[  Q_\lambda(z,w) = \hgf\Big(\frac{1}{2} + \lambda,\frac{1}{2} -\lambda; 1;
  \frac{1-\beta(z,\oline w)}{2}\Big), \]
so that
\[  Q(z,\be_1) = \hgf\Big(\frac{1}{2} + \lambda,\frac{1}{2} -\lambda; 1;
  \frac{1+ z_1}{2}\Big), \]
which is singular for $z_1 \geq 1$.
Note that, for $m = 0$, resp., $\lambda = \shalf$,
  this is the constant function~$1$.

\nin (c) For $d = 3$ we use  the relation
\[ \hgf\Big(k+\lambda, k - \lambda; k + \frac{1}{2}; \frac{1-\cos t}{2}\Big)
  = \frac{\big(\shalf\big)_k}{(\lambda)_k (-\lambda)_k}
  \Big(\frac{2}{\sin t} \frac{d}{dt}\Big)^k
  \hgf\Big(\lambda, - \lambda;  \frac{1}{2}; \frac{1-\cos t}{2}\Big)  \]
(cf.\ \cite[Rem.~4.14]{NO20}) to obtain 
\begin{align*}
 \hgf\Big(1+\lambda, 1 - \lambda; 1 + \frac{1}{2}; \frac{1-\cos t}{2}\Big)
& = \frac{1}{-\lambda^2}  \Big(\frac{1}{\sin t} \frac{d}{dt}\Big)
  \hgf\Big(\lambda, - \lambda;  \frac{1}{2}; \frac{1-\cos t}{2}\Big) \\ 
& = \frac{1}{-\lambda^2}  \Big(\frac{1}{\sin t} \frac{d}{dt}\Big) \cos(\lambda t
    = \frac{1}{\lambda}  \frac{\sin(\lambda t)}{\sin t}    
\end{align*}
(see also \cite[15.4]{DLMF}). 
For the spherical function, this means that
\begin{align*}
  \phi_\lambda(\exp th)
=     \hgf\Big(1+\lambda, 1 - \lambda; \frac{3}{2}; \frac{1-\cosh t}{2}\Big)
  = \frac{1}{\lambda}  \frac{\sinh(\lambda t)}{\sinh t}.    
\end{align*}
For $m \geq 1$ this turns into 
\begin{align*}
  \phi_\lambda(\exp th)
  = \frac{1}{(m^2-1)^{1/2}}  \frac{\sin((m^2-1)^{1/2} t)}{\sinh t}
  \quad \mbox{ with } \quad
  \phi_0(\exp th)
  = \frac{t}{\sinh t},   \quad
  \end{align*}
and for $0 < m < 1$, we get
\begin{align*}
  \phi_\lambda(\exp th)
  = \frac{1}{(1-m^2)^{1/2}}  \frac{\sinh((1-m^2)^{1/2} t)}{\sinh t}.    
\end{align*}
\end{exs}

\section{Spherical functions and kernels} 
\mlabel{sec:sphfunc}

In the case of the trivial representation of $K\cong \SO_d(\R)$ 
the kernel $Q_\lambda$ is given by 
\begin{equation}
  \label{eq:qlambda}
  Q_\lambda (z,w) = {}_2F_1\Big(\frac{d-1}{2}+ \lambda , \frac{d-1}{2}-\lambda ; \frac{d}{2};
  \frac{1-\beta (z,\oline w)}{2}\Big),
\end{equation}
see \cite[Thm. 4.12]{NO20} or \cite[Thm. 7.4.11]{NO18}
and the preceding section.

To study the singularities of $Q_\lambda$,
we need information on the set of all pairs $(x,y)$ in $\dS^d$
with $\beta (x,y) \leq -1$, which is the complement of the cut domain
(cf.\ Lemma~\ref{lem:dsd-spacelike}). 
For that we recall for $x\in \dS^d$ the future and past sets for $x$
from \eqref{eq:I(x)}.

\begin{lemma}  The crown domains
    $\Xi_\pm $ from \eqref{eq:crown} satisfy:
\begin{itemize}
\item[\rm (1)] $\{\beta (z,\oline w)\: (z,w)\in \Xi_+\times \Xi_+ \}
  =\beta([\Xi_+,\Xi_-]) = \C\setminus (-\infty ,-1]$.
\item[\rm (2)]
  $ \beta (\dS^d \times \Xi_-)\cap \R = \beta (\dS^d \times \Xi_+)\cap \R = (-1,1)$.
\item[\rm (3)]  For $x,y \in \dS^d$, the inequality
    $\beta(x,y)\le -1 $ is equivalent to 
$y\in \overline{I(x)}$.
\item[\rm(4)]  $\dS^d = \partial \Xi_+ \cap \partial\Xi_-$.
\end{itemize}
\end{lemma}

\begin{proof} (1) is \cite[Lemma~3.5]{NO20}.

  \nin  (2)  As $G$ acts transitively on $\dS^d$ and
  both  $\Xi_+$ and $\Xi_-$ are
  $G$-invariant, it suffices to consider $\beta(\be_1, z)$ for $z \in \Xi_+$.
  Then $z = x + i y$ with $y \in V_+$ and $\beta(x,x) - \beta(y,y) = -1$,
  $\beta(x,y) = 0$. Acting on $y$ with $G_{\be_1}$, we may w.l.o.g.\ assume that
  $y = y_0 \be_0$, so that we obtain $x_0 = 0$ and the elements 
  $(0,\bx) + i y_0 \be_0 \in \Xi_+$ correspond to those $\bx \in\R^d$ with
  $\bx^2 + y_0^2 = 1$. As $y_0 > 0$, the possible values of $x_1$ exhaust the
  open interval $(-1,1)$. 

\nin (3) follows from Lemma~\ref{lem:dsd-spacelike}. 

\nin (4) From \cite[Lem 3.7]{NO20} we obtain 
  \[\partial \Xi_\pm = \{x+iy\in \C^{d+1}\:  \beta (x,x) = -1,
    \beta(y,y)= 0 , \pm y_0\ge 0, \beta(x,y)=0\}.\]
This implies that 
$\partial \Xi_+ \cap \partial\Xi_-=\{x\in\R^{d+1}\: \beta (x,x) =-1\}=\dS^d .$
\end{proof}
 
\begin{prop} $Q_\lambda (z,w)$ extends to a sequiholomorphic kernel
  on $(\dS^d_\C \times \dS^d_\C)^\cut$, which contains in particular the subsets
  \[ (\Xi_+\cup \dS^d)\times \Xi_+, \quad
    \Xi_+ \times (\Xi_+\cup \dS^d) \quad \mbox{ and } \quad
    (\dS^d \times \dS^d)^\cut.\]
Furthermore, for $y\in\dS^d$ the singularities of the  function $x\mapsto Q_\lambda (x,y)$ are contained in $\overline{I(y)}$.
\end{prop}

\begin{proof} From \cite[Lem 3.7]{NO20} we obtain 
  \[\partial \Xi_\pm = \{x+iy\in \C^{d+1}\:  \beta (x,x) = -1,
    \beta(y,y)= 0 , \pm y_0\ge 0, \beta(x,y)=0\}.\]
This implies that 
$\partial \Xi_+ \cap \partial\Xi_-=\{x\in\R^{d+1}\: \beta (x,x) =-1\}=\dS^d .$
\end{proof}

The next theorem describes the analytic  continuation and boundary values of analytic functions as  distributions. For
manifolds we use local coordinates. In particular for our situation we can use
the exponential function $\Exp_y : T_y(\dS^d)\to \dS^d$ which is
a local analytic diffeomorphism.

\begin{theorem}\label{thm:bd}
Let $k\in \N$, let $X \subset \R^k$ be an open set and $\Gamma\subset \R^k$ an open convex cone. For $\gamma > 0$ let
\[ Z = \{ z \in \C^k : \Re z \in X, \Im z \in \Gamma, \|\Im z\| < \gamma\}.\]
If $\Theta$ is an analytic function in $Z$ such that 
\[ |\Theta(z)| \leq C \|\Im z\|^{-N}\]
for some $N$ and some constant $C >0$, then
\[\lim_{\Gamma\ni y\to 0} \Theta(\cdot +iy)
  = \Theta_0\quad \mbox{\rm exists in }\quad C^{-\infty}(X)\]
and is of order $N$.  
\end{theorem}

\begin{proof}  See \cite[Thms.~3.1.15, 8.4.8]{H90}.
\end{proof}

We now apply this to our special situation.
 We use   $\alpha = \frac{d-1}{2} + \lambda$, $\beta = \frac{d-1}{2} -\lambda$ and $\gamma =d/2 $ to get 
 \[\gamma - \alpha = \frac{1}{2} -\lambda, \quad \gamma - \beta = \frac{1}{2} +\lambda \quad
 \text{and}\quad\gamma - \alpha -\beta = \frac{2-d}{2}.\]
Then 
 \[\gamma - (\gamma -\alpha) - (\gamma - \beta)= -\gamma +\alpha +\beta
   = \frac{d-2}{2}\ge 0 \quad \mbox{ for } \quad d\ge 2,\]
 and we get with Theorem~\ref{thm:gamma-limit}(i)
   for every $0<\epsilon <\pi$ and $|\arg (1+\beta(z,\oline w ))|
   <\pi - \epsilon$ the limit relation 
 \begin{align} 
\left( \frac{1+\beta(z,\oline w)}{2}\right)^{\frac{d-2}{2}}Q_\lambda (z,w ) &=  {}_2F_1\left(\frac{1}{2}-\lambda, \frac{1}{2}+\lambda ; \frac{d}{2};
 \frac{1-\beta(z,\oline w)}{2}\right)\nonumber\\
 &
 \mapright{\beta(z,\oline w)\to -1} \ 
   \frac{\Gamma(\frac{d}{2}) \Gamma(\frac{d-2}{2})}
   {\Gamma (\frac{d-1}{2}+ \lambda )\Gamma (\frac{d-1}{2}-\lambda)} 
\quad \mbox{ for } \quad d > 2. \label{eq:ConstOne}
 \end{align}
 
For $d=2$ we have $\gamma = \alpha + \beta = d -1  =\frac{d}{2} = 1$,
and in this case Theorem~\ref{thm:gamma-limit}(iii)
with the restriction $|\arg (1+\beta(z,\oline w ))|
   <\pi - \epsilon$ as above leads to 
 \begin{equation}\label{eq:ConstTwo}
 \frac{Q_\lambda (z,w)}{-\log \Big(\frac{1+\beta(z,\oline w)}{2}\Big)} 
   \ \
\  \mapright{\beta(z,\oline w)\to -1} \ 
   \frac{1}{\Gamma (\frac{1}{2}+\lambda)
 \Gamma (\frac{1}{2} - \lambda )}.
 \end{equation}
 Hence Theorem~\ref{thm:gamma-limit}(ii) for $d>2$ and
 Theorem~\ref{thm:gamma-limit}(iv) for $d=2$ together with
 Theorem \ref{thm:bd} imply the following theorem. 
  
\begin{theorem} For $y\in \dS^d$ the function 
  $Q_\lambda(\cdot,y)$ extends to a distribution on $\dS^d$ of order
$\lceil\frac{d}{2}\rceil-1$ where $\lceil v\rceil$
   denotes the smallest integer greater or equal to $v$.
 \end{theorem}

Note that 
$\lceil\frac{d}{2}\rceil-1 = 0$ for $d= 1,2$.

 \begin{proof}  
   See Section 5 in \cite{FNO23} and \cite[Thm. 4.5]{OS23}.
 \end{proof}

More detailed information is given by \cite[Thm. 4.5]{OS23}.
 On  $\Xi_-=\{\oline w \: w\in\Xi_+\}$ we consider the kernel 
\[ \oQ_\lambda (z,w)=Q_\lambda (\oline w,\oline z)
  = \oline{Q_\lambda(\oline z, \oline w)} 
  ={Q_{\oline\lambda}(z, w)}, \]
 where we have used the explicit formula \eqref{eq:qlambda} 
  for $Q_\lambda$ and the relations
\[\oline{{}_2 F_1(\alpha,\beta;\gamma; z)}
  = {}_2 F_1(\oline\alpha, \oline\beta; \oline \gamma; \oline z)
  \quad \mbox{ and } \quad 
{}_2F_1(\alpha, \beta ; \gamma ;z) = {}_2F_1(\beta, \alpha ; \gamma ;z).\] 
We also write
\[z_t = i\cos (t)\be_0 + \sin (t)\be_1 = \exp (-ith)(i\be_0).\]
We refer to \cite{OS23}, in particular Theorem 4.3, Corollary 6.4, Theorem 1.1 for the proof of the following, where 
\[F(x \pm i0) = \lim_{ t\to 0\pm} F(x+it)\]
means the limit in the space of distributions.
For the following theorem we recall the sets
  $I^\pm(x)$ from \eqref{eq:I(x)} and note that
  $I^\pm (\be_1) =H\exp (\pm \R_+ h) .\be_1$.

\begin{theorem}\label{thm:psi}
  Assume that $d \geq 2$.
  Let $m = \sqrt{\big(\frac{d-1}{2}\big)^2-\lambda^2}$. Then, for
  $g \in \SO_{1,d}(\R)_e$, the limits
\[y\mapsto D^+_\lambda (g.\be_1, y) = \lim_{t\to \frac{\pi}{2}-} Q_\lambda (gz_t.y)
 \quad\text{and} \quad y\mapsto D^-_\lambda (g.\be_1, y) = \lim_{t\to \frac{\pi}{2}-} \oQ_\lambda (gz_t.y) \]
exist in the space of distributions on $\dS^d$. Let
\[ D^\pm_{\lambda,x} (y) = D^\pm_\lambda(x ,y) \quad \mbox{ and }  \quad
  D^\pm _\lambda = D^\pm_{\lambda,\be_1}.\] Then the following holds: 
\begin{itemize}
\item[\rm (1)] 
  The distributions $D^\pm_{\lambda,x}$ are given by analytic functions
on the domains $\dS^d \setminus \oline{I(x)}$ and
    $I^\pm(x)$ as follows  
 \begin{align*}
        D^+_{\lambda,x } (y) = \begin{cases}
       {}_2F_1\big( \drho+ \lambda, \drho-\lambda; \frac{d}{2}; \frac{1- \beta(x,y)}{2}\big) &\text{if } y \notin \overline{I(x)},\\
       {}_2F_1\big( \drho+ \lambda, \drho-\lambda; \frac{d}{2}; \frac{1-\beta(x,y)}{2} - i0\big)  &\text{if  } y  \in I^+(x), \\
    {}_2F_1\big( \drho+ \lambda, \drho-\lambda; \frac{d}{2}; \frac{1-\beta(x,y)}{2} + i0\big)  &\text{if  }  y  \in I^-(x);
          \end{cases}\\
       D^-_{\lambda,x}(y ) = \begin{cases}
       {}_2F_1\big ( \drho+ \lambda, \drho-\lambda; \frac{d}{2}; \frac{1- \beta(x,y)}{2}\big ) &\text{if } y \notin \overline{I(x)},\\
       {}_2F_1\big(\drho+ \lambda, \drho-\lambda; \frac{d}{2}; \frac{1 - \beta(x,y)}{2} + i0\big)  &\text{if } y  \in I^+(x), \\
    {}_2F_1\big( \drho+ \lambda, \drho-\lambda; \frac{d}{2}; \frac{1- \beta(x,y)}{2} - i0\big)  &\text{if } y  \in I^-(x). 
          \end{cases}  
    \end{align*} 
 \item[\rm (2)] $(\Delta - m^2) D^+_{\lambda, x}  = (\Delta -m^2)D^-_{\lambda,x} =0.$ 
\item[\rm (3)] The distributions $D^+_{\lambda ,x}$ and $ D^-_{\lambda ,x}$ are linearly independent and
$\{D^+_\lambda,D^-_\lambda\}$ is a basis for the space of $H$-invariant distributions satisfying the
differential equation in {\rm(2)}.
\item[\rm(4)] The distributions $D^\pm_{\lambda,x}$ 
  do not vanish on any non-empty open subset of $\dS^d$. 
 \end{itemize}
\end{theorem}

\begin{proof} (1) We can assume that $g=e$. Then
\[ \frac{1- \beta(z_t , y)}{2} = \frac{1 + y_1\sin (t)}{2} - i\frac{y_0\cos (t)}{2}\]
and then (1) follows from the definition of $F(\cdot \pm i0)$.

\nin (2) is  \cite[Lem 2.10]{NO20}. 

\nin (3) The first part of (3) follows from (1). For $H$-invariant distributions on $\dS^d$ the differential
equation in (2) reduces to the hypergeometrical differential equation on 
$\R\simeq A$ in the same way as \cite[p.~484]{He84}
by replacing $K$ by $H$. The dimension of the solution
space is $2$ and  hence the second part of (3) follows.

\nin (4) is \cite[Cor. 6.4]{OS23} where it is proved using the 
wave front set of the distributions.
\end{proof}

\begin{corollary} \label{cor:Diff} Writing
  $F (z) = {}_2 F_1 (\drho + \lambda ,\drho- \lambda ; \frac{d}{2}; z)$, we have:
\[        D^+_{\lambda,x } (y)  - D^-_{\lambda ,x} (y ) = \begin{cases} 0 & \text{if } y\not\in \overline{I(x)},\\ 
       F\Big(  \frac{1-\beta(x,y)}{2} - i0\Big)    - F\Big(\frac{1- \beta(x,y)}{2} +i0\Big)& \text{if } y\in I^+ (x) , \\
    -\Big(F\Big(  \frac{1-\beta(x,y)}{2} - i0\Big)   - F\Big(\frac{1- \beta(x,y)}{2} +i0\Big) \Big)& \text{if } y\in I^- (x)
          \end{cases}\]
\end{corollary}
         
\subsection{The jump singularities} 

In this section we give more detailed expressions for the
right hand side of Corollary \ref{cor:Diff}. For that we need
to discuss the case where $d$ is even or odd separately. We also refer
to \cite[Thm.~A.1]{OS23}:

\begin{theorem}\label{thm:Main1} If {\bf $d \geq 3$ is odd} and we put
\[ c_\lambda
 =  2i  (-1)^{\frac{d-1}{2}}\frac{\Gamma(d /2)\Gamma((d-2)/2)}{\Gamma(\drho + \lambda)\Gamma(\drho-\lambda)},
 \quad a=\frac{1}{2}-\lambda, \quad\text{and} \quad  b= \frac{1}{2}+\lambda, \]
then 
 \[ \big( D^+_{\lambda , x}    - D^-_{\lambda ,x} \big)(y)  =
         c_\lambda \begin{cases}
            0 &\text{if } y \notin \overline{I(x)}\\
           \big(-\frac{1+\beta(x,y)}{2}\big)^{\frac{2-d}{2}}
            {}_2F_1\big(a,b ; \frac{4-d}{2} ;\frac{1+\beta(x,y)}{2}\big) 
          &\text{if } y \in I^+(x) \\
          -\big(-\frac{1+\beta(x,y)}{2}\big)^{\frac{2-d}{2}}
          {}_2F_1\big(a,b ;\frac{4-d}{2}  ;\frac{1+\beta(x,y)}{2}\big)
          &\text{if } y \in I^-(x).
        \end{cases}\]
\end{theorem}

\begin{proof} Let $\alpha = \frac{d-1}{2} +\lambda$, $\beta = \frac{d -1}{2} -\lambda$
  and $\gamma = \frac{d}{2}$, so that
    $\gamma - \alpha - \beta = 1 - d/2 \not\in \Z$.
By Theorem~\ref{thm:gamma-limit}(iv),
  we get  
\begin{align*}
&  \lim_{\varepsilon\searrow0}\left({_2F_1}(\alpha ,\beta ;\gamma ;t+i\varepsilon)-{_2F_1}(\alpha ,\beta ;\gamma ;t-i\varepsilon)\right)\\
  &=C(\lambda )\cdot {_2F_1}\Big(\frac{1}{2} -\lambda   ,
    \frac{1}{2} +\lambda    ;  \frac{4-d}{2}   ;1-t\Big)\lim_{\varepsilon\searrow0}
	\left((1-t-i\varepsilon)^{\frac{2-d}{2}}   -(1-t+i\varepsilon)^{ \frac{2-d}{2}  }  \right).
\end{align*}
To compute the limit, observe that for $\mu \in\C$   and $t > 1$:
\[ (1-t\pm i\varepsilon)^\mu  = |1-t\pm i\varepsilon|^\mu  e^{i\mu \arg(1-t\pm i\varepsilon)} \to |1-t|^\mu  e^{\pm i\pi\mu }=
(t-1)^\mu e^{\pm i\pi\mu}, \]
so we obtain
\begin{align*}
&  \lim_{\varepsilon\searrow0}\left( {_2F_1}(\alpha ,\beta;\gamma ;t+i\varepsilon)-{_2F_1}(\alpha ,\beta ;\gamma ;t-i\varepsilon)\right)\\
  & =2iC(\lambda ) \cdot {_2F_1}\Big(\frac{1}{2} -\lambda , \frac{1}{2}
    +\lambda ; \frac{4-d}{2} ;1-t\Big)(t-1)^{\frac{2-d}{2} }
    \underbrace{\sin\Big(\frac{d-2}{2} \pi\Big)}_{= (-1)^{\frac{d+1}{2}}}.
\end{align*}
The claim now follows  with Corollary~\ref{cor:Diff}
  with $t = \frac{1-\beta(x,y)}{2} > 1$.
\end{proof}

 \begin{theorem}\label{thm:Main2} If $d  \geq 2$ is {\bf even}, 
   then we put 
   \[c_\lambda
     = (-1)^{\frac{d}{2}}  \frac{2\pi i}{ \Gamma(\frac{d-1}{2} + \lambda)
       \Gamma(\frac{d-1}{2} - \lambda)},
   \]
and then   
 \[
  \big(D^+_{\lambda , x}    - D^-_{\lambda ,x}\big)(y)  =
         c_\lambda \begin{cases}
            0 &\text{if } y \notin \overline{I(x)}\\
            {}_2F_1\left(\drho + \lambda,\drho-\lambda; \frac{d}{2} ;\frac{1+\beta(x,y)}{2}\right) 
          &\text{if } y \in I^+(x) \\
     -  {}_2F_1\left(\drho + \lambda,\drho-\lambda; \frac{d}{2} ;\frac{1+\beta(x,y)}{2}\right) &\text{if } y \in I^-(x).
        \end{cases}\] 
 \end{theorem}

 \begin{prf}
We apply Lemma~\ref{lem:jump-hypergeo}(i)
with $n = \frac{d}{2}- 1 = \frac{d-1}{2}-\shalf$ to obtain with
\[ \alpha = \frac{d-1}{2} + \lambda, \quad 
  \beta = \frac{d-1}{2} - \lambda \quad
  n = \frac{d}{2} - 1 \]
that   
 \begin{equation}\label{eq:tr3}
   {}_2F_1\big(\drho + \lambda  , \drho -\lambda ;
   {\textstyle\frac{d}{2}}; t\big)=
 (1-t)^{1- \frac{d}{2}} {}_2F_1\big(\shalf+\lambda,
 \shalf-\lambda; {\textstyle\frac{d}{2}}; t\big).
\end{equation}
Now we use Lemma~\ref{lem:jump-hypergeo}(ii)
with 
\[ \alpha' := \alpha - n = \shalf + \lambda , \quad 
  \beta' := \beta- n = \shalf -\lambda, \quad
  \alpha' + \beta' + n = \alpha + \beta - n = \frac{d}{2}\]
and 
\[ a(n)
  = \frac{\Gamma(\alpha'+\beta' + n)}{\Gamma(\alpha'+n)
  \Gamma(\beta'+n)}
  = \frac{\Gamma(\frac{d}{2})}{\Gamma(\alpha)
  \Gamma(\beta)}.
= \frac{\Gamma(n+1)}{\Gamma(\frac{d-1}{2} + \lambda)
  \Gamma(\frac{d-1}{2} - \lambda)}. 
\]
We thus obtain 
\begin{align*}
\lim_{\eps \to 0} \Big( {}_2F_1&\Big(\shalf+ \lambda  , \shalf-\lambda ;
  \frac{d}{2}; t + i\eps\Big) -  {}_2F_1\Big(\shalf+ \lambda,\shalf -\lambda ;
  \frac{d}{2}; t - i\eps\Big)\Big)\\
 &  =  a(n)\lim_{\eps\to 0} \sum_{k=0}^\infty \frac{(\alpha' +n)_k(\beta'+n)_k}{(n+k)!k! }
          \Big[\psi (k+1) + \psi (n+k+1)\\
         &\qquad  - \psi(\alpha' +n +k ) - \psi(\beta' +n+k)\Big]((1-(t+i\eps ))^k - (1-(t-i\eps ))^k)\\
&\qquad  -  a(n)(t-1)^n\lim_{\eps\to 0}\sum_{k=0}^\infty \frac{(\alpha' +n)_k(\beta'+n)_k}{(n+k)!k! }\times  \\
&\quad \times   (\log (1-(t+i\eps))
(1-(t+i\eps))^ k   -\log (1-(t-i\eps))(1-(t-i\eps))^k\\
   &= -  a(n)(t-1)^n
  \Big(\sum_{k=0}^\infty \frac{(\alpha' +n)_k(\beta'+n)_k}{(n+1)_k \Gamma(n+1)k! }(1-t)^k\Big) 
       \Big(\lim_{\eps\to 0}(\underbrace{\log (1-(t+i\eps))}_{\to\pi i}
     -   \underbrace{\log (1-(t-i\eps))}_{\to - \pi i}\Big)
\\
 &= -\frac{2\pi i\cdot  a(n)(t-1)^n}
 {\Gamma (n+1)}\cdot {}_2F_1(\alpha' +n,\beta' +n; n+1; 1-t) \\
 &= (-1)^{n+1}\frac{2\pi i (1-t)^n}{\Gamma(\frac{d-1}{2} + \lambda)
  \Gamma(\frac{d-1}{2} - \lambda)} \cdot 
   {}_2F_1(\alpha,\beta;   n+1; 1-t)  \\ 
 &= (-1)^{d/2}\frac{2\pi i (1-t)^n}{\Gamma(\frac{d-1}{2} + \lambda)
  \Gamma(\frac{d-1}{2} - \lambda)} \cdot  
   {}_2F_1\Big(\frac{d-1}{2} + \lambda, \frac{d-1}{2} - \lambda; \frac{d}{2};  1-t \Big) \\
 &= (1-t)^n c_\lambda \cdot 
   {}_2F_1\Big(\frac{d-1}{2} + \lambda, \frac{d-1}{2} - \lambda; \frac{d}{2};  1-t \Big).
\end{align*} 
Now the assertion follows with $1 - t = \frac{1 + \beta(x,y)}{2}$
  for $|1-t| < 1$, resp., for $|\beta(x,y) + 1| < 2$.
  By $G$-invariance of the kernel, it suffices to consider $x = \be_1$.
Then the distribution $ D_{\lambda, \be_1}^+ - D^-_{\lambda, \be_1}$
is $H$-invariant and satisfies the differential equation
\[ \Delta (D_{\lambda, \be_1}^+ - D^-_{\lambda, \be_1})
  =(\lambda^2-\left(\drho \right)^2)(D_{\lambda, \be_1}^+ - D_{\lambda, \be_1}^-)\]
on the open set $I^+(\be_1)$. In particular
$t\mapsto (D_{\lambda, \be_1}^+ - D_{\lambda, \be_1}) (ha_t.\be_1)$ satisfies the hypergeometric 
differential equation  and hence is analytic on $\R_+$. The same
holds for the function
\[ (1-t)^n c_\lambda \cdot   {}_2F_1\Big(\frac{d-1}{2} + \lambda, \frac{d-1}{2} - \lambda; \frac{d}{2}; 1-t\Big).\]
Hence those functions agree on $I^+(\be_1)$.
 \end{prf}

 Recall that we are only considering
 \[\lambda \in i\R_{\ge 0} \cup (0,(d-1)/2).\] 
 Combining Theorems \ref{thm:Main1}  and  \ref{thm:Main2}, 
we get the following theorem. 

\begin{theorem}{\rm(Huygens' Principle)} \label{Thm:HP} If  $d \geq 2$ and
  $\lambda \in i\R_{\ge 0} \cup (0,(d-1)/2)$, then the distribution
$D^+_{\lambda ,x} - D^-_{\lambda ,x}$ is supported on  $\partial \overline{I(x)}$ if and only if $d\ge 4$ is even and
$\lambda \in \left\{\frac{1+2k}{2}\: k = 0, 1,\ldots ,
 \lfloor\frac{d}{2} - 2\rfloor \right\}$.
\end{theorem}

\begin{proof}  For $\lambda \in i\R^\times$, we have $c_\lambda\not=0$,
    independently
  of the parity of $d$. 
For $\lambda \in (0, \frac{d-1}{2})$,
we have $c_\lambda \not= 0 $ if $d$ is odd. For $d $ even we have 
$c_\lambda = 0$ if and only if $\frac{1}{2}- \lambda\in -\N_0$ or
\[\lambda \in \Big(\frac{1}{2} +\N_0\Big)\cap \Big(0, \frac{d-1}{2}\Big)\]
which implies the statement. 
\end{proof}

\begin{remark} Note that  the values for $\lambda$ in Theorem \ref{Thm:HP}
are the same as the knot points in \cite[p. 7]{Ya13}.
\end{remark}
 
 \subsection{Evaluation of the constants}

 We now evaluate the constants that appear in the limit in \eqref{eq:ConstOne} and \eqref{eq:ConstTwo}. For the
 value of the $\Gamma$-function see \cite[5.4]{DLMF}. 

 \begin{lemma}  
   Let $\lambda = is  \in i\R_+$. Then   
\[\Gamma (\drho + \lambda ) \Gamma (\drho -\lambda) = |\Gamma (\drho + \lambda)|^2=\begin{cases}
  \pi \frac{\prod_{j=0}^{k-1}(s^2+ (j+1/2)^2))}{\cosh \pi s}   & \text{ if } d= 2(k+1) \text{ is even }\\
 \frac{\pi \prod_{j=0}^{k-1}  (s^2 +j^2)  }{s\sinh (\pi s)}  & \text{ if } d= 2k + 1 \text{ is odd. }
\end{cases} \]

\end{lemma} \begin{proof}  Assume first that $d= 2k+1$ is odd. Then $\drho = k$ and
  \[\Gamma (\drho \pm \lambda ) = (\pm \lambda  +k -1)
    \cdots (\pm \lambda) \Gamma ( \pm\lambda).\]
Hence with $\lambda = is$:
\begin{align*}
  \Gamma (\drho +\lambda )\Gamma (\drho  -\lambda)  & = |\lambda + k -1|^2\cdots |\lambda |^2 |\Gamma (\lambda )|^2
   = \frac{\pi}{s\sinh (\pi s)} \prod_{j = 0}^{k-1} (s^2 + j^2)  .
 \end{align*}
 For $d = 2(k+1)$ even, we likewise obtain with $\drho  = k +1/2$:
 \begin{align*}
   \Gamma (\drho +\lambda )\Gamma (\drho  -\lambda)  & = \prod_{j=0}^{k-1}(s^2+ (j+1/2)^2))  |\Gamma (\lambda +1/2)|^2
                                                     = \pi \frac{\prod_{j=0}^{k-1}(s^2+ (j+1/2)^2))}{\cosh \pi s}.
 \qedhere
 \end{align*}
\end{proof}

\begin{theorem} If $\lambda =is \in i\R_+$, then
 the spherical function from \eqref{eq:philambda-def} satisfies
\[\lim_{t\to  \pi-} \cos (t/2)^{ d-2}\varphi_\lambda (a_{it}) =
2^{k} \begin{cases} \cosh (\pi s)\frac{ k! (k-1)!} {\pi \prod_{j=0}^{k-1}(s^2+ (j+1/2)^2)) } & \text{ if } d = 2(k+1),  k \in \N \\
s \sinh (\pi s) \frac{ \prod_{j=1}^{k-1} (\frac{1}{2} + j) \prod_{j=1}^{k-2} (\frac{1}{2} + j)}{  \prod_{j=0}^{k-1}  (s^2 +j^2) } &  \text{ if } d = 2k+1, k\in \N.
\end{cases} \]
\end{theorem} 

\begin{proof}  First we apply \eqref{eq:ConstOne} with
  $z = a_{it}.i\be_0$ and $w = i\be_0$ to get with
  \[ \beta(z,\oline w) = \cosh(it) = \cos(t) \to -1 \]
    and \eqref{eq:philambda-def} that  
\begin{align*}
 \lim_{t \to \pi-} \cos(t/2)^{d-2} \phi_\lambda(a_{it})
  &= \lim_{t \to \pi-} \Big(\frac{1 + \cos t}{2}\Big)^{\frac{d-2}{2}}
    Q_\lambda(a_{it}.i\be_0, i \be_0)
=    \frac{\Gamma(\frac{d}{2}) \Gamma(\frac{d-2}{2})}
   {\Gamma (\frac{d-1}{2}+ \lambda )\Gamma (\frac{d-1}{2}-\lambda)}.
\end{align*}

It remains to evaluate this constant.
If $d = 2(k+1) $ is even, then $\Gamma (d/2) = \Gamma (k+1) = k!$ and
$\Gamma ((d-2)/2) = \Gamma (k) =(k-1)!$. If $d = 2k +1$ is odd, then 
\[\Gamma (d/2) = \Gamma (k+1/2)
  =  \Gamma (1/2) \prod_{j=1}^{k-1} \Big(\frac{1}{2} + j\Big)
  =\sqrt{\pi } \prod_{j=1}^{k-1} \Big(\frac{1}{2} + j\Big)\]
and
\[\Gamma ((d-2)/2) = \Gamma ((k-1)+1/2)
  = \Gamma (1/2)\prod_{j=1}^{k-2} \Big(\frac{1}{2} + j\Big)
  =\sqrt{\pi } \prod_{j=1}^{k-2} \Big(\frac{1}{2} + j\Big)\]
and the claim follows 
\end{proof}

\begin{corollary} If $\lambda \in i\R$, then
\[\cos(t/2)^{d-2} \varphi_\lambda (a_t)   \stackrel{|\lambda|\to \infty}{\longrightarrow } \infty.\]
\end{corollary}
 
\begin{theorem} The function $(0, \drho) \to \R_+ $, $\lambda \mapsto \lim_{t\to \pi^-} \cos (t/2)^{d-2} \varphi_\lambda (a_t)$,
is bounded and goes to zero as $\lambda \to \drho$.
\end{theorem}
\begin{proof} This follows as $\Gamma( \drho - \lambda)$ has a pole as $\lambda = \drho$.
\end{proof}

\subsection{Jumps at the cut}

In this section we discuss the jump
distributions 
$D^+_{\lambda, x} - D^-_{\lambda , x}$ on the regions $I^\pm(x)$. In view of the
relation  
$D_{\lambda ,g.x}^\pm (y)= 
D^\pm_{\lambda , x}(g^{-1}.y)$, 
we can assume that $x =\be_1$. Hence $y\mapsto D_{\lambda}^\pm := D_{\lambda ,\be_1}^\pm$.

We start by recalling the Legendre functions (see \cite[(5), (6), \S 3.2,
p.~122]{Er53}):

\begin{align}
Q^\mu_\nu (z) &= e^{\mu i \pi} 2^{-\nu -1} \pi^{1/2}\frac{\Gamma( \nu + \mu +1)}{\Gamma (\nu +3/2)}
z^{-\nu-\mu -1}(z^2-1)^{\frac{\mu}{2}}\label{eq:Q2} \\
&\quad \times  {}_2F_1\left(\frac{\nu + \mu }{2}+1, \frac{\nu + \mu +1}{2};\nu + \frac{3}{2}; \frac{1}{z^2}\right) \quad \nu \not\in -\frac{3}{2} -\N_0\nonumber\\
P^\mu_\nu (z) &=\frac{2^\mu }{\Gamma (1-\mu) (z^2-1)^{\mu/2}}{}_2F_1\left(1-\mu +\nu, -\mu -\nu; 1-\mu;\frac{1-z}{2}\right),
\mu \not\in \N .
\nonumber
\end{align} 
Note that $P^\mu_{-\nu -1 } = P^\mu_{\nu}$. We collect some of the facts here:

\begin{theorem}\label{thm:Qs}
The following holds:
\begin{itemize}
\item[\rm (i)] We have for $s>0$, $\Re (\nu + \mu +1)>0$ and $\Re \mu <\frac{1}{2}$  
\begin{align*}
  Q^\mu_\nu (\cosh s ) &= \frac{ e^{i\mu \pi}\sqrt{\pi}
    \Gamma (\nu +\mu +1)}  
    {2^\mu \Gamma (\nu -\mu +1)
    \Gamma (\mu +\frac{1}{2})}\sinh (s)^\mu     \times\\
&\quad \times \int_0^\infty (\cosh (s) + \sinh (s)\cosh (t))^{-\nu - \mu -1}\sinh (t)^{2\mu}\, dt
\end{align*}
\item[\rm (ii)] $e^{-i\mu \pi}Q^\mu_\nu (x\pm i0)=
e^{\pm \frac{ i\mu \pi }{2}}\left( Q^\mu_\nu (x) \mp \frac{\pi i}{2} P^\mu_\nu (x)\right)$.
\end{itemize}
\end{theorem}

\begin{proof}  (i) is \cite[(3),\S 3.7, p.~155]{Er53} and (ii)
  is \cite[(9), \S 3.4, p.~144]{Er53}.
\end{proof} 
On $G$, the function  
\begin{equation}\label{eq:pl}
  p_\lambda (g) = \beta(g.(\be_0+\be_1), -\be_1)^{-\lambda -\drho} 
\end{equation}
 is defined on the open set $\{g \in G\: g_{10} > -g_{11}\}$. Furthermore we have
 \[p_\lambda (h'gma_t n) = e^{(-\lambda - \drho)t}p_\lambda (g)
   \quad \mbox{ for } \quad
   h' \in H = G_{\be_1}, g \in G, m \in Z_K(A), n \in N, a_t = \exp(th).\]
If $\Re \lambda \le -\drho$ then $p_\lambda$ is continuous and hence defines a $H$-invariant distribution
vector in the principal series representation with parameter $-\lambda$.
  
 The first part of the following theorem follows from the fact that $p_\lambda$ is an eigenfunction. The second
 part is \cite[Thm.  10.1]{FHO95}
but we give the proof here.
\begin{theorem}\label{thm:FHO} The integral
\[\Psi_\lambda (g)=\int_{H}p_\lambda (gh)dh\]
converges for $\Re \lambda > \frac{3-d}{2}$ and we have
\begin{itemize}
\item[\rm i)] As a function on $\dS^d$, we have $\Delta \Psi_\lambda = \left(\lambda^2 - \frac{d-1}{2}\right)\Psi_\lambda$.
\item[\rm ii)] For 
\[c^+(\lambda) =2^{-\lambda + \frac{d-3}{2}}\frac{\Gamma \left(\drho \right) \Gamma\left(\lambda -\frac{d-3}{2}\right)}
{\Gamma (\lambda + 1)}  \]
we have
\begin{align*}
\Psi_\lambda (a_t. \be_1) &
=  \frac{e^{\frac{(2-d)\pi i}{2}} 2^{\frac{d-2}{2}} \Gamma \left(\frac{d-1}{2}\right)\Gamma\left(\lambda - \frac{d-3}{2}\right)}
{\sqrt{\pi}\Gamma\left(\lambda + \drho\right) \sinh (t)^{\frac{d}{2}-1}}
Q^{\frac{d}{2}-1}_{\lambda - \frac{1}{2}} (\cosh (t))\\
&= \frac{c^+(\lambda)}{\cosh(t)^{\lambda + \drho}}
 {}_2F_1\left(\frac{\lambda  + \drho +1}{2}, \frac{\lambda + \drho}{2}; 
1+\lambda ;\frac{1}{\cosh^2(t)}\right) . 
\end{align*}
\end{itemize}
\end{theorem}

\begin{proof} Let $x_h =E_{0,d-1} + E_{d-1,0}$
  and $\fb = \R x_h$ where $E_{ij}$ denotes the matrix with $1$ 
in row $i$ and column $j$ and otherwise
zero. Then $\fb$ is maximal abelian in $\fh\cap \Sym_{d+1}(\R)$.
Let $B=\exp \fb$ and $b_s= \exp (s x_h)$. Then we can normalize the invariant
measure on $H$ such that 
\[\int_H f(h)dh = \int_{K\cap H}\int_0^\infty \int_{K\cap H} f (k_1b_s k_2)
  \sinh (s)^{d-2}dk_2ds dk_1 , \quad f\in C_c (H) \]
(see \cite[Thm.~I.5.8]{He84}). 
For $k_1, k_2 \in K\cap H$, we have
\[p_\lambda (a_tk_1b_s k_2)= p_\lambda (a_t b_s )=(\cosh (t) + \sinh(t)\cosh(s))^{-\lambda - \drho}.\]
The claim now follows from Theorem~\ref{thm:Qs}(i) and \eqref{eq:Q2}.
\end{proof} 

The function
\[t\mapsto  {}_2F_1\left(\frac{\drho +\lambda}{2}, \frac{\drho -\lambda}{2}; \frac{d}{2};  -\sinh (t)^2\right)
 = {}_2F_1\left(\frac{\drho +\lambda}{2},
    \frac{\drho -\lambda}{2}; \frac{d}{2};  \frac{1- \cosh(2t)}{2}\right)\]
satisfies the same hypergeometric differential equation
\[\dfrac{d^2\varphi}{dt^2} +  (d-1)(\coth t)\dfrac{d\varphi}{dt} +(\left(\drho \right)^2 - \lambda^2)\varphi = 0\]
as the function $t\mapsto Q_\lambda  (a_t.\be_1,\be_1)$ (compare \cite[p. 283]{OP04} and \cite[(4.4)]{NO20}), see
also \cite[(9.6.1) p. 309]{LS65}). As both are
analytic on $\R$ and take the value $1$ at $t=0$ it follows that
\begin{equation}\label{eq:D}
Q_\lambda (ha_t.\be_1,\be_1)=  {}_2F_1\left(\frac{\drho +\lambda}{2}, \frac{\drho -\lambda}{2}; \frac{d}{2};  -\sinh (t)^2\right).
\end{equation}

\begin{lemma}\label{lem:Q} Let $a,b,c, \lambda \in \C$, $c, 1\pm \lambda\not\in -\N$ and $a-b\not\in \Z$. Then
\begin{align*}
  {}_2F_1(a,b;c;z)  & = \frac{\Gamma (c)\Gamma (b-a)}{\Gamma (c-a)\Gamma (b)}
   (1-z)^{-a} \cdot {}_2 F_1 \left(a, c-b; a+1-b;\frac{1}{1-z}\right)\\
&\quad
                                                                          + \frac{\Gamma (c)\Gamma (a-b)}{\Gamma (c-b)\Gamma (a)} (1-z)^{-b}
   {}_2 F_1\left(b, c-a; b+1-a;\frac{1}{1-z}\right) .
\end{align*}
\end{lemma} 

\begin{proof} This follows by combining \cite[(11), 2.9]{Er53}, \cite[(15), 2.9]{Er53} and \cite[(34), 2.9]{Er53}.
\end{proof}

\begin{theorem} For $x\in  H\exp (\R_+h )\be_1 = I^+(\be_1)$ and $\lambda \not\in \Z$ let
\[\Phi_\lambda (x ) =   \frac{1}{c^+(\lambda )}\Psi_\lambda (x)
\quad \text{and}\quad c(\lambda) = \frac{2^{\drho-\lambda  -1}}{\sqrt{\pi}}\frac{\Gamma \left(\frac{d}{2}\right)\Gamma (-\lambda )}{\Gamma\left( \drho -\lambda \right) }.
  \]
Then
\[Q_\lambda (x,\be_1) = c(\lambda ) \Phi_\lambda (x) + c(-\lambda ) \Phi_{-\lambda } (x)\quad \mbox{ for } \quad x\in I^+(\be_1) = H\exp(\R_+ h).\be_1.\]
\end{theorem}

\begin{proof} Let $a =\frac{\drho + \lambda}{2}$, $b= \frac{\drho - \lambda}{2}$ and $c= \frac{d}{2}$. 
Then $b-a = -\lambda$, $c-b = \frac{\drho +1+\lambda}{2}$ and $a-b +1=1+\lambda$. Then
\[\Gamma (c-a)\Gamma (b) =\Gamma \left(\frac{\drho - \lambda}{2} +\frac{1}{2}\right)
\Gamma \left(\frac{\drho -\lambda}{2}\right) = 2^{\lambda - \drho +1}\sqrt{\pi} \Gamma \left(\drho -\lambda\right)\]
by the relation
$\Gamma(z + \shalf)\Gamma(z) = 2^{-2z+1}\sqrt \pi \Gamma(2z)$ 
(\cite[(1.2.2), p. 19]{LS65}).
  The claim now
follows from \eqref{eq:D}, Theorem \ref{thm:FHO} and Lemma~\ref{lem:Q}.
\end{proof}

\section{Matrix-valued spherical functions for
  $\SO_{1,3}(\R)$} 
\mlabel{sec:msphfunc}

To give an idea of how matrix-valued spherical functions look like, we describe them explicitly in the case $d=3$. The description works also for the universal cover $\widetilde{G}=\mathrm{Spin}_{1,3}(\mathbb{R})_e\simeq\mathrm{SL}_2(\mathbb{C})$. The main statements make use of \cite{GPT02}.

By the Casselman Embedding Theorem (see e.g.~\cite[\S 3.8.3]{Wa88}), every irreducible unitary representation $(U,\mathcal{H})$ of $\mathrm{SL}_2(\mathbb{C})$ occurs, on the level of smooth vectors, as a subrepresentation of a (possibly non-unitary) principal series representation $(U_{\mu,\lambda},\mathcal{H}_{\mu,\lambda})$ (see Section~\ref{sec:PSofSL2C} below for the definition). Consequently, every spherical function $\varphi_{\mathcal E}$ associated to a $K$-type $\mathcal E\subseteq\mathcal H^\infty\subseteq\mathcal H_{\mu,\lambda}^\infty$ is of the following form
\begin{equation}
  \varphi_{\mathcal E}(g) = P_{\mathcal E}\circ U_{\mu,\lambda}(g)\circ I_{\mathcal E}\in\mathrm{GL}(\mathcal E) \quad \mbox{ for } \quad
  g\in\mathrm{SL}_2(\mathbb{C}),\label{eq:SphFctInPS}
\end{equation}
where $I_{\mathcal E}:\mathcal E\to\mathcal H_{\mu,\lambda}$ denotes the embedding and $P_{\mathcal E}:\mathcal H_{\mu,\lambda}\to\mathcal E$ the orthogonal projection onto the $K$-type $\mathcal E$. The right hand side of \eqref{eq:SphFctInPS} makes sense even if $\mathcal E$ is not a $K$-type of a unitarizable subrepresentation $\mathcal H^\infty$ of $\mathcal H_{\mu,\lambda}^\infty$. We provide an explicit expression for all spherical functions of the form \eqref{eq:SphFctInPS} in Theorem~\ref{thm:SphericalFunctionsSL2C}, and in Proposition~\ref{prop:PosDefSphFctSL2C} we specify which of them are positive definite,
i.e., belong to an irreducible unitary representation.

\subsection{Principal series representations of $\mathrm{SL}_2(\mathbb{C})$}\label{sec:PSofSL2C}

For $\mu\in\mathbb{Z}$ and $\lambda\in\mathbb{C}$ we define a character $\chi_{\mu,\lambda}$ of the parabolic subgroup $B$ of upper triangular matrices in $\mathrm{SL}_2(\mathbb{C})$ by
$$ \chi_{\mu,\lambda}:B\to\mathbb{C}^\times, \quad \begin{pmatrix}a&b\\0&a^{-1}\end{pmatrix} \mapsto \left(\frac{a}{|a|}\right)^\mu|a|^\lambda \qquad (a\in\mathbb{C}^\times,b\in\mathbb{C}). $$
The normalized parabolically induced representation $U_{\mu,\lambda}=\mathrm{Ind}_B^{\mathrm{SL}_2(\mathbb{C})}(\chi_{\mu,\lambda})$ is the left-regular representation of $\mathrm{SL}_2(\mathbb{C})$ on the Hilbert space
$$ \mathcal H_{\mu,\lambda} = \Bigg\{f:\mathrm{SL}_2(\mathbb{C})\to\mathbb{C}\mbox{ measurable }\Bigg|\begin{array}{l}f(gb)=\chi_{\mu,\lambda+1}(b)^{-1}f(g)\mbox{ for }g\in\mathrm{SL}_2(\mathbb{C}),b\in B,\\f|_{\mathrm{SU}_2(\mathbb{C})}\in L^2(\mathrm{SU}_2(\mathbb{C}))\end{array}\Bigg\}. $$
Note that $U_{\mu,\lambda}$ resp. any subrepresentation of $U_{\mu,\lambda}$ factors through $G=\mathrm{SO}_{1,3}(\mathbb{R})$ if and only if $\mu$ is even.

The restriction of $U_{\mu,\lambda}$ to the maximal compact subgroup $\mathrm{SU}_2(\mathbb{C})$ decomposes into a direct sum of irreducible representations of $\mathrm{SU}_2(\mathbb{C})$. The irreducible representations of $\mathrm{SU}_2(\mathbb{C})$ are of the form $(\pi_\ell,S^\ell(\mathbb{C}^2))$, $\ell\in\mathbb{Z}_{\geq0}$, where $\pi_\ell$ is the symmetric power of the standard representation of $\operatorname{SU}_2(\mathbb{C})$ on $\mathbb{C}^2$. By Frobenius reciprocity, the multiplicity of $\pi_\ell$ in $U_{\mu,\lambda}|_{\mathrm{SU}_2(\mathbb{C})}$ is equal to the multiplicity of $\chi_{\mu,\lambda+1}|_{B\cap\mathrm{SU}_2(\mathbb{C})}$ in $\pi_\ell|_{B\cap\mathrm{SU}_2(\mathbb{C})}$. Since $B\cap\mathrm{SU}_2(\mathbb{C})$ is a maximal torus in $\mathrm{SU}_2(\mathbb{C})$, it is a standard fact in the representation theory of $\mathrm{SU}_2(\mathbb{C})$ that the latter multiplicity is one if $\mu\in\{-\ell,-\ell+2,\ldots,\ell-2,\ell\}$ and zero otherwise. It follows that $U_{\mu,\lambda}$ contains the $K$-type $\pi_\ell$ if and only if $\mu\in\{-\ell,-\ell+2,\ldots,\ell-2,\ell\}$, and in this case it occurs with multiplicity one. We denote the spherical function $\varphi_{\mathcal E}$ corresponding to $\mathcal E=S^\ell(\mathbb{C}^2)$ by $\varphi_{\mu,\lambda}^\ell$.

\begin{lem}
	Let $\ell\in\mathbb{Z}_{\geq0}$ and $\mu\in\{-\ell,-\ell+2,\ldots,\ell-2,\ell\}$, $\lambda\in\mathbb{C}$. Then
	$$ \varphi_{\mu,\lambda}^\ell = \varphi_{-\mu,-\lambda}^\ell. $$
\end{lem}

\begin{proof}
	This follows either from \cite[Prop's~6.9 and 6.11]{GPT02} or from the fact that $U_{\mu,\lambda}$ and $U_{-\mu,-\lambda}$ have the same composition series.
\end{proof}

In view of the previous lemma,
it suffices to describe $\varphi_{\mu,\lambda}^\ell$ for $\operatorname{Re}(\lambda-\mu)\geq0$. Note that $\varphi_{\mu,\lambda}^\ell$ satisfies
\begin{equation}
  \varphi_{\mu,\lambda}^\ell(k_1gk_2) = \pi_\ell(k_1)\circ\varphi_{\mu,\lambda}^\ell(g)\circ\pi_\ell(k_2) \quad \mbox{ for } \quad
  g\in\mathrm{SL}_2(\mathbb{C}),k_1,k_2\in\mathrm{SU}_2(\mathbb{C}).\label{eq:EquivarianceSphFctSL2C}
\end{equation}

\subsection{Explicit matrix-valued spherical functions}

The Euler element $h\in\mathfrak{so}_{1,3}(\mathbb{R})$ defined in \eqref{eq:h} corresponds to
$$ h = \frac{1}{2}\begin{pmatrix}1&0\\0&-1\end{pmatrix}\in\mathfrak{sl}_2(\mathbb{C}). $$
Every element $g\in\mathrm{SL}_2(\mathbb{C})$ can be written as $g=k_1\exp(th)k_2$ with $k_1,k_2\in\mathrm{SU}_2(\mathbb{C})$ and $t\in\mathbb{R}$. Together with the equivariance property \eqref{eq:EquivarianceSphFctSL2C}, this implies that $\varphi_{\mu,\lambda}^\ell$ is uniquely determined by its values at $\exp(th)$, $t\in\mathbb{R}$.

Since the Euler element $h$ commutes with the Cartan subgroup $\mathrm{SU}_2(\mathbb{C})\cap B$ of $\mathrm{SU}_2(\mathbb{C})$, the operators $\varphi_{\mu,\lambda}^\ell(\exp(th))$ leave each of the corresponding weight
spaces in $\cE$ invariant. These weight spaces are one-dimensional, so $\varphi_{\mu,\lambda}^\ell(\exp(th))$ acts on each weight vector by a scalar. Following \cite[p.740]{GPT02}, we let $v_0,\ldots,v_\ell$ be a basis of $S^\ell(\mathbb{C}^2)$ such that
\begin{align}
	d\pi_\ell\begin{pmatrix}1&0\\0&-1\end{pmatrix}v_i &= (\ell-2i)v_i,\label{eq:PiEllOnH}\\
	d\pi_\ell\begin{pmatrix}0&1\\0&0\end{pmatrix}v_i &= (\ell-i+1)v_{i-1},\label{eq:PiEllOnE}\\
	d\pi_\ell\begin{pmatrix}0&0\\1&0\end{pmatrix}v_i &= (i+1)v_{i+1},\label{eq:PiEllOnF}
\end{align}
where we use the convention $v_{-1}=v_{\ell+1}=0$ and extend $d\pi_\ell$ to a representation of the complexification $\mathfrak{su}_2(\mathbb{C})_{\mathbb C}=\mathfrak{sl}_2(\mathbb{C})$. With respect to this basis, $\varphi_{\mu,\lambda}^\ell(\exp(th))$ is diagonal, and we write
\[  \varphi_{\mu,\lambda}^\ell(\exp(th))v_i
=\varphi_{\mu,\lambda}^{\ell,i}(t)\cdot v_i
\quad \mbox{ for }\quad 0\leq i\leq\ell \] 
for some scalar-valued analytic functions $\varphi_{\mu,\lambda}^{\ell,i}$ on $\mathbb{R}$.

\begin{theorem}\label{thm:SphericalFunctionsSL2C}	
	Let $\ell\in\mathbb{Z}_{\geq0}$ and $\mu\in\{-\ell,-\ell+2,\ldots,\ell-2,\ell\}$, $\lambda\in\mathbb{C}$. Assume that $\operatorname{Re}(\lambda-\mu)\geq0$. Then we have for $0\leq i\leq\ell$ and $t\in\mathbb{R}$ 
	$$ \varphi_{\mu,\lambda}^{\ell,i}(t) = e^{(\frac{\ell}{2}-i)t}\sum_{j=0}^\ell u_{ij}a_j(2\sinh t)^j{_2F_1}\left(\frac{\lambda-\mu+2}{2}+j,\frac{\mu-\lambda+2}{2}+j;j+\frac{3}{2};\frac{1 - \cosh(t)}{2}\right), $$
	where
	$$ u_{ij} = {_3F_2}\big(-j,-i,j+1;1,-\ell;1\big) $$
	and
	$$ a_j = \alpha_j\times\begin{cases}\frac{(\frac{\lambda-\mu-2}{2})_{j+1}}{(j+1)_{j+1}}&\mbox{if $\frac{\lambda-\mu}{2}\not\in\mathbb{Z}$ or $\frac{\lambda-\mu}{2}\in\mathbb{Z}$ and $0\leq j\leq\frac{\lambda-\mu}{2}-1$,}\\\frac{(-\frac{\lambda-\mu-2}{2})_{j+1}}{(j+1)_{j+1}}&\mbox{if $\frac{\lambda-\mu}{2}\in\mathbb{Z}$ and $\frac{\lambda-\mu}{2}\leq j\leq\ell$,}\end{cases} $$
	with $(\alpha_0,\ldots,\alpha_\ell)^\top$ the unique eigenvector of $U^{-1}LU$ of eigenvalue $-\mu(\lambda-\mu-2)-(\ell+\mu)(\ell-\mu+2)$ such that $a_0=1$ resp. $b_0=1$ with $U=(u_{ij})_{0\leq i,j\leq\ell}$ and
	$$ L = \sum_{i=0}^\ell \Big(4i(\ell-i+1)E_{i,i-1}+\big[(\ell-2i)(\lambda-\mu-2)-4i(\ell-i+1)\big]E_{i,i}\Big). $$
\end{theorem}

\begin{proof}
  In the notation of \cite[Prop.~6.11]{GPT02}, we have $\varphi_{\mu,\lambda}^\ell=\Phi_{(p,k)}$, where $p=\frac{1}{4}(\lambda-\mu+2)$ and $k=\frac{1}{2}(\mu+\ell)$. (Note that $\lambda$ corresponds to $iv$ and $\mu$ corresponds to $r$ in the definition of the principal series in
  \cite[p.~760]{GPT02}.) In \cite[p.~733]{GPT02}, the spherical function $\Phi_{(p,k)}$ is written as
	\begin{equation}
          \Phi_{(p,k)}(g) = H(g)\pi_\ell(g)
 \quad \mbox{ for } \quad g\in\mathrm{SL}_2(\mathbb{C}),\label{eq:SphFctSL2CEq1}
	\end{equation}
	where $\pi_\ell$ is extended to a holomorphic representation of $\mathrm{SL}_2(\mathbb{C})$ and $H:\mathrm{SL}_2(\mathbb{C})\to\operatorname{End}(S^\ell(\mathbb{C}^2))$. By \eqref{eq:PiEllOnH}, the values of $\pi_\ell$ on $\exp(th)$ are given by
	\begin{equation}
		\pi_\ell(\exp(th))v_i = e^{(\frac{\ell}{2}-i)t}v_i \qquad (0\leq i\leq\ell).\label{eq:SphFctSL2CEq2}
	\end{equation}
	To find an expression for $H(\exp(th))$, the authors introduce in \cite[p.~736, eq. (4.2), and p.~741]{GPT02} the notation
\begin{equation}
\tilde H(s) =  H(\exp(-\tfrac{1}{2}\log(s)h))
   = \operatorname{diag}(h_0(s),\ldots,h_\ell(s))
   \quad \mbox{ for } \quad s>0.\label{eq:SphFctSL2CEq3}
\end{equation}
(Note that this is abusing notation since $H$ is used both for a function on $\mathrm{SL}_2(\mathbb{C})$ and for a function on $(0,\infty)$.) They identify the diagonal matrix $\tilde H(s)$
with the corresponding vector
$(h_0(s),\ldots,h_\ell(s))^\top\in\mathbb{C}^{\ell+1}$.
In \cite[p.744]{GPT02} the vector $\tilde H(s)
=(h_0(s),\ldots,h_\ell(s))^\top$ is further written as
	\begin{equation}
	\tilde H(s) = UH^\vee(s),\label{eq:SphFctSL2CEq4}
	\end{equation}
	where $U=(u_{ij})_{0\leq i,j\leq\ell}$ is the matrix defined above and
	$$ H^\vee(s) = (h^\vee_0(s),\ldots,h^\vee_\ell(s))^\top. $$
	The functions $h^\vee_j(s)$ are found in \cite[Thm~6.8]{GPT02}. For the case $2p\not\in\mathbb{Z}$ they read
	$$ h^\vee_j(s) = a_j s^{1-p}(1-s)^j\cdot {_2F_1}\big(j+1,j+2-2p;2j+2;1-s\big) \qquad (0\leq j\leq\ell) $$
	with $a_j$ as defined in the statement of the theorem. To match this function with the claimed expression above, we apply the following quadratic transformation which follows from \cite[Thm~3.1.3]{AAR99} combined with Pfaff's transformation \cite[Thm~2.2.5]{AAR99}:
	$$ {_2F_1}(a,b;2a;z) = (1-z)^{-\frac{b}{2}}{_2F_1}\left(\frac{b}{2},a-\frac{b}{2};a+\frac{1}{2};\frac{z^2}{4(z-1)}\right). $$
	This yields
	$$ h^\vee_j(s) = a_js^{-\frac{j}{2}}(1-s)^j{_2F_1}\left(p+\frac{j}{2},1-p+\frac{j}{2};j+\frac{3}{2};-\frac{(1-s)^2}{4s}\right). $$
	We put $s=e^{-2t}$, then
	$$ \frac{(1-s)^2}{4s}=\sinh^2(t)=4\sinh^2(\tfrac{t}{2})(1+\sinh^2(\tfrac{t}{2})) $$
	and \cite[(3.1.3)]{AAR99} implies
	\begin{equation}
		h^\vee_j(e^{-2t}) = a_j(2\sinh(t))^j{_2F_1}\left(2p+j,2-2p+j;j+\frac{3}{2}; \frac{1 - \cosh(t)}{2}\right).\label{eq:SphFctSL2CEq5}
	\end{equation}
	Putting \eqref{eq:SphFctSL2CEq1}, \eqref{eq:SphFctSL2CEq2}, \eqref{eq:SphFctSL2CEq3}, \eqref{eq:SphFctSL2CEq4} and \eqref{eq:SphFctSL2CEq5} together shows the claimed formula for $2p\not\in\mathbb{Z}$. The case $2p\in\mathbb{Z}$ is handled in the same way.
\end{proof}

\begin{ex}
	\begin{enumerate}
		\item[\rm(a)] For $\ell=0$ we find $u_{00}=1$ and $a_0=1$, hence
		$$ \varphi_{0,\lambda}^{0,0}(t) = {_2F_1}\left(1+\frac{\lambda}{2},1-\frac{\lambda}{2};\frac{3}{2}; \frac{1 - \cosh(t)}{2}\right). $$
		This agrees with the function in the formula \eqref{eq:qlambda} for $d=3$.
		\item[\rm(b)] For $\ell=1$ we have $\mu=\pm1$ and $u_{00}=u_{10}=u_{01}=-u_{11}=1$. A short computation shows that for $\mu=-1$:
		$$ a_0=1 \qquad \mbox{and} \qquad a_1=-\frac{3+\lambda}{12}, $$
		and for $\mu=+1$:
		$$ a_0=1 \qquad \mbox{and} \qquad a_1=-\frac{3-\lambda}{12}. $$
		This implies, for $\mu=-1$ and $i=0,1$:
		\begin{multline*}
                  \varphi_{-1,\lambda}^{1,i}(t) = e^{(\frac{1}{2}-i)t}\Bigg({_2F_1}\left(\frac{3}{2}+\frac{\lambda}{2},\frac{1}{2}-\frac{\lambda}{2};\frac{3}{2};
                     \frac{1 - \cosh(t)}{2}\right)\\
			-(-1)^i\frac{3+\lambda}{6}\sinh(t){_2F_1}\left(\frac{5}{2}+\frac{\lambda}{2},\frac{3}{2}-\frac{\lambda}{2};\frac{5}{2}; \frac{1 - \cosh(t)}{2}\right)\Bigg),
		\end{multline*}
		and for $\mu=+1$ and $i=0,1$:
		\begin{multline*}
  \varphi_{+1,\lambda}^{1,i}(t) = e^{(\frac{1}{2}-i)t}\Bigg({_2F_1}\left(\frac{1}{2}+\frac{\lambda}{2},\frac{3}{2}-\frac{\lambda}{2};\frac{3}{2};
                    \frac{1 - \cosh(t)}{2}\right)\\
			-(-1)^i\frac{3-\lambda}{6}\sinh(t)\cdot {_2F_1}\left(\frac{3}{2}+\frac{\lambda}{2},\frac{5}{2}-\frac{\lambda}{2};\frac{5}{2};
			\frac{1 - \cosh(t)}{2}\right)\Bigg).
		\end{multline*}
	\end{enumerate}
\end{ex}

\subsection{Positive definite matrix-valued spherical functions}

By inspecting unitarizability of the principal series $U_{\mu,\lambda}$ and its subrepresentations/quotients, we can determine which spherical functions $\varphi_{\mu,\lambda}^\ell$ are positive definite.

\begin{prop}\label{prop:PosDefSphFctSL2C}
	Let $\ell\in\mathbb{Z}_{\geq0}$ and $\mu\in\{-\ell,-\ell+2,\ldots,\ell-2,\ell\}$, $\lambda\in\mathbb{C}$. The spherical function $\varphi_{\mu,\lambda}^\ell$ is positive definite if and only if one of the following statements holds:
	\begin{itemize}
		\item (unitary principal series) $\lambda\in i\mathbb{R}$,
		\item (complementary series) $\mu=0$ and $\lambda\in(-2,2)$,
		\item (trivial representation) $\mu=0$, $\lambda=\pm2$ and $\ell=0$,
		\item (relative discrete series) $\mu=0$, $\lambda\in\pm\{2,4,\ldots,\ell\}$ with $\ell$ even.
	\end{itemize}
\end{prop}

\begin{proof}
	$\varphi_{\mu,\lambda}^\ell$ is positive definite if and only if the $K$-type $S^\ell(\mathbb{C}^2)$ of $U_{\mu,\lambda}$ belongs to a unitarizable subrepresentation/quotient. The unitarizable subrepresentations/quotients of the principal series of $\mathrm{SL}_2(\mathbb{C})$ are determined in \cite[Ch.~XVI.\S1]{Kn86}.
\end{proof}

Note that the relative discrete series representation which occurs as unitarizable subrepresentation/quotient of the non-unitary principal series $U_{0,\lambda}$ for $\lambda\in\pm\{2,4,\ldots,\ell\}$, $\ell$ even, is isomorphic to the unitary principal series $U_{\lambda,0}$. This is reflected by the symmetry property (cf.~\cite[Prop.~6.9]{GPT02})
$$\varphi_{\mu,\lambda}^\ell=\varphi_{\lambda,\mu}^\ell \qquad \mbox{whenever }\lambda\in\{-\ell,-\ell+2,\ldots,\ell-2,\ell\}. $$

\section{Perspectives}
\mlabel{sec:pers}

\subsection{Covering groups acting on Sitter space}

Our discussion of locality in de Sitter space in Section~\ref{sec:6}
was based on the fact that the set of spacelike pairs in
$\dS^d \times \dS^d$ coincides with $G.(W \times -W)$ for any wedge region
$W \subeq \dS^d$ and $G = \SO_{1,d}(\R)_e$.
This situation changes if we pass to covering spaces.

As de Sitter space $\dS^d$ is homotopy equivalent to the sphere $\bS^{d-1}$,
it is simply connected for $d \geq 3$, so that non-trivial coverings
only exist for $d = 2$. In this case $\pi_1(\dS^2) \cong \Z$, so that
we have for each $n\in \{2,3,\ldots \} \cup \{ \infty\}$
an $n$-sheeted  covering
$M^{(n)}$. These are causal symmetric spaces on which
covering groups $G^{(n)}$ of $G = \SO_{1,2}(\R)_e$ act transitively.

A key difference between $\dS^2$ and its coverings is that in
the coverings the positivity domain $W_{M^{(n)}}^+(h)$ of the boost vector
field decomposes into several connected components
which are natural candidates for wedge regions~$W$. Likewise the ``causal complement''
$W'$ is no longer uniquely specified, and this creates intricate
locality properties that should be understood from the representation
theoretic perspective. For some first steps in this direction, we
refer to the discussion of twisted locality 
in \cite{MN21}.

For $d \geq 3$, the group $G$ has a simply connnected $2$-fold
covering $\tilde G = \Spin_{1,d}(\R)$, and we may also consider $\dS^d$ as a
causal symmetric space of this group. The passage to this covering
group leads to more homogeneous vector bundles, associated to
representations of $\Spin_{1,d-1}(\R)$, including in particular spinor bundles.

\subsection{Invariant Hilbert spaces in $C^{-\infty}(\dS^d)$}
\mlabel{subsec:pers.2}

In the present paper we use the realization of
de Sitter space $\dS^d$ as a submanifold of the boundary
of the two domains $\Xi_\pm \subeq \dS^d_\C$. From the construction
in \cite{FNO23}, we obtain realizations of all irreducible
unitary representations of the group $G = \SO_{1,d}(\R)_e$ in
holomorphic sections of vector bundles over $\Xi_\pm$ 
and we have seen how to derive by boundary value 
realizations in bundle-valued distributions on~$\dS^d$.

For the special case of spherical representations,
this construction leads to the reproducing kernel Hilbert spaces 
$\cH_m^\pm \subeq \cO(\Xi_\pm)$ and further to Hilbert subspaces
$\cH_m^\pm \subeq C^{-\infty}(\dS^d), m \geq 0,$ 
where $m = 0$ corresponds to the space of constant functions.
These embeddings are implemented by
$H$-fixed distribution vectors for $H = G_{\be_1} = \SO_{1,d-1}(\R)_e$.
Any such distribution vector defines an $H$-invariant
positive definite distribution $D_m^\pm \in C^{-\infty}(\dS^d)$.

In this context it makes sense to ask
to which extent are $H$-invariant positive definite distributions on
$\dS^d$ integrals of the distributions $D_m^\pm$? The existence of the
integral representations follows from \cite{FT99}, once the extremal
rays in the cone of positive  definite $H$-invariant distributions are
parametrized. In the QFT context this corresponds to
so-called K\"allen--Lehmann representations of
``$2$-point functions'', resp., to generalized free fields
(cf.~\cite[Prop.~3.3]{BV96}).
 The extent to which such explicit integral decompositions,
  resp., parametrizations of extremal rays can be carried
  out has to be clarified. Related results that can most likely
  be combined to solve this problem can be found
  in \cite{Fa79}, \cite{OS80},  \cite{Sh95} and \cite{vD09}.

\subsection{Aspects of interacting QFTs}

In \cite{BJM23} Barata, J\"akel and Mund construct a family
of quantum field theories in the sense of Haag--Kastler on
$2$-dimensional de Sitter space. It would be interesting to connected
the present project with these constructions.
For scalar theories, the Hilbert spaces consist of distributions on
$\dS^d$, hence are represented by $H$-invariant positive definite
distributions on $\dS^d$ (cf.\ Subsection~\ref{subsec:pers.2}). 

The construction of local nets of operator algebras on de Sitter space
in \cite{BJM23} is based on euclidean models satisfying
a reflection positivity condition (cf.\ \cite{NO18}).
Presently we do not know to which extent such euclidean models
exist for non-spherical irreducible unitary representations of
$\SO_{1,d}(\R)_e$, but the discussion of the general
connection between standard subspaces of
reflection positivity on the circle group
(acting by rotations on the complex de Sitter space $\dS^1_\C
\cong \C^\times$),  \cite{NO19} suggests non-trivial phenomena
also for $d \geq 2$.

\subsection{Realizing reducible representations by   boundary values} 

The construction of realizations of unitary representations
of $\SO_{1,d}(\R)_e$ in distributional sections of vector bundles
on de Sitter space starts, according to \cite{FNO23},
with $K$-finite vectors in unitary representations which
are finite sums of irreducible ones.
In view of the relevance of more complicated representations
in the context of interacting QFTs (cf.\ \cite{BJM23}),
it would be interesting to obtain similar results for
representations that are direct integrals or infinite sums of
irreducible representations. 
Following the strategy outlined in \cite{FNO23}
and the discussion of the case $d = 1$ in the present paper,
we would like to use $J$-fixed $K$-finite vectors $v \in
\cH^J_{\rm temp}$, i.e.,
for which there exists some $C> 0$ with 
\[ \|e^{i t \partial U(h)}v\| \leq \frac{C}
  {(\frac{\pi}{2}-t)^N} \quad \mbox{ for } \quad t \in [0,\frac{\pi}{2}).\]
A characterization of  such $K$-finite vectors can probably
be achieved in terms of the spectral decomposition 
with respect to the Casimir operator or more general invariant
differential operators on the corresponding $K$-isotypical subspaces,
which could be infinite-dimensional. 
This line of investigation requires in particular a
better quantitative understanding of decompositions of finite
tensor products of irreducible representations.
Tensor products of two irreducible unitary representations of
  $\SO_{1,d}(\R)_e$ decompose with finite multiplicities, 
  but for triple tensor products the multiplicities
  become  infinite (unless tensoring exclusively
  highest/lowest weight representations for $d = 2$).
  The first statement follows from the fact that the pair
  $(\SO_{1,d}(\R)_e^3, {\rm diag}(\SO_{1,d}(\R)_e))$ is
  real spherical  (cf.\ \cite[Thm.~1.3, Fact 1.8]{KM14}).
  Note that, up to coverings, $G = \SO_{1,d}(\R)_e$ is the only
  non-compact connected simple real Lie group for which
  $(G^3, {\rm diag}(G))$  has this property.
  The second one can be derived with the holomorphic
  character of the corresponding semigroup representation.

\subsection{Huygens' Principle}

In Theorem~\ref{Thm:HP}
we have seen that, for certain parameters,
the support of the jump distributions lies inside the
``lightlike'' part of the sets $I^\pm(x)$ and 
this has strong implications for the locality behavior of the
corresponding Klein--Gordon equation.

In the vector-valued case the generalized spherical function
$\phi_\cE$ also permits us to study jump distributions on
de Sitter space. In this case it would be interesting
to determine if these matrix-valued distributions also
exhibit some Huygens' Principle in the sense that their
support lies in~$\partial I^\pm(\be_1)$. 

It may also be possible to determine the jump singularities
  for the operator-valued spherical functions for $\SL_2(\C) \cong
  \Spin_{1,3}(\R)$ that we find in Section~\ref{sec:sphfunc}.
  
\appendix

\section{Hypergeometric functions}

Our standard references for this section are \cite[Ch.~II]{Er53}
and \cite[Ch.~9]{LS65}. 

The Gau\ss{} hypergeometric functions, or simply the
{\it hypergeometric function}, $\hgfabc $ is given by the series
\begin{equation}\label{def:hgf}
  \hgfabc =\sum_{k=0}^\infty \frac{ (\alpha)_k (\beta)_k}{(\gamma)_k}\,
  \frac{z^k}{k!},\quad \gamma \not= 0, -1, -2, \ldots 
\end{equation}
where 
\[(a)_k =a(a+1)\cdots (a+k-1), \quad k =0, 1,2, \ldots.\]
We note the following, see \cite[Ch. II]{Er53} and \cite[Ch. 9]{LS65}:
\begin{lemma} 
\begin{itemize} 
\item[\rm (a)] If $\alpha $ or $\beta$ is contained in $-\N_0$, 
  then the series in \eqref{def:hgf} is finite and $\hgf$ is a polynomial function.
\item[\rm (b)] The series in \eqref{def:hgf} converges absolutely for $|z|<1$.
\item[\rm (c)] The hypergeometric function $\hgfabc$ is the unique solution to the differential equation
\[z(1-z)u^{\prime\prime} +(\gamma - (\alpha + \beta +1)z)u^\prime - \alpha\beta u=0\]
which is regular at $z=0$ and takes the value $1$ at that point.
\end{itemize}
\end{lemma}

The following can be found in \cite[p.~241,\ 246,\ 245]{LS65}. For the second
theorem we use in our formulation that
$\arg (1-z)< \pi $ is equivalent to $z\in \C\setminus [1,\infty)$. 

\begin{theorem} The hypergeometric function $\hgfabc$
  extends holomorphically to the slit plane $\C \setminus [1,\infty)$,
  and for fixed $z\in \C\setminus [1,\infty)$, the
  function
\[ (\alpha,\beta,\gamma) \mapsto \frac{\hgfabc}{\Gamma (\gamma)}  \]
is an entire function of $\alpha,\beta$ and $\gamma$.
\end{theorem}

In particular, the functions 
$\gamma \mapsto \hgfabc$ are meromorphic with simple poles at
$\gamma \in - \N_0$. 
Some more facts that we will use are collected in the following  theorem.

\begin{thm} \label{thm:gamma-limit} The following assertions hold:
  \begin{itemize}
  \item[\rm(i)] 
    If $\Re (\gamma - \alpha - \beta) >0$
   and $\gamma \not \in -\N_0$, then
   \[ \lim_{z\to 1} {}_2F_1(\alpha,\beta;\gamma;z)
          = \frac{\Gamma (\gamma)
            \Gamma (\gamma - \alpha-\beta )}{\Gamma (\gamma - \alpha)\Gamma(\gamma -\beta)},\]
        provided there exists an $0<\eps <\pi$ with
        $|\arg (1-z)| <\pi - \eps$.   
\item[\rm(ii)] 
  For $z \in \C \setminus [1,\infty)$, we have
\[\hgfabc= (1-z)^{\gamma - \alpha -\beta} \hgf (\gamma - \alpha ,\gamma - \beta; \gamma ; z).\]
\item[\rm(iii)] 
  If $\gamma =\alpha +\beta
 \not\in -\N_0$, then 
 \[\lim_{t\to 1-}\frac{{}_2F_1(\alpha ,\beta ;\gamma ;t)}
{-\log (1-t)}    
   = \frac{\Gamma (\alpha + \beta)}{\Gamma (\alpha ) \Gamma (\beta)} 
   = \frac{\Gamma (\gamma)}{\Gamma (\alpha ) \Gamma (\beta)}.\]
\item[\rm (iv)] If $\gamma -\alpha -\beta \not\in\mathbb{Z}$ and 
\[ C(\lambda ) =
  \frac{\Gamma (\gamma) \Gamma (\alpha + \beta -\gamma)}
  {\Gamma (\alpha )\Gamma (\beta )} \quad \mbox{ and } \quad
  D(\lambda) = \frac{\Gamma (\gamma) \Gamma (\gamma-\alpha -\beta )}
  {\Gamma (\gamma - \alpha )\Gamma (\gamma - \beta )},\]
then
\begin{align*}
	{}_2F_1(\alpha ,\beta ;\gamma ;z) & = 
	(1-z)^{\gamma -\alpha -\beta }\, C (\lambda) \cdot 
	{}_2F_1(\gamma -\alpha ,\gamma -\beta ;\gamma -\alpha -\beta +1;1-z)\\
	& \quad +D(\lambda ) \cdot {}_2F_1(\alpha ,\beta ;\alpha +\beta -\gamma +1;1-z).
\end{align*}
  \end{itemize}
\end{thm}

\begin{prf} (i) is \cite[p.~244 (9.3.4)]{LS65}  which relies on
  the general integral representation \cite[(9.1.6)]{LS65}. 

  \nin (ii) is  \cite[p.~248]{LS65}. Here
  we use that $\arg (1-z) < \pi$
  if and only if $z\in \C\setminus [1,\infty)$
  (see also \cite[(15.4.20,15.4.22]{DLMF}).

  \nin (iii)  We use \cite[(15.4.21)]{DLMF} with $\gamma =\alpha + \beta$.

\nin (iv) This is \cite[Cor.~2.3.3]{AAR99}.
\end{prf}

If $\Re (\gamma - \alpha -\beta)<0$, then we can
use the hypergeometric identity (ii) to evaluate limits as in (i) because 
$\gamma - (\gamma - \alpha) - (\gamma - \beta) = - \gamma +\alpha + \beta$ 
has positive real part.

\begin{lemma} \mlabel{lem:jump-hypergeo}
  We have the following:
\begin{itemize}
\item[\rm (i)]  If $n \in \Z$ and
  $\alpha +\beta \pm n\not\in -\N_0$, then with $\alpha^\prime = \alpha -n$ and $\beta^\prime = \beta -n$ we have 
\[{}_2F_1 (\alpha ,\beta ; \alpha + \beta -n ; z)= (1-z)^{-n } {}_2F_1(\alpha^\prime , \beta^\prime   ;\alpha^\prime + \beta^\prime + n;z).\] 
\item[\rm (ii)] Let $n\in \N$, $\psi(z) = \Gamma'(z)/\Gamma(z)$,
  $a(n)  = \frac{\Gamma(\alpha' + \beta' +n)}{\Gamma(\alpha' +n)
\Gamma(\beta'+n)}$
and 
\[ F(n; z ) = {}_2F_1(\alpha' ,\beta' ;\alpha' + \beta' +n;z).\]
Then we have for $|z-1| < 1$ and $z\not\in [1,2)$ the identity
\begin{align*}
F(n; z) &= a (n)\Big(\sum_{k=0}^{n-1} \frac{(-1)^k(n-k-1)! (\alpha')_k(\beta')_k}{k!} (1-z)^k\\
\quad & + (z-1)^n\sum_{k=0}^\infty \frac{(\alpha' +n)_k(\beta'+n)_k}{(n+k)!k! }
          \Big[\psi (k+1) + \psi (n+k+1)\\
&\qquad\qquad\qquad\qquad\qquad - \psi(\alpha' +n +k ) - \psi(\beta' +n+k) -\log (1-z)\Big](1-z)^k\Big)
\end{align*} 
    \end{itemize}
 \end{lemma} 
 \begin{proof} (i) follows from Theorem \ref{thm:gamma-limit}(ii)
   by taking $\gamma = \alpha  + \beta - n$.

   \nin (ii) is \cite[p. 315]{LS65}.
 \end{proof}

\section{The complex $\cosh$-function}
\mlabel{sec:bap}

\begin{rem} \mlabel{rem:cosh} (The complex $\cosh$ function)
  We consider the holomorphic function
  \[ \gamma \: \C^\times \to \C,\quad \gamma(z) = \frac{1}{2}(z + z^{-1}) \]
  and note that $\cosh = \gamma \circ \exp \: \C \to \C.$ 
  This function is surjective and satisfies
  \begin{equation}
    \label{eq:cosh1}
 \cosh(z + \pi i) = -\cosh(z) \quad \mbox{ and } \quad
 \cosh\Big(z + \frac{\pi i}{2}\Big) = i \sinh(z).
\end{equation}
The criticial points of $\cosh$ are the zeros of $\sinh$, i.e.,
$\pi i \Z$, and in these points we have
$\cosh( \pi i k) = (-1)^k$. Hence the critical values of $\cosh$
are $\pm 1$. We also conclude from \eqref{eq:cosh1} that
$\cosh(\cS_\pi)$ is invariant under multiplication with $-1$, so that
\[  \cosh(\oline{\cS_\pi}) = \C.\]
On the boundary we have
\begin{equation}
  \label{eq:cosh2}
 \cosh(\R) = [1,\infty) \quad \mbox{ and } \quad 
 \cosh(\R + \pi i) = (-\infty,-1],
\end{equation}
and
\[ \cosh  \: \cS_\pi \to \C \setminus ((-\infty,-1] \cup [1,\infty)) \]
is biholomorphic.

From
\[ \cosh(x + i y) = \cosh(x) \cos(y) + i \sinh(x) \sin(y) \]
we obtain 
\[ \Re\cosh(x + i y) = \cosh(x) \cos(y)\quad \mbox{ and }  \quad
  \Im \cos(x + i y) = \sinh(x) \sin(y). \]
\end{rem}

\end{document}